%% file: main.tex
\documentclass[mnsc,sglanonrev]{informs4}

\input{preamble.tex}
\input{tikz/setting.tex}

\OneAndAHalfSpacedXI

\EquationsNumberedThrough    

\TheoremsNumberedThrough     
\ECRepeatTheorems  %

\begin{document}




\TITLE{Large-Scale LLM Inference with Heterogeneous Workloads: Prefill-Decode Contention and Asymptotically Optimal Control}

\ARTICLEAUTHORS{%
  \AUTHOR{
    Ruihan Lin\textsuperscript{a},
    Zean Han\textsuperscript{b},
    Zezhen Ding\textsuperscript{c},
    Jiheng Zhang\textsuperscript{d}
  }
  \AFF{
    Department of Industrial Engineering and Decision Analytics,
    The Hong Kong University of Science and Technology\\
    \EMAIL{
      \textsuperscript{a}rlinah@connect.ust.hk, 
      \textsuperscript{b}zhanax@connect.ust.hk, 
      \textsuperscript{c}zdingah@connect.ust.hk, 
      \textsuperscript{d}jiheng@ust.hk
    }
  }
}

\ABSTRACT{

Large Language Models (LLMs) are rapidly becoming critical infrastructure for enterprise applications, driving unprecedented demand for GPU-based inference services.
A key operational challenge arises from the two-phase nature of LLM inference: a compute-intensive \emph{prefill} phase that processes user input, followed by a memory-bound \emph{decode} phase that generates output tokens.
When these phases share GPU resources, prefill tasks throttle the processing speed of concurrent decodes, creating state-dependent contention.
This contention is further complicated by workload heterogeneity, as different applications exhibit vastly different input and output lengths.
We develop a stochastic control framework for scheduling heterogeneous LLM workloads across large GPU clusters.
We formulate LLM inference as a multiclass many-server queueing network with state-dependent service rates, grounded in empirical iteration-time measurements.
We analyze the fluid approximation of this system and solve steady-state linear programs that characterize optimal resource allocation.
We design gate-and-route policies that regulate prefill admission and decode routing, and prove that they are asymptotically optimal in the many-GPU limit under both bundled and separate token-pricing schemes.
We further extend the framework to incorporate Service Level Indicators (SLIs) such as latency and fairness, providing a general approach to constrained scheduling.
\textcolor{tc-ins}{Trace-driven replay simulations on empirical Azure traces, together with controlled synthetic experiments, show that our policies outperform representative serving heuristics and expose interpretable service--revenue tradeoffs.}

}%




\maketitle


\input{body.tex}

\bibliography{references}

\ECSwitch

\input{appendix}

\end{document}

%% file: preamble.tex
\usepackage{mathtools}          

\usepackage{microtype}          

\usepackage{booktabs}           
\usepackage{multirow}           
\usepackage{tabularx}           

\usepackage{caption}
\usepackage{subcaption}         

\usepackage{algorithm}
\usepackage{algpseudocode}      

\usepackage{enumitem}           



\DeclareFontShape{OML}{cmm}{b}{it}{<->cmmib10}{}
\DeclareFontShape{OMS}{cmsy}{b}{n}{<->cmsy10}{}

\newcommand{\To}{\Rightarrow}                   

\usepackage[normalem]{ulem}
\definecolor{tc-del}{RGB}{0, 0, 0}    
\definecolor{tc-ins}{RGB}{0, 0, 0}    
\definecolor{tc-cmt}{RGB}{0, 0, 0}    


%% file: tikz/setting.tex
\usepackage{tikz}
\usetikzlibrary{shapes,shapes.arrows,arrows,arrows.meta,positioning,calc,patterns,decorations.pathreplacing,topaths,automata}
\usepackage{pgfplots}
\usepgfplotslibrary{groupplots,dateplot}
\usepackage{pgfplotstable}
\pgfplotsset{compat=1.18}

\tikzset{
  cir/.style={
    ellipse,
    minimum width=15pt,
    minimum height=23pt,
    align=center,
    text width=20pt,
    inner sep=2pt,
    draw=black,
  }
}

\tikzset{
  pcir/.style={
    circle,
    minimum width=15pt,
    minimum height=15pt,
    align=center,
    text width=20pt,
    inner sep=2pt,
    draw=blue,
  }
}

\tikzset{
  box/.style={
    rectangle,
    rounded corners=3pt,
    minimum width=25pt,
    minimum height=20pt,
    inner sep=2pt,
    draw=black,
  }
}

%% file: body.tex
\setlength{\abovedisplayskip}{4pt plus 2pt minus 2pt}
\setlength{\belowdisplayskip}{4pt plus 2pt minus 2pt}
\setlength{\abovedisplayshortskip}{1pt plus 1pt minus 1pt}
\setlength{\belowdisplayshortskip}{2pt plus 1pt minus 1pt}

\section{Introduction}
Large Language Models (LLMs) have emerged as a foundational technology in contemporary artificial intelligence, leading to a substantial increase in computational demand~\citep{zhao2025surveylargelanguagemodels}.
To accommodate this growth, production-scale infrastructures have expanded correspondingly, frequently requiring the concurrent utilization of thousands of GPUs to sustain worldwide inference workloads~\citep{Kwon2023vLLM,aminabadi2022deepspeed}.
Since commercial LLM services predominantly adopt per-token pricing~\citep{openaiPricing2025,vertexPricing2025,anthropic_billing_2025,gemini_billing_2025,deepseek_pricing_2025}, revenue is closely tied to the volume and composition of tokens processed, making efficient resource allocation critical for both profitability and user experience.

\paragraph{\textbf{Prefill--decode contention.}}
A distinguishing feature of LLM inference is that each request proceeds in two stages: a \emph{\textbf{prefill}} stage, in which the model processes the user's input prompt, followed by a \emph{\textbf{decode}} stage, in which it generates output tokens autoregressively.
Modern serving systems batch multiple requests on the same GPU to exploit the GPU's parallel processing capability, but these two stages interact in nontrivial ways when sharing GPU resources.
Prefill is compute-intensive and, when present in a batch, dominates the iteration time and thereby throttles the processing speed of co-located decode tasks; decode, in contrast, is memory-bound and proceeds faster when running alone~\citep{Kwon2023vLLM,agrawal2024sarathiserve}.
Moreover, iteration time grows roughly linearly in the amount of prefill work, so adding a second prefill to a batch would nearly double the iteration time without improving parallelism.
For this reason, practical systems process at most one prefill per GPU at a time.
This \emph{prefill-decode contention} creates a fundamental scheduling tension: admitting more prefills increases the rate at which new requests enter the system but slows down all concurrent decodes, and the scheduler must carefully balance how many GPUs are devoted to prefill versus decode-only operation.

\paragraph{\textbf{Workload heterogeneity.}}
Real-world LLM services do not see a single homogeneous stream of requests.
They serve a mix of applications, such as summarization, creative writing, and question answering, which differ widely in typical input and output lengths~\citep{Sun2024Llumnix,Zheng2024LmsysChat1M,Zhao2024WildChat}.
For example, summarization tasks average over 1,000 input tokens with moderate output, while creative writing requires fewer than 100 input tokens yet generates over 900 output tokens on average (see Table~\ref{tab:workload_heterogeneity} in the e-companion for detailed statistics).
This heterogeneity amplifies the scheduling challenge: a class with long prefills and short decodes (e.g., summarization) consumes GPU time during prefill but releases decode capacity quickly, whereas a class with short prefills and long decodes (e.g., creative writing) admits quickly but occupies decode slots for extended periods.
Maximizing revenue requires serving an appropriate mix of both to balance pipeline utilization.
Consequently, this multiclass resource allocation problem is not well served by simple static priority rules such as first-come-first-served or shortest-job-first; instead, it calls for finer control over how resources are distributed across classes.

This paper addresses the following question: \emph{\textbf{How should a large-scale LLM inference system jointly control admission and scheduling across multiple request classes to maximize token-based revenue while respecting service level indicators (SLIs)?}}
Answering this question requires overcoming several challenges.
First, the state-dependent service rates arising from prefill-decode contention create analytical difficulties that preclude direct application of standard queueing results.
Second, heterogeneous workloads induce a multiclass resource allocation problem where the optimal policy depends on the composition of the workload mix.
Third, practical systems must balance revenue maximization against SLIs such as latency and fairness across request classes.

\paragraph{\textbf{Our approach.}}
We model the system as a multiclass many-server queueing network where each GPU operates in one of two modes: \emph{mixed} (running one prefill alongside decodes) or \emph{solo} (decode-only).
Service rates are state-dependent, specifically, decodes run slower in mixed mode due to intensive computation of prefill, and we derive these rates from empirical iteration-time measurements in Section~\ref{sec:speed_abstraction}, capturing essential GPU physics in a tractable analytical framework.

Because production LLM clusters typically comprise hundreds to thousands of GPUs, we study the system through fluid approximation in the many-server regime, which is a standard and well-established approach in operations research for analyzing large-scale stochastic networks.
In this scaling, stochastic fluctuations average out and the system trajectory converges to a deterministic limit, yielding both analytical tractability and high accuracy at scale.
The fluid model reduces to a steady-state linear program (LP) whose solution prescribes how to partition cluster capacity between mixed and solo modes and how to distribute workload across request classes.

We translate the fluid solution into implementable control via a \emph{gate-and-route} architecture: a \emph{prefill gate} regulates admission by tracking class-level occupancy targets from the LP, and a \emph{decode router} directs completed prefills to available GPU slots.
This decomposition into static planning (solving the LP) and dynamic control (enforcing LP targets) is central to our design and underpins the asymptotic optimality results we establish.

\paragraph{\textbf{Our contributions.}}
We make the following contributions.
\begin{enumerate}
    \item \emph{Multiclass many-server model with prefill-decode contention.}
    We develop a queueing network where each GPU operates in mixed or solo mode, with state-dependent service rates that capture how prefill operations throttle co-located decodes.
    The service-rate parameters are calibrated from controlled experiments on production hardware (A100 GPUs) to closely reflect real system behavior.

    \item \emph{Fluid approximation and LP-based planning.}
    We establish convergence of the scaled stochastic system to a deterministic fluid limit in the many-server regime.
    The steady-state analysis reduces to a linear program that prescribes optimal capacity partitioning between mixed and solo modes and class-level occupancy targets.
    This LP formulation provides analytical tractability and serves as the foundation for the control policies we develop.

    \item \textcolor{tc-ins}{\emph{Asymptotically optimal control policies and extensions.}
    We translate the LP solution into an implementable gate-and-route policy and prove its asymptotic optimality under bundled charging, where revenue is credited upon completion.
    Our policy operates at the class level: the LP yields class-specific admission targets that jointly regulate the mixed/solo split and the admission mix, going beyond aggregate prefill/decode split rules.
    We then show that the same planning architecture extends to separate charging, where phase-based revenue recognition changes admission incentives and leads to a priority-based prefill gate, and to service-level indicators (fairness, latency) modeled as constraints or penalties.}

    \item \textcolor{tc-ins}{\emph{Adaptive online control and empirical evaluation.}
    We develop an adaptive online version of the policy that estimates class-level arrival rates from a rolling window and periodically replans the fluid LP, without assuming a known traffic process.
    Trace-driven replay experiments on empirical Azure traces show that it attains higher revenue than representative serving heuristics, while controlled synthetic experiments confirm that per-GPU revenue converges to the fluid optimum as the cluster scales and quantify the revenue cost of imposing fairness and latency SLIs.}
\end{enumerate}

\subsection{Literature Review}
Efficient LLM inference serving has become a critical systems challenge as production deployments face mounting demands for throughput, latency, and resource efficiency.
Early systems work concentrated on improving performance on single GPUs through architectural innovations.
Iteration-level batching made it practical to schedule at token granularity \citep{yu2022orca}, while paged attention reduced KV-cache fragmentation and enabled high utilization under dynamic workloads \citep{Kwon2023vLLM}.
Chunked prefill further enabled interleaving prefill chunks with ongoing decodes \citep{agrawal2024sarathiserve}.

As deployments scaled from single-device prototypes to large-scale production clusters, designers began allowing the prefill and decode stages to be executed on different GPUs to better match the compute-bound prefill phase with the memory-bandwidth-bound decode phase \citep{Patel2024Splitwise,zhong2024distserve}.
\textcolor{tc-ins}{Recent systems analysis further emphasizes that whether and how to disaggregate inference depends pragmatically on model size, traffic mix, and the balance between prefill and decode rates~\citep{mitra2025beyond}.}
This transition from single-GPU optimization to cluster-scale orchestration raises new questions that go beyond engineering heuristics: how should a service provider allocate GPU capacity and route requests across GPUs under heterogeneous workloads to maximize long-run revenue while considering customized service-level indicators?
Existing serving stacks acknowledge these trade-offs but do not yet provide a formal framework for revenue-driven resource allocation at scale.

This operational challenge has attracted growing attention from the operations research (OR) and operations management (OM) communities, as reflected in several recent surveys at the AI/LLM and OR/OM interface~\citep{dai2025ai,zhao2025surveylargelanguagemodels,zhou2024surveyefficientinferencelarge,li2024llminferenceservingsurvey,MitzenmacherShahout2025}.
Work here moves in two directions.
One uses LLMs to enhance OR/OM workflows, such as democratizing optimization~\citep{simchi2025democratizing}, automated optimization modeling~\citep{huang2025orlm}, and supply chain decision-making~\citep{simchilevi2025largelanguagemodelssupply}.
The other, which we pursue, applies OR/OM methodology to LLM inference itself, spanning output-quality improvements~\citep{ai2025majorityvotingllmaggregation,hu2025pretrainedaimodelassisted} and the acceleration of inference through principled resource allocation and scheduling that is the focus of this paper.

In the context of accelerating LLM inference, a recent strand of work employs competitive analysis from the online algorithms community to formalize LLM serving.
Zhou and coauthors model KV-cache-constrained batching and benchmark online schedulers against a hindsight integer program.
They show that under fully adversarial arrivals, no deterministic online algorithm achieves a constant competitive ratio; under additional assumptions, they provide a polynomial-time scheduler with a constant competitive ratio \citep{Jaillet2025KV}.
Follow-up work obtains constant-competitive policies when heterogeneous prefill/decode lengths are modeled directly \citep{WangYeZhou2025} and logarithmic-competitive guarantees under interval predictions of decode length via an adaptive policy (A\_min) contrasted with a conservative upper-bound policy (A\_max) \citep{ChenYeZhou2025}.

Meanwhile, the advancement of queueing theory provides an analytical framework for stochastic systems and control design.
In many-server settings, fluid approximations provide tractable first-order descriptions for capacity planning and performance analysis \citep{Whitt2006Fluid,zhang2013fluid}, including accuracy guarantees for sizing under impatience \citep{bassamboo2010accuracy} and asymptotically optimal scheduling structures for multiclass systems with abandonment \citep{atar2010cmu,long2020dynamic,long2024generalized}.
Subsequent work enriches these models by exploiting within-queue heterogeneity, allowing dependence between service requirements and patience \citep{bassamboo2016scheduling,wu2019service}, and analyzing state-dependent service rates and slowdowns \citep{DongFeldmanYomTov2015Slowdowns}.
In parallel, delay estimation and information-sharing frameworks have been developed to manage latency considerations in complex service systems \citep{IbrahimWhitt2009DelayEst,Ibrahim2018DelayInfoSurvey}.
Finally, the literature addresses methodological concerns such as robustness to model/input uncertainty in simulation \citep{Lam2016RobustSensitivity,GhoshLam2019RobustSimulation}, alongside control-oriented heavy-traffic perspectives on dynamic admission/sequencing \citep{HarrisonZeevi2004HW,Ata2006Thin} and state-space-collapse analyses in parallel-server networks \citep{DaiTezcan2011SSC}.

{\color{tc-ins}
Applying these queueing tools to LLM inference, recent theoretical work has begun to model LLM serving through stochastic and fluid approximations.
\citet{Ao2025WAIT} study KV-memory growth and batch-time linearity on a single GPU and design threshold policies that approach fluid-optimal throughput while controlling latency.
Complementary results establish throughput optimality of work-conserving rules under simplified token-time abstractions and motivate piecewise-linear iteration-time models~\citep{li2025throughput}.
These studies provide useful foundations for single-GPU scheduling and throughput analysis.

Our work instead studies cluster-scale serving with many GPUs, heterogeneous request classes, and token-based revenue objectives.
The many-server scaling increases the number of GPUs while keeping per-device service primitives fixed, and the planning problem must decide both the mixed/solo capacity split and the class-aware admission mix.
We therefore build a multiclass fluid approximation using service rates calibrated from empirical iteration-time measurements, solve a steady-state LP for occupancy targets, and translate those targets into a gate-and-route policy with asymptotic optimality guarantees.
Our contribution is novel along two complementary dimensions.
Relative to the systems literature on LLM serving, which largely treats prefill--decode scheduling through engineering heuristics or fixed capacity splits, we provide a principled cluster-scale control framework that ties admission and routing to an explicit revenue and SLI objective and comes with provable optimality guarantees in the many-GPU limit.
Relative to the queueing literature on multiclass many-server systems, our model goes beyond the classical fixed-pool service network: prefill admission itself reshapes the mixed/solo GPU composition, and hence the decode capacity available to each class, so that admission and routing are coupled through the server configuration in a way the standard many-server template does not capture and that calls for new structural and convergence arguments.}

\subsection{Organization}

The remainder of the paper is organized as follows.
Section~\ref{sec:problem} introduces the iteration-time abstraction, the multiclass many-server stochastic network with mixed and solo decode modes, and the bundled revenue objective.
Section~\ref{sec:fluid} develops the many-GPU fluid limit, formulates the steady-state linear program (LP) over per-GPU occupancies, and establishes structural properties such as decode-buffer elimination.
\textcolor{tc-ins}{Section~\ref{sec:policy} constructs the occupancy-anchored gate-and-route policy and proves its asymptotic optimality under bundled pricing.
Section~\ref{sec:slo} presents two extensions: a separate prefill/decode charging scheme with its prioritize-and-route policy, and an SLI-aware formulation that integrates fairness and latency proxies.
Section~\ref{sec:numerical} presents trace-driven replay experiments on empirical Azure traces together with synthetic multiclass workloads, while additional EC experiments validate fluid convergence and component mechanisms.}
Section~\ref{sec:conclusion} concludes and discusses directions for extending the framework.
Technical proofs and additional lemmas are collected in the electronic companion.

\section{Problem Formulation}
\label{sec:problem}
This section develops the queueing model in three parts.
Section~\ref{subsec:preliminaries} provides background on modern LLM inference systems.
Section~\ref{sec:speed_abstraction} characterizes GPU iteration time.
Section~\ref{subsec:stochastic-model} derives the resulting service rates and embeds them into a many-server stochastic network, specifying the state, flows, admissible controls, capacity coupling, and revenue objectives.

\subsection{Preliminaries: Modern LLM Inference Systems}
\label{subsec:preliminaries}

\paragraph{\textbf{Two stages of LLM inference.}}
We formally introduce the two stages in LLM inference, namely \emph{\textbf{prefill}} and \emph{\textbf{decode}}.
In the \emph{\textbf{prefill}} stage, the model reads the entire input prompt once and, layer by layer, builds an internal representation of all input tokens.
Implementations store per-token intermediate states in a key--value (KV) cache so that, in the subsequent \emph{\textbf{decode}} stage, the model can generate output tokens one by one by attending to this cache instead of recomputing all past KV values \citep{shazeer2019fast,dai2019transformer,Kwon2023vLLM}.
Empirical studies show that these two stages stress the hardware differently: prefill behaves like a large, single-shot matrix computation, while decode repeatedly accesses and extends the KV cache and is more sensitive to memory traffic \citep{Kwon2023vLLM,agrawal2024sarathiserve}.
This asymmetry is the root cause of the scheduling trade-offs that we model.

\paragraph{\textbf{Continuous batching.}}
State-of-the-art serving systems do not run one request at a time per GPU.
Instead, they use \emph{continuous} or \emph{iteration-level} batching: in each short iteration, the GPU advances many active requests by one output token in parallel, then repeats this process for the next token \citep{yu2022orca,Kwon2023vLLM}.
In practice, these systems impose a fixed upper bound on how many requests can be batched on a GPU, chosen for engineering reasons such as avoiding out-of-memory errors and reducing dynamic memory management overhead \citep{Kwon2023vLLM,agrawal2024sarathiserve}.
We denote this constant by $B$.

\paragraph{\textbf{Chunked prefill and GPU modes.}}
Long prompts make it inefficient to run prefill in isolation.
Recent systems therefore \emph{chunk} prefills into smaller pieces (with each $C$ tokens) that can be interleaved with decode work on the same GPU \citep{agrawal2024sarathiserve}.
Measurements show that once a prefill chunk is present in an iteration, it tends to dominate the iteration time; running additional prefills in parallel on the same GPU brings no extra benefit because compute capacity is already saturated \citep{agrawal2024sarathiserve}.
Following this evidence, we adopt the standard assumption that each GPU runs at most one prefill chunk at a time.

\begin{figure}[t]
    \centering
    \includegraphics[width=0.98\linewidth]{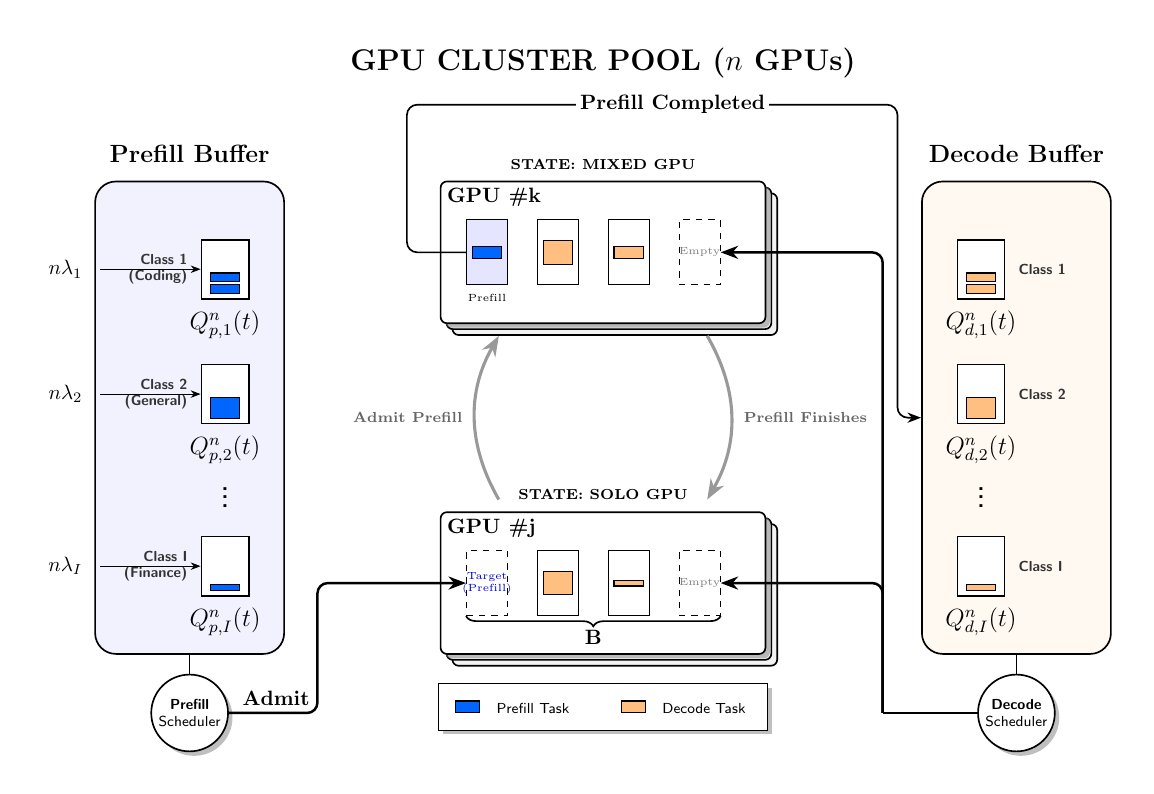}
    \caption{Schematic of the Dynamic GPU Scheduling Architecture. The system manages a cluster of $n$ GPUs with batch size $B$. GPUs transition between the Solo State (decode-only) and Mixed State (one prefill + decodes) based on assignments from the Prefill Scheduler. Completed prefills enter a virtual Decode Buffer, from which the Decode Scheduler populates available slots in either state.}
    \label{fig:llm_inference}
\end{figure}

Figure~\ref{fig:llm_inference} summarizes the resulting GPU-level architecture, which resembles key elements of systems such as vLLM and Sarathi-Serve~\citep{Kwon2023vLLM,agrawal2024sarathiserve}.
New requests from all classes first enter \emph{prefill queues}.
A host-side \emph{Prefill Scheduler} selects some queued requests and starts their prefill stage on GPUs that are not currently running another prefill.
While a prefill is in progress, that GPU is in a \emph{mixed} mode: one slot is occupied by the prefill chunk, and the remaining at most \(B-1\) slots can be used to advance decodes from decode-ready requests of various classes.
When the prefill finishes, the request moves into the \emph{decode buffer}, representing the set of requests ready to decode, with their KV caches typically kept resident in GPU memory.
The GPU then returns to a \emph{solo} decode mode, where all activated slots are devoted to decode.
A separate \emph{Decode Scheduler} continuously fills empty slots on both mixed and solo GPUs from the multiclass decode buffer.
Because a mixed-mode GPU shares its compute and memory bandwidth between a large prefill and several decodes, the per-token progress of those decodes is slower than on a solo GPU; admitting more prefills thus increases the rate at which new requests enter decode but slows down decodes sharing the same GPU.

\subsection{Token Processing on GPUs}
\label{sec:speed_abstraction}

To build a tractable stochastic model, we first characterize GPU iteration time.
In each iteration, the GPU processes a batch that may contain both prefill and decode tasks: all decode tasks in the batch each generate one output token, and if a prefill chunk is present, the GPU processes one chunk of its input tokens.
The iteration time is defined as the duration required to complete this batch processing step.

\paragraph{\textbf{Iteration time.}} A key empirical finding, documented by \citet{li2025throughput}, is that iteration time depends on the total number of tokens processed in a batch:
\begin{equation}
    \tau(b') = c + a \cdot \max\{0, b' - b_0\}.
    \label{eq:iteration_time_general}
\end{equation}
Here, $b'$ denotes the effective token count for the iteration, equal to the prefill chunk size plus the number of concurrent decode tasks.
The constant $c > 0$ captures fixed overheads (e.g., kernel launches), $a > 0$ is the marginal cost per token, and $b_0 \ge 0$ is a threshold below which overheads dominate.
{\color{tc-ins}
This two-regime form summarizes the operating ranges relevant for scheduling.
Decode-only iterations are dominated by fixed per-iteration overhead, while mixed iterations with a sizeable prefill chunk are governed by the marginal cost of processing additional prefill tokens.}
This formula captures two operating regimes:
\begin{itemize}
    \item \textbf{Decode-only iteration:} In a decode-only batch, the effective token count $b'$ equals the batch size, which is typically smaller than $b_0$.
    The $\max$ term in Equation~\eqref{eq:iteration_time_general} then vanishes, and the iteration time reduces to a near-constant value
    \begin{equation}
        \tau_{\text{solo}} = \tau(b') = c, \quad \text{for } b' \le b_0.
    \end{equation}

    \item \textbf{Mixed-batch iteration:} When a batch includes a prefill chunk of size $C$, the chunk dominates the effective token count (i.e., $b' \approx C$).
    For practical chunk sizes where $C > b_0$, iteration time grows linearly in $C$.
    This is consistent with empirical observations from Sarathi-Serve, where decode iteration time is largely independent of batch size but prefill iteration time scales linearly with chunk size~\citep{agrawal2024sarathiserve}.
\end{itemize}
Our experiments in Section~\ref{subsec:calibration} confirm both regimes: Figure~\ref{fig:prefill_calibration} reports the mixed-iteration calibration, where iteration time grows linearly in the prefill chunk size $C$, and the solo-decode calibration, where the KV-cache slope is small relative to the fixed iteration overhead.

Since practical chunk sizes typically satisfy $C \ge b_0$, we adopt the linear form for mixed-batch iteration time:
\begin{equation}
    \tau_{\text{mix}}(C) = \alpha + \beta C, \quad \text{where } \alpha := c - a b_0 \text{ and } \beta := a > 0.
    \label{eq:iteration_time_mixed}
\end{equation}
This two-regime abstraction captures the key prefill–decode interaction while remaining analytically tractable.

\subsection{The Stochastic Model}
\label{subsec:stochastic-model}

We now embed the iteration-time characterization from Section~\ref{sec:speed_abstraction} into a multiclass many-server stochastic network.
The model is indexed by the number of GPUs $n\in\mathbb{N}$.

\paragraph{\textbf{Primitives and service parameters.}}

Requests belong to a finite set of classes $\mathcal{I} := \{1,\dots,I\}$.
\textcolor{tc-ins}{A consolidated notation summary is provided in EC Table~\ref{tab:notation-summary}.}
A class-$i$ request is characterized by its representative prompt length $P_i$ and decode length $D_i$ (in tokens).
The system consists of $n$ homogeneous GPUs.
Each GPU can host at most $B\in\mathbb{N}$ concurrent decode streams and at most one prefill at a time.
Prefill is executed in fixed-size chunks of $C>0$ tokens per iteration.

Recall from Section~\ref{sec:speed_abstraction} that the iteration time is $\tau_{\text{solo}} = c$ in decode-only mode and $\tau_{\text{mix}}(C) = \alpha + \beta C$ when a prefill chunk is present.
For notational convenience, write $\tau := \tau_{\text{mix}}(C)$ for the mixed-iteration time.
Service rates are derived as follows.

\begin{itemize}
    \item \textbf{Prefill rate.} A prefill job of length $P_i$ tokens advances $C$ tokens per iteration, each taking time $\tau$.
    Completing the prefill requires $(P_i/C)$ iterations, so the mean service time is $(P_i/C)\tau$ and the rate is
    \[
    \mu_{p,i} = \frac{C}{P_i\,\tau}.
    \]

    \item \textbf{Mixed decode rate.} In mixed mode, a decode job produces one token per iteration.
    A class-$i$ job needs $D_i$ tokens, so the mean service time is $D_i\tau$ and the rate is
    \[
    \mu_{m,i} = \frac{1}{D_i\,\tau}.
    \]

    \item \textbf{Solo decode rate.} In decode-only mode, each token takes $\tau_{\text{solo}}$ seconds.
    Defining $\gamma := 1/\tau_{\text{solo}}$ to be the token generation rate per slot, the mean service time for $D_i$ tokens is $D_i/\gamma$ and the rate is
    \[
    \mu_{s,i} = \frac{\gamma}{D_i}.
    \]
\end{itemize}
We collect these rates as
\begin{equation}
    \mu_{p,i} = \frac{C}{P_i\,\tau},
    \qquad
    \mu_{m,i} = \frac{1}{D_i\,\tau},
    \qquad
    \mu_{s,i} = \frac{\gamma}{D_i}.
    \label{eq:mu-def}
\end{equation}

For analytical tractability, we model prefill, mixed decode, and solo decode service times as independent exponential random variables with rates $\mu_{p,i}$, $\mu_{m,i}$, and $\mu_{s,i}$, respectively, for each class $i\in\mathcal{I}$.

We assume Poisson arrivals: in the $n$th system, class-$i$ arrivals form a Poisson process with rate $\lambda_i^n := n\lambda_i$, where $\lambda_i>0$ is the nominal arrival rate per GPU, so the total offered load grows proportionally with $n$.
Customers are impatient in both the prefill and decode queues: any class-$i$ job that is waiting is endowed with an independent exponential patience time with rate $\theta_i\ge 0$.
Interarrival times, service times, and patience times are assumed mutually independent across all jobs and classes.

\paragraph{\textbf{State and control processes.}}

Fix $n\in\mathbb{N}$.
For each class $i\in\mathcal{I}$ and time $t\ge 0$, denote by $Q_{p,i}^n(t)$ the number of class-$i$ jobs waiting for prefill, by $X_i^n(t)$ the number in prefill service, by $Q_{d,i}^n(t)$ the number waiting for decode (prefill completed), and by $Y_{m,i}^n(t)$ and $Y_{s,i}^n(t)$ the numbers in mixed and solo decode, respectively.
These processes are right-continuous, integer-valued, and change by unit jumps when individual jobs enter or leave the corresponding stage.
The total class-$i$ content in prefill and decode is
\begin{equation}
    Z_{p,i}^n(t) := Q_{p,i}^n(t) + X_i^n(t),
    \qquad
    Z_{d,i}^n(t) := Q_{d,i}^n(t) + Y_{m,i}^n(t) + Y_{s,i}^n(t).
    \label{eq:Z-def}
\end{equation}

The cumulative primitive counting processes are defined as follows.
Let $A_i^n(t)$ be the total number of exogenous arrivals of class-$i$ jobs by time $t$.
Let $B_{p,i}^n(t)$ and $B_{d,i}^n(t)$ denote the total abandonments from the prefill and decode queues of class $i$.
Let $S_{p,i}^n(t)$ be the total prefill completions, and $S_{d,m,i}^n(t)$ and $S_{d,s,i}^n(t)$ the total mixed and solo decode completions; the total decode completions are
\begin{equation}
    S_{d,i}^n(t) := S_{d,m,i}^n(t) + S_{d,s,i}^n(t).
    \label{eq:Sd-def}
\end{equation}

These counting processes admit a standard random time-change representation.
Let
\[
\{N_{A,i}, N_{B_p,i}, N_{B_d,i}, N_{p,i}, N_{d,m,i}, N_{d,s,i}\}_{i\in\mathcal{I}}
\]
be mutually independent unit-rate Poisson processes.
Then, for each $i\in\mathcal{I}$ and $t\ge 0$,
\begin{align}
    A_i^n(t) 
    &= N_{A,i}\!\big(\lambda_i^n\, t\big), &
    B_{p,i}^n(t) 
    &= N_{B_p,i}\!\left(\int_0^t \theta_i\, Q_{p,i}^n(s)\,ds\right), \label{eq:rt-arrival-Bp}\\[0.3em]
    B_{d,i}^n(t) 
    &= N_{B_d,i}\!\left(\int_0^t \theta_i\, Q_{d,i}^n(s)\,ds\right), &
    S_{p,i}^n(t) 
    &= N_{p,i}\!\left(\int_0^t \mu_{p,i}\, X_i^n(s)\,ds\right), \label{eq:rt-Bd-Sp}\\[0.3em]
    S_{d,m,i}^n(t) 
    &= N_{d,m,i}\!\left(\int_0^t \mu_{m,i}\, Y_{m,i}^n(s)\,ds\right), &
    S_{d,s,i}^n(t) 
    &= N_{d,s,i}\!\left(\int_0^t \mu_{s,i}\, Y_{s,i}^n(s)\,ds\right). \label{eq:rt-Sdm-Sds}
\end{align}
Equations~\eqref{eq:rt-arrival-Bp}–\eqref{eq:rt-Sdm-Sds} state that each cumulative count is driven by a unit-rate Poisson process with the corresponding integrated intensity.

Scheduling decisions are encoded through cumulative control processes.
For each class $i$, let $U_{p,i}^n(t)$ be the number of jobs admitted into prefill service by time $t$, and $U_{d,m,i}^n(t)$ and $U_{d,s,i}^n(t)$ the numbers admitted into mixed and solo decode.
Mode switches between decode submodes are counted by
\begin{equation}
    M_{s\to m,i}^n(t) \;\text{ and }\; M_{m\to s,i}^n(t),
    \label{eq:M-sm-ms-def}
\end{equation}
the cumulative numbers of class-$i$ decodes switched from solo to mixed and from mixed to solo by time $t$.
These mode-switch processes are endogenous: they have no external Poisson clocks and are induced by changes in prefill activity on each GPU.
In particular, solo-to-mixed switches occur when a prefill is admitted to a GPU that was previously in pure decode mode, and mixed-to-solo switches occur when that prefill completes; these transitions are structural consequences of the prefill dynamics rather than direct control actions of the scheduling policy.

For later use, define the aggregate in-service counts
\begin{equation}
    X^n(t) := \sum_{i} X_i^n(t),\qquad
    Y_m^n(t) := \sum_{i} Y_{m,i}^n(t),\qquad
    Y_s^n(t) := \sum_{i} Y_{s,i}^n(t),
    \label{eq:agg-XY}
\end{equation}
and similarly $M_{s\to m}^n(t) := \sum_i M_{s\to m,i}^n(t)$ and $M_{m\to s}^n(t) := \sum_i M_{m\to s,i}^n(t)$.
The per-GPU physical constraints imply
\begin{align}
    0 \;\le\; X^n(t) &\;\le\; n, \label{eq:cap-prefill}\\[0.3em]
    0 \;\le\; Y_m^n(t) &\;\le\; (B-1)\,X^n(t), \label{eq:cap-mixed}\\[0.3em]
    0 \;\le\; Y_s^n(t) &\;\le\; B\big(n - X^n(t)\big). \label{eq:cap-solo}
\end{align}
Equation~\eqref{eq:cap-prefill} enforces at most one prefill per GPU, while~\eqref{eq:cap-mixed}--\eqref{eq:cap-solo} bound mixed and solo decodes according to whether a GPU is running a prefill.

\paragraph{\textbf{Admissible policies.}}
We now formalize the notion of a policy.
Let
\[
\pi^n := \Big(
    Q_{p}^n, Q_{d}^n, X^n, Y_m^n, Y_s^n,\;
    A^n, B_{p}^n, B_{d}^n,\;
    S_{p}^n, S_{d,m}^n, S_{d,s}^n,\;
    U_{p}^n, U_{d,m}^n, U_{d,s}^n,\;
    M_{s\to m}^n, M_{m\to s}^n
\Big)
\]
denote the collection of all non-primitive processes in the $n$th system (state, cumulative flows, and control processes), where each symbol stands for the vector over classes $i\in\mathcal{I}$.
Let $\Pi^n$ denote the set of policies that satisfy the admissibility conditions below:

(i) the resulting state processes satisfy the capacity constraints \eqref{eq:cap-prefill}–\eqref{eq:cap-solo} and balance equations \eqref{eq:bal-Qp}-\eqref{eq:bal-Ys} for all $t\ge 0$;

(ii) the policy is \emph{event-driven}, i.e., each control process $U_{\cdot,i}^n(t)$ can change only at arrival epochs, abandonment epochs, service-completion epochs, or at $t=0$;

(iii) within each class, prefill and decode queues are served in first-come-first-served order, and service is non-preemptive.


We say that any $\pi^n\in\Pi^n$ is an \emph{admissible policy}, under which the state processes are then uniquely determined from the primitives and the controls via the balance equations \eqref{eq:bal-Qp}–\eqref{eq:bal-Ys}.

\paragraph{\textbf{Balance equations.}}

Under any policy $\pi^n\in\Pi^n$, the state and cumulative processes satisfy the following flow-balance identities for all $i\in\mathcal{I}$ and $t\ge 0$:
\begin{align}
    Q_{p,i}^n(t) 
    &= Q_{p,i}^n(0) + A_i^n(t) - U_{p,i}^n(t) - B_{p,i}^n(t),
    \label{eq:bal-Qp}\\[0.3em]
    X_i^n(t) 
    &= X_i^n(0) + U_{p,i}^n(t) - S_{p,i}^n(t),
    \label{eq:bal-X}\\[0.3em]
    Q_{d,i}^n(t) 
    &= Q_{d,i}^n(0) + S_{p,i}^n(t) - U_{d,m,i}^n(t) - U_{d,s,i}^n(t) - B_{d,i}^n(t),
    \label{eq:bal-Qd}\\[0.3em]
    Y_{m,i}^n(t) 
    &= Y_{m,i}^n(0) + U_{d,m,i}^n(t) - S_{d,m,i}^n(t) + M_{s\to m,i}^n(t) - M_{m\to s,i}^n(t),
    \label{eq:bal-Ym}\\[0.3em]
    Y_{s,i}^n(t) 
    &= Y_{s,i}^n(0) + U_{d,s,i}^n(t) - S_{d,s,i}^n(t) + M_{m\to s,i}^n(t) - M_{s\to m,i}^n(t).
    \label{eq:bal-Ys}
\end{align}
Equation~\eqref{eq:bal-Qp} says that the prefill queue-length equals its initial content plus arrivals, minus admissions and abandonments.
Equation~\eqref{eq:bal-X} tracks prefill jobs in service as admissions minus completions.
Equation~\eqref{eq:bal-Qd} balances the decode queue as completed prefills minus admissions into decode and abandonments.
Equations~\eqref{eq:bal-Ym}--\eqref{eq:bal-Ys} track mixed and solo decodes as admissions minus completions, plus net inflow from mode switches.

Adding~\eqref{eq:bal-Ym} and~\eqref{eq:bal-Ys} eliminates the mode-switch terms and yields
\begin{equation}
    Y_{m,i}^n(t) + Y_{s,i}^n(t)
    = Y_{m,i}^n(0) + Y_{s,i}^n(0)
      + U_{d,m,i}^n(t) + U_{d,s,i}^n(t)
      - S_{d,m,i}^n(t) - S_{d,s,i}^n(t),
    \label{eq:bal-mode-invariance}
\end{equation}
so mode switches only redistribute ongoing decodes between mixed and solo, without changing their total number.

\paragraph{\textbf{Revenue models and objective functions.}}

Commercial LLM services predominantly use token-based pricing.
We consider two revenue models that differ in when revenue is recognized.

\medskip
\noindent\textit{(1) Bundled charging scheme.}  
The provider charges a single price per request based on the total number of tokens, and revenue is recognized only when the entire request completes (after decode).
For class $i$,
\begin{equation}
    w_i := c_p P_i + c_d D_i,
    \label{eq:w-i-def}
\end{equation}
where $c_p,c_d\ge 0$ are unit prices per prefill and decode token.
The per-GPU average reward over $[0,T]$ under policy $\pi^n$ is
\begin{equation}
    R^n(T;\pi^n) 
    := \frac{1}{nT}\,
        \mathbb{E}^{\pi^n}\!\Bigg[\sum_{i=1}^I w_i\,S_{d,i}^n(T)\Bigg].
    \label{eq:Rn-bundled}
\end{equation}
Only completed requests contribute to~\eqref{eq:Rn-bundled}; prefill work without decode completion yields no revenue.

\medskip
\noindent\textit{(2) Separate charging scheme.}  
Alternatively, prefill and decode tokens may be billed and recognized separately.
The corresponding per-GPU average reward is
\begin{equation}
    \tilde R^{n}(T;\pi^n) 
    := \frac{1}{nT}\,
       \mathbb{E}^{\pi^n}\!\Bigg[\sum_{i=1}^I \big( c_p P_i S_{p,i}^n(T) + c_d D_i S_{d,i}^n(T) \big)\Bigg].
    \label{eq:Rn-separate}
\end{equation}
Both objectives depend on token throughput but induce different scheduling incentives: in particular, the separate scheme~\eqref{eq:Rn-separate} may encourage aggressive prefill admissions to harvest immediate prefill revenue at the expense of downstream decode congestion.
\textcolor{tc-ins}{We develop the main bundled-charging control in Section~\ref{sec:policy} and discuss separate charging as an extension in Section~\ref{sec:separate-charging}.}

\textcolor{tc-ins}{\paragraph{\textbf{Remark on partial-output early termination.}}
In production LLM serving, users may stop generation before the model emits its end-of-sequence token, e.g.\ when the displayed answer already meets the information need.
This behavior can be accommodated by reinterpreting the decode primitive as an effective service time $T_i^{\mathrm{eff}}=\min\{T_i^{\mathrm{comp}},T_i^{\mathrm{stop}}\}$, where $T_i^{\mathrm{comp}}$ is the full generation time and $T_i^{\mathrm{stop}}$ is the user's in-service patience.
The induced effective decode rate $\mu_{m,i}^{\mathrm{eff}}=1/(D_i^{\mathrm{eff}}\tau)$ and the corresponding per-class expected reward enter the planning LP as primitives, so the gate-and-route architecture and the fluid analysis developed below apply with these substituted primitives.
For clarity of exposition we take $D_i$ as fixed in the main development and treat $D_i^{\mathrm{eff}}$ as a straightforward modeling extension.}

\section{Fluid Approximation and Steady-State Planning}

\label{sec:fluid}
In large-scale LLM deployments, providers typically operate hundreds or thousands of GPUs in parallel.
At this scale, the system state is high-dimensional and stochastic, with arrivals, service completions, and abandonments fluctuating across time and devices, making direct stochastic optimization analytically intractable and hard to interpret.
We therefore adopt a many-GPU fluid approximation: consider a sequence of systems indexed by the number of GPUs $n$, scale all queue lengths and in-service counts by $1/n$, and let $n\to\infty$.
In this limit, random fluctuations average out and the network is described by a deterministic set of flow-balance equations and capacity constraints.
Steady-state solutions of this fluid model specify per-GPU occupancies, which serve as planning targets for the stochastic control policies in Section~\ref{sec:policy}.
The formulation itself does not assume that the workload is always overloaded; depending on arrival rates, service primitives, revenue weights, and SLI constraints, the resulting fluid operating point may be underloaded, critically loaded, or congested.

\paragraph{\textbf{Fluid-scaled processes.}}
For any stochastic process $W^n(t)$ in the $n$-th system, we define its fluid-scaled version by
\[
\bar W^n(t) := \frac{1}{n}\,W^n(t),\qquad t\ge 0.
\]
We use an overbar to indicate such scaled processes (e.g., $\bar Q_{p,i}^n(t)$, $\bar X_i^n(t)$, $\bar Y_{m,i}^n(t)$), and we write the corresponding lowercase letters (e.g., $q_{p,i}(t)$, $x_i(t)$, $y_{m,i}(t)$) for generic deterministic fluid trajectories that arise as limits of these scaled processes in Section~\ref{subsec:stochastic-model}.
These functions satisfy, for all $t\ge 0$ and all $i\in\mathcal{I}$, the flow-balance equations:


\begin{align}
q_{p,i}(t) 
&= q_{p,i}(0) + a_i(t) - u_{p,i}(t) - b_{p,i}(t), \label{eq:fluid-qp}\\
x_i(t) 
&= x_i(0) + u_{p,i}(t) - s_{p,i}(t), \label{eq:fluid-x}\\
q_{d,i}(t) 
&= q_{d,i}(0) + s_{p,i}(t) - u_{d,m,i}(t) - u_{d,s,i}(t) - b_{d,i}(t), \label{eq:fluid-qd}\\
y_{m,i}(t) 
&= y_{m,i}(0) + u_{d,m,i}(t) - s_{d,m,i}(t) + m_{s\to m,i}(t) - m_{m\to s,i}(t), \label{eq:fluid-ym}\\
y_{s,i}(t) 
&= y_{s,i}(0) + u_{d,s,i}(t) - s_{d,s,i}(t) + m_{m\to s,i}(t) - m_{s\to m,i}(t). \label{eq:fluid-ys}
\end{align}
Here $q_{p,i}(t)$ and $q_{d,i}(t)$ are the prefill and decode queue contents, $x_i(t)$, $y_{m,i}(t)$ and $y_{s,i}(t)$ are the in-service masses in prefill, mixed decode and solo decode, $u_{p,i}(t)$, $u_{d,m,i}(t)$ and $u_{d,s,i}(t)$ are the cumulative admissions into these stages, $b_{p,i}(t)$ and $b_{d,i}(t)$ are cumulative abandonments, $s_{p,i}(t)$, $s_{d,m,i}(t)$ and $s_{d,s,i}(t)$ are cumulative service completions, and $m_{s\to m,i}(t)$, $m_{m\to s,i}(t)$ are cumulative mode switches between decode submodes.

The primitive arrivals, abandonments, and service completions evolve at their mean rates:
\begin{align}
a_i(t) &= \lambda_i\, t, \label{eq:fluid-a}\\
b_{p,i}(t) &= \int_0^t \theta_{i}\,q_{p,i}(s)\,ds,\qquad
b_{d,i}(t) = \int_0^t \theta_{i}\,q_{d,i}(s)\,ds, \label{eq:fluid-b}\\
s_{p,i}(t) &= \int_0^t \mu_{p,i}\,x_i(s)\,ds, \label{eq:fluid-dp}\\
s_{d,m,i}(t) &= \int_0^t \mu_{m,i}\,y_{m,i}(s)\,ds,\qquad
s_{d,s,i}(t) = \int_0^t \mu_{s,i}\,y_{s,i}(s)\,ds. \label{eq:fluid-dd}
\end{align}
Equation~\eqref{eq:fluid-a} gives the deterministic arrival rate, \eqref{eq:fluid-b} the abandonment flows from the prefill and decode queues, and \eqref{eq:fluid-dp}--\eqref{eq:fluid-dd} the prefill and decode completion flows driven by the in-service masses.

Let
\[
x(t):=\sum_{i}x_i(t),\qquad
y_m(t):=\sum_{i}y_{m,i}(t),\qquad
y_s(t):=\sum_{i}y_{s,i}(t),
\]
and similarly $u_p(t):=\sum_i u_{p,i}(t)$ and $s_p(t):=\sum_i s_{p,i}(t)$.
The per-GPU capacity constraints in the fluid model are
\begin{align}
& 0 \le x(t) \le 1,\label{eq:fluid-cap-x}\\
& 0 \le y_m(t) \le (B-1)\,x(t),\label{eq:fluid-cap-ym}\\
& 0 \le y_s(t) \le B\,\big(1-x(t)\big),\label{eq:fluid-cap-ys}
\end{align}
which mirror the prefill and decode caps in \eqref{eq:cap-prefill}–\eqref{eq:cap-solo} after scaling by $n$.

Finally, admission feasibility holds at the fluid level.
For all $0\le s\le t$ and all $i\in\mathcal{I}$,
\begin{align}
u_{d,m,i}(s,t)+u_{d,s,i}(s,t)
&\le  q_{d,i}(s)+s_{p,i}(s,t)-b_{d,i}(s,t),\label{eq:fluid-feas-decode}\\
u_{p,i}(s,t) &\le  q_{p,i}(s)+a_i(s,t)-b_{p,i}(s,t),\label{eq:fluid-feas-prefill}
\end{align}
where $u_{p,i}(s,t):=u_{p,i}(t) - u_{p,i}(s)$ and $a_i(s,t)$, $s_{p,i}(s,t)$, $b_{p,i}(s,t)$, etc.\ are defined analogously.
Inequalities~\eqref{eq:fluid-feas-decode}--\eqref{eq:fluid-feas-prefill} state that, over any time interval $[s,t]$, admissions into each stage cannot exceed the fluid already present in the corresponding buffer plus the net inflow into that buffer.

\begin{assumption}[Convergence of initial state]
\label{ass:initial_state}
    The initial states of fluid-scaled processes converge to a deterministic state:
$\big(\bar Q_{p,i}^n(0),\bar Q_{d,i}^n(0),\bar X_i^n(0),\bar Y_{m,i}^n(0),\bar Y_{s,i}^n(0)\big)\Rightarrow (q_{p,i}(0),q_{d,i}(0),x_i(0),y_{m,i}(0),y_{s,i}(0))$;
\end{assumption}

\begin{theorem}[Fluid limit]\label{thm:fluid-limit}
Fix any finite horizon $T>0$.
The sequence of fluid-scaled stochastic processes
\[
\begin{aligned}
\bar{\mathcal{X}}^n(t):=
\big(&\{\bar Q_{p,i}^n,\bar Q_{d,i}^n,\bar X_i^n,\bar Y_{m,i}^n,\bar Y_{s,i}^n\}_{i\in\mathcal{I}},
\{\bar S_{p,i}^n,\bar S_{d,m,i}^n,\bar S_{d,s,i}^n\}_{i\in\mathcal{I}},
\{\bar B_{p,i}^n,\bar B_{d,i}^n\}_{i\in\mathcal{I}},\\
&\{\bar U_{p,i}^n,\bar U_{d,m,i}^n,\bar U_{d,s,i}^n\}_{i\in\mathcal{I}},
\{\bar M_{s\to m,i}^n-\bar M_{m\to s,i}^n\}_{i\in\mathcal{I}}\big)
\end{aligned}
\]
is tight in $\mathbb{D}([0,T],\mathbb{R}^d)$ under the Skorokhod $J_1$ topology, where $d$ is the total dimension of the vector above.
Moreover, any subsequential weak limit $\bar{\mathcal{X}}(t)$ is almost surely continuous and, on $[0,T]$, satisfies the fluid model equations \eqref{eq:fluid-qp}--\eqref{eq:fluid-ys} with \eqref{eq:fluid-a}--\eqref{eq:fluid-dd}, the capacity constraints \eqref{eq:fluid-cap-x}--\eqref{eq:fluid-cap-ys}, with initial state given by $\mathbf{z}_0$.
\end{theorem}

\begingroup
\setlength{\abovedisplayskip}{4pt}
\setlength{\belowdisplayskip}{4pt}
\setlength{\abovedisplayshortskip}{2pt}
\setlength{\belowdisplayshortskip}{2pt}

\paragraph{\textbf{Fluid Reward Objectives.}}
Consistent with the revenue models defined in Section~\ref{sec:problem}, we formulate the fluid control objectives for the two charging schemes separately.

\noindent \textit{(1) Bundled Objective.}
Under the bundled scheme, value is realized only upon request completion.
Thus, the objective maximizes the weighted throughput of the decode phase:
\begin{equation}
    R(T) := \frac{1}{T} \int_0^T \sum_{i=1}^I w_i \, (y_{m,i}(\tau) \mu_{m,i}+ y_{s,i}(\tau) \mu_{s,i}) \, d\tau,
\end{equation}
where $w_i = c_p P_i + c_d D_i$ is the total reward for a completed class-$i$ request.
Note that the prefill activity $x_i(\tau)$ contributes to the objective only indirectly by feeding the decode queue.

\noindent \textit{(2) Separate Objective.}
Under the separate scheme, the system accumulates value continuously as tokens are processed in both phases.
The objective becomes:
{\color{tc-ins}
\begin{equation}
    \tilde R(T) := \frac{1}{T} \int_0^T \sum_{i=1}^I \Big[ \underbrace{(c_p P_i)}_{\text{prefill value}} \mu_{p,i} \, x_i(\tau) + \underbrace{(c_d D_i)}_{\text{decode value}} (y_{m,i}(\tau) \mu_{m,i}+ y_{s,i}(\tau) \mu_{s,i}) \Big] \, d\tau.
\end{equation}
}

\subsection{Fluid Control Problem}\label{subsec:fluid-routing}

We first consider the bundle objective (the separate charging scheme will be discussed in Section~\ref{sec:separate-charging}) and solve a steady-state (fluid) optimization to choose the optimal long-run occupancy shares and the routing of prefilled tasks across mixed and solo decode pools.
This formulation intentionally abstracts away the transient effects of stochastic variability in interarrival, prefill, and decode times (and abandonment/patience, if present), and instead enforces constraints only in terms of average arrival and service rates.
The solution delivers capacity-splitting targets that will guide the stochastic control policy developed in the next section.

\begin{equation}
\begin{aligned}
\label{eq:LP-original}
\max_{\{x_i,\, y_{m,i},\, y_{s,i},\, q_{p,i},\, q_{d,i}\}} 
\quad & \sum_{i=1}^{I}  w_i\,\bigl(y_{m,i}\, \mu_{m,i} \;+\;  y_{s,i}\, \mu_{s,i}\bigr) \\[4pt]
\text{s.t.}\quad 
& \sum_{i=1}^{I} x_{i} \;\le\; 1, && \text{(Prefill Capacity)} \\[4pt]
& \sum_{i=1}^{I}  y_{m,i} \;\le\; (B-1)\sum_{i=1}^{I} x_{i}, && \text{(Mixed Decode Capacity)} \\[4pt]
& \sum_{i=1}^{I}  y_{s,i} \;\le\; B\Big(1 - \sum_{i=1}^{I} x_{i}\Big), && \text{(Solo Decode Capacity)} \\[6pt]
& \lambda_i - \theta_i\, q_{p,i} \;=\; \mu_{p,i}\, x_i, \qquad \forall i, && \text{(Prefill Flow Balance)} \\[4pt]
& \mu_{p,i}\, x_i\ - \theta_i\, q_{d,i} \;=\; 
\mu_{m,i}\,  y_{m,i} \;+\; \mu_{s,i}\,  y_{s,i}, \qquad \forall i, && \text{(Decode Flow Balance)} \\[6pt]
& x_i,\,  y_{m,i},\,  y_{s,i},\, q_{d,i},\, q_{p,i} \;\ge\; 0, \qquad \forall i. && \text{(Non-negativity)}
\end{aligned}
\end{equation}

The linear program \eqref{eq:LP-original} describes the steady-state fluid model for routing prefill and decode work of multiple classes across GPUs with batch size $B$.
The decision variables are long-run, per-GPU averages: $x_i$ is the fraction of time a GPU devotes to class-$i$ prefill; $y_{m,i}$ and $y_{s,i}$ are the average class-$i$ decode occupancies in mixed mode and solo mode; and $q_{p,i}$ and $q_{d,i}$ are the average prefill and decode queue masses.

The capacity constraints in the first three lines of \eqref{eq:LP-original} enforce the per-GPU limits: at most one prefill can run on a GPU, and, conditional on whether a prefill is present, at most $B-1$ (mixed) or $B$ (solo) decodes can be served in parallel.
The flow-balance constraints for each class $i$ state that, in steady state, arrivals net of prefill abandonments equal the prefill completion rate, and that prefill completions, net of decode abandonments, are exactly matched by the total decode completion rate.
The objective in the first line of \eqref{eq:LP-original} maximizes the per-GPU long-run reward by weighting class-$i$ decode completions in mixed and solo modes with $w_i = c_p P_i + c_d D_i$, the total value of a completed request under bundled pricing.

\endgroup

\begin{proposition}[Decode-buffer elimination]\label{prop:qd-zero}
\textcolor{tc-ins}{In the calibrated GPU regime where solo decode is at least as efficient as mixed decode, i.e., $\gamma\,\tau \ge (B-1)/B$, the steady-state fluid LP admits an optimal solution with $q_{d,i}^\star = 0$ for all $i$.}
\end{proposition}

\textcolor{tc-ins}{The efficiency condition in Proposition~\ref{prop:qd-zero} reflects the measured prefill--decode asymmetry: a solo decode GPU uses all $B$ slots for decode, while a mixed GPU reserves one slot for prefill and runs slower decode iterations.}
Since revenue is generated only when a request \emph{finishes} decode, fluid mass held in the decode buffer yields no reward and only delays completions.
We prove that any LP solution with $q_{d,i}>0$ can be improved by moving this backlog upstream while keeping the capacity constraints, which weakly increases the completion rate; hence, at optimality, the decode buffer is empty in steady state.

\textcolor{tc-ins}{This coupling also changes the nature of the control problem. Admitting prefill work does not merely add future decode demand; it also changes the mixed/solo GPU composition and hence the downstream decode service capacity available to that demand. As a result, the policy must regulate admission and routing jointly through occupancy targets, rather than treat routing over a fixed server pool as a separate subproblem.}

A key implication for control is that the optimal reward rate is determined by \emph{occupancy proportions}, i.e. how much GPU time is spent on prefill and how decode slots are filled.
In a backlogged system, the fluid-optimal plan keeps both prefill and decode fully utilized and fixes the fraction of GPUs running prefills at its target level.
The remaining design question is how to route completed prefills between mixed and solo decode so that capacity is saturated without building a decode backlog.
This leads to a simple gate-and-route architecture: a \emph{prefill gate} that regulates the target prefill occupancy (and class mix) and a \emph{decode router} that splits work between mixed and solo decodes to keep slots busy while preventing persistent decode queues.

\section{\textcolor{tc-ins}{Gate-and-Route Control under Bundled Charging}}
\label{sec:policy}

{\color{tc-ins}
Building upon the fluid-optimal solution, we now develop the core implementable control framework for the stochastic $n$-GPU system under the bundled, completion-based revenue objective.
Our approach operationalizes the fluid prescriptions by decomposing the scheduling problem into two hierarchical stages: a static resource partitioning phase that fixes the cluster configuration, and a dynamic control phase that manages job admission and routing in real time.
This section focuses on the occupancy-based \emph{Gate-and-Route} policy, which regulates prefill occupancies and decode routing to attain the fluid-optimal throughput under bundled charging.
Extensions that modify revenue recognition or add service-level requirements are deferred to Section~\ref{sec:slo}.}

\subsection{Bundled Charging Scheme}
\label{subsec:policy-occupancy}
The design of this policy is inspired by the structural insight from Proposition~\ref{prop:qd-zero}, which reveals that the relative efficiency gain of solo decoding is class-independent.
This property suggests a decomposition of the complex scheduling problem.
We can \emph{statically} partition the cluster resources to ensure that the aggregate decode capacity is critically loaded in the fluid limit, capable of fully digesting the downstream workload generated by the optimal prefill throughput.
With the decode stage dimensioned to clear the traffic naturally, the burden of optimization shifts upstream.
Consequently, we focus our \emph{fine-grained dynamic control} on the prefill admission to strictly regulate the job mix and occupancy, allowing the decode stage to operate under a simple work-conserving discipline (FCFS) while still guiding the system toward the fluid-optimal state.
\subsubsection*{Static Planning}
\label{subsec:policy-static-occ}

Let $n$ be the number of GPUs and $B$ the per-GPU decode stream cap.
Take any optimal per-GPU solution of the steady-state fluid LP, denoted by $\bigl(x_i^\star,\,y_{m,i}^\star,\,y_{s,i}^\star,\,q_{p,i}^\star\bigr)_{i\in\mathcal I}$.

Fix the number of mixed GPUs as
\[
M := \Big\lceil n \sum_{i\in\mathcal I} x_i^\star \Big\rceil,
\]
choose any subset $\mathcal{G}_{\mathrm{mix}}\subset\{1,\dots,n\}$ with $|\mathcal{G}_{\mathrm{mix}}|=M$, and set $\mathcal{G}_{\mathrm{solo}}:=\{1,\dots,n\}\setminus\mathcal{G}_{\mathrm{mix}}$.
A GPU $g\in\mathcal{G}_{\mathrm{mix}}$ permanently reserves one slot for prefill (or equivalently, those GPUs prioritize new prefill over decode jobs) and may run at most $(B-1)$ decodes; a GPU $g\in\mathcal{G}_{\mathrm{solo}}$ never runs prefills and may run at most $B$ decodes.

\subsubsection*{Dynamic Control}
\label{subsec:policy-dynamic-occ}

Let $X_i^n(t)$ denote the number of class-$i$ prefill tasks currently in service and $Q_{p,i}^n(t)$ denote the prefill queue length at time $t$.

\paragraph{\textbf{Upstream (prefill) gate on mixed GPUs.}}
Prefills run only on $\mathcal{G}_{\mathrm{mix}}$.
Whenever a mixed GPU $g\in\mathcal{G}_{\mathrm{mix}}$ has its reserved prefill slot idle, identify the set of classes with waiting jobs, $\mathcal{I}_{\mathrm{wait}} = \{i \in \mathcal{I} : Q_{p,i}^n(t^-) > 0\}$.
If $\mathcal{I}_{\mathrm{wait}}$ is empty, the slot remains idle.
Otherwise, compute the occupancy deviation index for each candidate class:
\[
\xi_i(t^-) := \frac{1}{x_i^\star}\bigl(X_{i}^n(t^-) - n\,x_{i}^\star\bigr).
\]
The scheduler admits the head-of-line prefill of the class $i^\star$ that minimizes this deviation:
\[
i^\star \in \arg\min_{i \in \mathcal{I}_{\mathrm{wait}}} \xi_i(t^-).
\]
If there are multiple classes achieving the minimum deviation, ties are broken by selecting the class with the largest queue deviation $\delta_i(t^-) := Q_{p,i}^n(t^-) - Q_{\mathrm{P},i}^{\dagger}$.
Service is non-preemptive and FCFS within each class.

The gate is a negative-feedback rule around the fluid targets $x_i^\star$: classes whose prefill occupancy exceeds their target (large $\xi_i$) are held back, while under-served classes (smallest $\xi_i$) are pulled up by being admitted first.
Since the total mixed prefill capacity $\sum_i x_i^\star$ is fixed by the static planning, not all classes can be above target at once, and repeatedly correcting the most deviated class keeps the long-run average occupancies fluctuating in a small neighborhood of the fluid-optimal levels.
\paragraph{\textbf{Downstream (decode) routing.}}
Maintain a single decode buffer with class-$i$ queue length $Q_{d,i}^n(t)$.
When a class-$i$ job requires decode (either immediately after prefill completion or upon a decode completion elsewhere), route as follows:
\begin{enumerate}
\item If some GPU in $\mathcal{G}_{\mathrm{solo}}$ has a free decode slot, place the job uniformly at random among such GPUs.
\item Otherwise, if some GPU in $\mathcal{G}_{\mathrm{mix}}$ has a free decode slot, place it randomly among such GPUs.
\item Otherwise, append the job to the decode buffer (FCFS across class).
\end{enumerate}

The key insight is that, from a token-level viewpoint, the decode stage only needs to keep up with the stream of decode tokens created by the prefill gate.
Once the policy stabilizes this token production rate and keeps decode compute fully utilized, the decode workload is always consumable, and the detailed class mix becomes secondary.
This is why a simple work-conserving rule such as FCFS is sufficient at decode.

Two mechanisms make this intuition rigorous.
Static planning fixes the mixed versus solo partition so that decode is never overloaded in the fluid limit, and in the binding case it is critically loaded at an LP-optimal point with zero steady-state decode buffer.
GPU physics further implies that the relative speed advantage of solo decoding over mixed decoding is the same across classes, which lets us translate capacity between solo and mixed in a class-agnostic way and treat decode as effectively homogeneous in heavy traffic.
Theorem~\ref{thm:asymptotic_optimality_occ} formalizes this insight by proving that the resulting gate-and-route policy, with FCFS decoding, achieves asymptotic optimality.

\begin{theorem}[Asymptotic optimality of occupancy-based Gate-and-Route Policy]\label{thm:asymptotic_optimality_occ}
Let $R^\star$ denote the optimal objective value of the steady-state fluid routing LP, let Assumption~\ref{ass:initial_state} hold, and assume $\theta_i>0$ for all $i\in\mathcal{I}$.
Let $\pi^{n,\star}$ denote the occupancy-based Gate-and-Route policy parameterized by an optimal solution of the steady-state fluid routing LP.
Then $\pi^{n,\star}$ is asymptotically optimal:
\[
\liminf_{T\to\infty}\;\liminf_{n\to\infty}\; R^n(T;\pi^{n,\star})\;=\;R^\star.
\]
\end{theorem}

\textcolor{tc-ins}{The strict positivity $\theta_i>0$ is used in the proof to obtain a negative Lyapunov drift whenever a class-$i$ queue persists; operationally, our online controller imposes a small common regularization $\theta_i=\theta=3\times10^{-4}$ in the planning LP regardless of whether real abandonment is observed (see Section~\ref{sec:numerical}).
When the fluid system is underloaded so that all queues drain in the limit, the LP optimum is attained on the slack of the prefill flow-balance constraint and the result holds trivially; the proof in the EC focuses on the binding-capacity case in which the gate-and-route argument is nontrivial.}

\section{\textcolor{tc-ins}{Extensions: Separate Charging and SLI-Aware Control}}
\label{sec:slo}

{\color{tc-ins}
The preceding section develops the main Gate-and-Route policy under bundled charging, where revenue is credited only when a request completes decode.
We now show that the same fluid-planning architecture supports two extensions.
First, changing the timing of revenue recognition leads to a separate-charging pricing extension and a corresponding priority-based prefill gate.
Second, adding service-level requirements leads to SLI-aware planning constraints and routing rules.}

\subsection{Separate Charging Scheme}
\label{sec:separate-charging}

The bundled-revenue formulation studied above serves as our primary benchmark and aligns with the objective of maximizing the throughput of completed requests.
To complement this view, we also study a \emph{separate charging} objective in which value is recognized separately at prefill and decode.
This variant lets us derive a counterpart optimal policy and clarify how the timing of revenue recognition changes the incentives for admission and routing.
Formally, we define the per-GPU time-averaged reward as follows.
\begin{equation}
\label{eq:sep-reward-stochastic}
\tilde R_n(T;\pi^n)
\;:=\;
\frac{1}{nT}\,\mathbb{E}_{\pi^n}\!\left[
\sum_{i=1}^I \Big(c_p\,P_i\,S^n_{p,i}(T) \;+\; c_d\,D_i\,S^n_{d,i}(T)\Big)
\right],
\end{equation}
where $c_p,c_d\ge 0$ are unit prices, and $S^n_{p,i}(T)$ and $S^n_{d,i}(T)$ are cumulative prefill and decode completions for class $i$ by time $T$.

In steady state, we optimize the corresponding fluid objective over the same feasibility constraints as the bundled LP (i.e., \eqref{eq:LP-original}).
Substituting the service-rate definitions shows that the objective coefficients are class-independent, so the separate-charging LP depends on $(x_i,y_{m,i},y_{s,i})$ only through the aggregate occupancies:
\begin{equation}
\label{eq:LP-separate}
\max_{(x,y,q)}
\; c_p\,\frac{C}{\tau}\sum_{i=1}^I x_i
\;+\;
\frac{c_d}{\tau}\sum_{i=1}^I y_{m,i}
\;+\;
c_d\,\gamma \sum_{i=1}^I y_{s,i}.
\end{equation}

These structural properties yield three key insights for policy design under separate charging.
First, since prefill occupancy earns the same marginal reward $c_p C/\tau$ across all classes, the pricing structure itself does not prioritize any specific class mix in the prefill stage.
Second, because solo-mode decode iterations are faster than mixed-mode ones ($\gamma > 1/\tau$), solo decode occupancy is strictly more valuable per unit time.
A revenue-maximizing controller should thus prioritize saturating solo capacity.

Finally, unlike the bundled scheme where revenue is deferred, separate charging incentivizes the system to maintain a high "inventory" of downstream work.
To keep the more valuable decode slots busy, the prefill gate should favor classes with a larger decode-to-prefill ratio $D_i/P_i$, as they generate more future decode revenue per unit of prefill capacity consumed.
This logic motivates the static priority index used in the \textit{Prioritize-and-Route} policy described below.
The separate-charging optimum may also tolerate persistent decode backlogs, since these backlogs serve as a buffer to ensure continuous decode revenue.

\subsubsection{Prioritize-and-Route policy}

Under the separate-charging objective we reuse the same gate-and-route architecture as in the bundled scheme; the only change is the prefill admission rule (priority).
We briefly summarize the modifications.

\emph{\textbf{Static planning.}} Solve the separate-charging LP~\eqref{eq:LP-separate}.
We define the prefill queue targets and partition the GPUs into $G_{\mathrm{mix}}$ and $G_{\mathrm{solo}}$ following the identical procedure as in the bundled case, using the optimal prefill occupancies $\tilde x_i^\star$ to determine the partition size $\tilde M$.

\emph{\textbf{Dynamic control.}}
The downstream decode routing is unchanged, and the only modification is the upstream prefill gate on $G_{\mathrm{mix}}$.
Whenever a reserved prefill slot becomes idle, let $\mathcal{I}_{\mathrm{wait}}(t^-):=\{i\in\mathcal{I}:Q_{p,i}^n(t^-)>0\}$ denote the set of classes with at least one waiting job.
If $\mathcal{I}_{\mathrm{wait}}(t^-)=\emptyset$, the slot remains idle; otherwise we admit the head-of-line job from a class with the largest decode-to-prefill ratio among waiting classes:
\[
i^\star \in \arg\max_{i\in \mathcal{I}_{\mathrm{wait}}(t^-)}\Bigl\{\frac{D_i}{P_i}\Bigr\},
\]
breaking ties arbitrarily.
Service remains FCFS within each class.

\begin{theorem}[Asymptotic optimality under separate charging]
\label{thm:sep-optimality}
Let $\tilde R^\star$ be the optimal value of the steady-state fluid LP \eqref{eq:LP-separate}, and let $\tilde R_n(T;\pi^n)$ be the per-GPU separate-charging reward in \eqref{eq:sep-reward-stochastic}.
Assume Assumption~\ref{ass:initial_state} holds.

For each $n\in\mathbb{N}$, define the prioritize-and-route policy $\tilde \pi^{n,\star}$ as above, then $\tilde \pi^{n,\star}$ is asymptotically optimal for the separate-charging objective:
\[
\liminf_{T\to\infty}\;\liminf_{n\to\infty}\; \tilde R_n(T;\tilde \pi^{n,\star}) \;=\; \tilde R^\star.
\]
\end{theorem}

\paragraph{\textbf{The Revenue-Congestion Trade-off and Operational Risks.}}
While Theorem~\ref{thm:sep-optimality} guarantees asymptotic optimality for the separate-charging objective, the underlying incentive shift comes from the structure of \eqref{eq:LP-separate}: revenue is recognized at prefill and decode separately, so a revenue-driven controller may exploit any available prefill capacity even when decode is already congested.
This decoupling changes where congestion accumulates.
Under bundled charging, the policy tends to regulate admissions so that the prefill buffer absorbs overload while the downstream decode buffer remains comparatively lean.
Under separate charging, the system may instead build substantial decode backlogs to keep decode slots continuously busy, as illustrated in Figure~\ref{fig:bundle_separate}.
Operationally, this is problematic because it can lead to memory pressure and create long post-prefill waits and, in extreme cases, requests that complete prefill (and generate revenue) but experience severely delayed decode or abandon before completion.
This motivates augmenting the revenue objective with explicit service constraints.
\textcolor{tc-ins}{The SLI-aware extension below introduces such constraints, and Section~\ref{sec:numerical} uses shadow-price analysis to quantify the resulting economic trade-offs.}

\begin{figure}[htbp]
    \centering
    \begin{minipage}{0.48\textwidth}
        \centering
        \includegraphics[width=\textwidth]{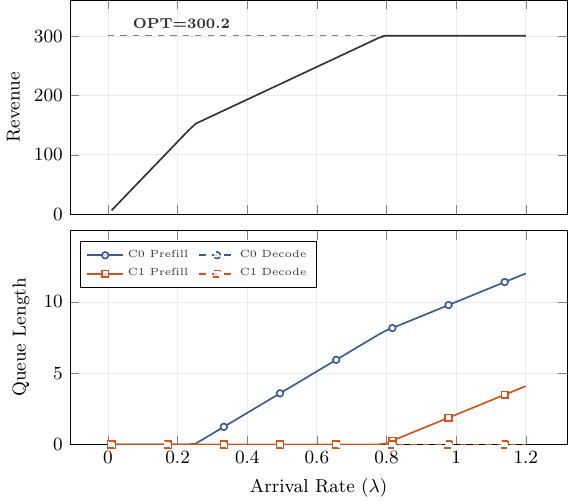}
    \end{minipage}
    \hfill  
    \begin{minipage}{0.48\textwidth}
        \centering
        \includegraphics[width=\textwidth]{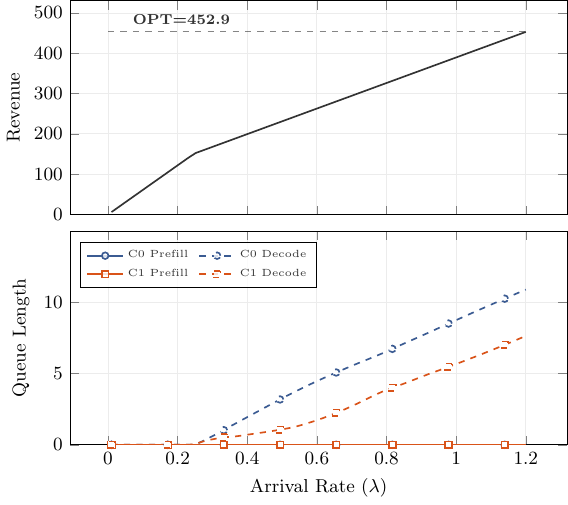}
    \end{minipage}
    \caption{Comparison of Revenue and Queue Lengths under Bundled vs. Separate Charging Schemes (C0: class 0, C1: class 1).}
    \label{fig:bundle_separate}
    \vspace{-4mm}
\end{figure}

{\color{tc-ins}
While the separate-charging extension changes the revenue objective, practical deployments must also adhere to different Service Level Indicators (SLIs), such as fairness and latency limits.
We next show that our fluid-based planning framework can naturally accommodate these operational requirements.
By formulating SLIs as explicit constraints (or penalty terms) within the steady-state optimization, we can generate SLI-aware target occupancies without altering the fundamental structure of the control policy.
This approach offers a flexible way to trade off revenue against diverse service guarantees.}

SLIs are widely used in practice, but we emphasize that in our framework they are defined at the level of the steady-state fluid variables.
A \emph{service-level indicator (SLI)} is a user-facing performance metric (e.g., fairness or latency) computed from the system's steady-state behavior.
In our framework, an SLI is modeled as either (i) a \emph{hard constraint} of the form $g(\mathbf{x},\mathbf{y},\mathbf{q})\le \eta$, or (ii) a \emph{soft penalty} term subtracted from the revenue objective, both expressed in terms of the steady-state variables in the fluid optimization problem.

A central modeling principle we adopt throughout this section is that \emph{a reasonable SLI should not rely on persistent decode-buffer buildup}.
Indeed, in our setting decode-buffer mass generates no value under completion-based reward, but it increases waiting time and, more importantly, can lead to severe GPU-memory pressure due to KV-cache residency and migration.
Accordingly, we impose the following standing assumption for the SLI-aware planning problem: the chosen SLI specification is such that the optimization admits an optimal solution with \(\,q_{d,i}^\star=0\,\) for all \(i\) (decode-buffer elimination).
For completeness, we provide an extension that allows \(q_{d,i}^\star>0\) in the electronic companion; this case is particularly relevant under separate charging.
To this end, we illustrate several canonical SLI specifications and then present the corresponding SLI-aware planning problem and control policy.

\subsubsection*{Resource fairness (prefill and decode)}

 Fairness SLIs control the dispersion of prefill occupancies $\{x_i\}_{i\in\mathcal I}$ and solo decode occupancies $\{y_{s,i}\}_{i\in\mathcal I}$.
While attractive at a high level, fairness constraints can be costly in our setting because they directly constrain the workload mix ($\boldsymbol{x}$ and/or $\boldsymbol{y}_s$).
This can force the system away from a hardware-efficient operating point and create a structural mismatch between prefill output and downstream decode capacity, leading to idling/under-utilization and a reduction in revenue.
This effect is quantified in the Pareto frontiers in Section~\ref{sec:numerical} (Fig.~\ref{fig:frontiers}): Prefill Fairness has a steep shadow price (Fig.~\ref{fig:frontier1}), whereas Decode Fairness is comparatively cheap (Fig.~\ref{fig:frontier2}).

 \noindent\textbf{Prefill Fairness.}
\begin{equation}
\max_{i,j\in\mathcal I} \{x_i - x_j\} \;\leq\; \eta_1,
\qquad
q_{d,i} = 0,\ \forall\, i \in \mathcal I.
\label{eq:slo_fairness_prefill_constraint}
\end{equation}

Equivalently, the same fairness preference can be modeled in the objective via a penalty term
\begin{equation}
l_1 \;=\; \eta_1' \max_{i,j\in\mathcal I} \{x_i - x_j\},
\label{eq:slo_fairness_prefill}
\end{equation}
with weight $\eta_1' > 0$ that tunes the trade-off between revenue and Prefill Fairness.

 \noindent\textbf{Decode Fairness.}
\begin{equation}
\max_{i,j\in\mathcal I} \{y_{s,i} - y_{s,j}\} \;\leq\; \eta_2,
\qquad
q_{d,i} = 0,\ \forall\, i \in \mathcal I.
\label{eq:slo_fairness_decode_constraint}
\end{equation}

In penalty form, we instead add
\begin{equation}
l_2 \;=\; \eta_2' \max_{i,j\in\mathcal I} \{y_{s,i} - y_{s,j}\},
\label{eq:slo_fairness_decode}
\end{equation}
with weight $\eta_2' > 0$; larger $\eta_2'$ places more emphasis on equalizing solo decode usage across classes.

\subsubsection*{Average Time per Output Token}
\label{subsubsec:slo_TPOT}
While the worst-case Time per Output Token is governed by the chunk size $C$, the average TPOT depends on the cluster-wide balance between prefill and decode activity.
Each unit of prefill occupancy $x_i$ introduces mixed-mode iterations that slow co-resident decodes.
A natural SLI is to cap the average TPOT:
\begin{equation}
\frac{\tau(B-1) \sum_{i=1}^{I} x_i + \frac{1}{\gamma}B\bigl(1-\sum_{i=1}^{I} x_i\bigr)}
     {(B-1)\sum_{i=1}^{I} x_i + B\bigl(1-\sum_{i=1}^{I} x_i\bigr)}
\;\leq\; \eta_3,
\label{eq:slo_TPOT_constraint2}
\end{equation}
for some target $\eta_3>0$. In this formulation we retain the standard capacity constraints in \eqref{eq:LP-original}, so idling is permitted when the TPOT cap is tight; the revenue objective still discourages unnecessary idling whenever additional work can be served.

Alternatively, we can incorporate TPOT directly into the objective by penalizing total prefill load:
\begin{equation}
l_3 = \eta_3' 
\frac{\tau(B-1) \sum_{i=1}^{I} x_i + \frac{1}{\gamma}B\bigl(1-\sum_{i=1}^{I} x_i\bigr)}
     {(B-1)\sum_{i=1}^{I} x_i + B\bigl(1-\sum_{i=1}^{I} x_i\bigr)},
\label{eq:slo_TPOT_penalty}
\end{equation}
where $\eta_3' > 0$ controls the strength of the revenue–latency trade-off.

\subsection{SLI-Aware Gate-and-Route Control Policy}
\label{subsec:sli-aware-policy}

To incorporate SLIs into the steady-state optimization, we augment the objective of Section~\ref{subsec:fluid-routing} (or Section~\ref{sec:separate-charging} for separate charging) by subtracting a weighted sum of penalty terms:

\begin{equation}
\max_{(\mathbf{x}, \mathbf{y}, \mathbf{q}) \in \mathcal{F}_{\mathcal{K}}} \quad \sum_{i=1}^{I} w_i \left( \mu_{m,i} y_{m,i} + \mu_{s,i} y_{s,i} \right) - \sum_{k \in \mathcal{K}} l_k,
\label{eq:slo_objective}
\end{equation}
where $\mathcal{K}$ indexes the active SLIs and $l_k$ are chosen from \eqref{eq:slo_fairness_prefill}, \eqref{eq:slo_fairness_decode}, \eqref{eq:slo_TPOT_penalty}, or other application-specific penalties.
We append the SLI-specific constraints to the feasibility constraints from Section~\ref{subsec:fluid-routing}, and denote the feasibility region by $\mathcal{F}_{\mathcal{K}}$.

To realize the targets $(y_{m,i}^\star, y_{s,i}^\star)$ derived from the SLI-aware planning problem, we employ a randomized decode router.
For clarity of exposition, we focus on the zero-buffer case ($q_{d,i}^\star=0$) in this section, which avoids memory pressure and simplifies the tracking mechanism.
An extension that accommodates persistent decode queues ($q_{d,i}^\star > 0$) is provided in Section~\ref{sec:sli-aware-policy-general}.

\emph{\textbf{Static planning.}} This phase follows the identical procedure as defined in Section~\ref{sec:policy}, determining the cluster-level queue targets $Q_{\mathrm{P},i}^{\dagger}$ and the GPU partition sets ($\mathcal{G}_{\mathrm{mix}}$ and $\mathcal{G}_{\mathrm{solo}}$) based on the optimal solution $(x_i^\star, y_{m,i}^\star, y_{s,i}^\star, q_{p,i}^\star)$ of the corresponding SLI-aware optimization problem in~\eqref{eq:slo_objective}.

\emph{\textbf{Dynamic control.}} The prefill gate remains the same as the \emph{Gate-and-Route} policy in Section~\ref{subsec:policy-dynamic-occ}.
The decode router splits the decode buffer into mixed and solo components and computes the class-$i$ solo probability
\[
p_{s,i} \;:=\;
\begin{cases}
\dfrac{\mu_{s,i}\, y_{s,i}^\star}{\mu_{m,i}\, y_{m,i}^\star + \mu_{s,i}\, y_{s,i}^\star}, & \text{if }\mu_{m,i}\, y_{m,i}^\star + \mu_{s,i}\, y_{s,i}^\star>0,\\[8pt]
1, & \text{otherwise.}
\end{cases}
\]
Upon prefill or decode completion of class $i$, draw $U\sim\mathrm{Unif}(0,1)$; route to $\mathcal{G}_{\mathrm{solo}}$ if $U\le p_{s,i}$ and to $\mathcal{G}_{\mathrm{mix}}$ otherwise, placing the decode uniformly at random among GPUs in the targeted group with free slots (or queuing in the corresponding buffer if none available).
Note that here we logically split the decode buffer into mixed buffer and solo buffer instead of a single decode buffer mentioned in Section~\ref{sec:policy}.

The router implements a randomized load split that mirrors the fluid targets: each class-$i$ decode is sent to solo or mixed with probability $p_{s,i}$ chosen so that the long-run fraction of class-$i$ service provided by solo vs.\ mixed matches $(y_{s,i}^\star, y_{m,i}^\star)$.
As many decodes are routed over time, the law of large numbers forces the realized occupancies $\bigl(y_{m,i}^n/n,\,y_{s,i}^n/n\bigr)$ to track these target proportions, effectively ``reshuffling'' decode work until the stochastic system hovers around the desired steady-state levels.

The asymptotic optimality guarantee of Theorem~\ref{thm:asymptotic_optimality_occ} extends to the SLI-aware policy under mild regularity on the optimization problem (bounded penalties, etc.).

\begin{theorem}[Occupancy Convergence and Asymptotic Optimality of SLI-Aware Policy]
\label{thm:slo_optimality}
Let $(x_i^\star, y_{m,i}^\star, y_{s,i}^\star, q_{p,i}^\star, q_{d,i}^\star)_{i\in\mathcal{I}}$ denote an optimal solution of the SLI-aware LP \eqref{eq:slo_objective} with active SLI set $\mathcal{K}$ and corresponding penalties $\{l_k\}_{k\in\mathcal{K}}$ and constraints.
Assume the SLI-aware LP satisfies Slater's condition and each penalty $l_k$ is bounded and Lipschitz continuous in the decision variables, and assume $\theta_i>0$ for all $i\in\mathcal{I}$.
\emph{Further assume that the selected optimal solution satisfies $q_{d,i}^\star=0$ for all $i$ (e.g., by including $q_{d,i}=0$ in the constraints).}
Under the SLI-aware control policy, the scaled steady-state occupancies converge:
\[
\lim_{n\to\infty} \frac{1}{n}\mathbb{E}\bigl[X_i^{(n)}\bigr] = x_i^\star, \quad
\lim_{n\to\infty} \frac{1}{n}\mathbb{E}\bigl[Y_{m,i}^{(n)}\bigr] = y_{m,i}^\star, \quad
\lim_{n\to\infty} \frac{1}{n}\mathbb{E}\bigl[Y_{s,i}^{(n)}\bigr] = y_{s,i}^\star,
\]
for all $i\in\mathcal{I}$, and the per-GPU SLI-aware objective value converges to optimality:
\[
\lim_{n\to\infty} \frac{1}{n}\mathbb{E}\left[\sum_{i=1}^{I} w_i \left( \mu_{m,i} Y_{m,i}^{(n)} + \mu_{s,i} Y_{s,i}^{(n)} \right) - n \sum_{k\in\mathcal{K}} l_k\bigl(\boldsymbol{X}^{(n)}/n, \boldsymbol{Y}_m^{(n)}/n, \boldsymbol{Y}_s^{(n)}/n\bigr)\right]
\]
\[
= \sum_{i=1}^{I} w_i \left( \mu_{m,i} y_{m,i}^\star + \mu_{s,i} y_{s,i}^\star \right) - \sum_{k\in\mathcal{K}} l_k\bigl(\boldsymbol{x}^\star, \boldsymbol{y}_m^\star, \boldsymbol{y}_s^\star\bigr),
\]
where $\boldsymbol{X}^{(n)} = (X_1^{(n)}, \ldots, X_I^{(n)})$ and similarly for $\boldsymbol{Y}_m^{(n)}$, $\boldsymbol{Y}_s^{(n)}$.
\end{theorem}

\section{Numerical Experiments}
\label{sec:numerical}

In this section, we evaluate the performance of the proposed \emph{Gate-and-Route} policy through event-driven simulations calibrated with real-world LLM inference profiles.
{\color{tc-ins}
We use two types of experimental inputs.
The real-trace experiment replays empirical Azure-style arrival timestamps and request-level prompt/output lengths, while estimating arrival rates online.
The synthetic experiments in the main text and EC use multiclass workloads whose class lengths, arrival rates, and patience rates are specified as controlled experimental inputs.
Across both settings, the GPU service primitives are calibrated from hardware measurements.}
We first compare the online policy against representative serving heuristics on real traces, then use controlled synthetic workloads to study SLI tradeoffs, with the corresponding convergence checks reported in the EC.

\subsection{Calibration of Hyperparameters}
\label{subsec:calibration}
{\color{tc-ins}
We begin by calibrating the iteration-time primitives used by both the real-trace replay and the synthetic experiments.
The measurements are conducted on a server equipped with 4$\times$ NVIDIA A100-SXM4-40GB GPUs; detailed hardware and software specifications are reported in the EC (Table~\ref{tab:computing-infrastructure}).
We use vLLM version~0.11.0 and Qwen3-8B as the reference model for the calibration.
The key quantities are the mixed-iteration time, which determines prefill and mixed-decode service rates, and the solo-decode iteration time, which determines the solo decode speed.

For mixed-mode calibration, we vary the prefill chunk size $C$ and record the mean iteration time $\tau$ on a single GPU, then fit the linear model
\[
    \tau_{\mathrm{mix}}(C) \approx \alpha + \beta C.
\]
For Qwen3-8B, this gives $\alpha\approx 0.0174$ and $\beta\approx 6.2\times 10^{-5}$, with $R^2=0.998$.
For solo-decode calibration, we vary the KV-cache load $K$ and fit
\[
    T_{\mathrm{solo}}(K) \approx a_s + b_s K.
\]
The fitted solo relation is $T_{\mathrm{solo}}(K)=1.08\times 10^{-7}K+0.0089$ with $R^2=0.9892$.
Thus, for the 8B model, the KV-cache slope is already small relative to the intercept term; for larger models, the fixed per-iteration component becomes more dominant, so the slope-to-intercept scale is expected to decrease further.
This reconciles the constant $\tau_{\mathrm{solo}}$ used in the analytical model with the linear solo fit used in the trace experiment: the theory keeps only the first-order constant iteration time, while the trace replay includes the small KV-cache slope as a second-order refinement so that tail TTFT/TPOT metrics are evaluated under a more faithful finite-system simulator.

For the hardware and memory model, the only manually chosen constants in the trace experiment are the memory safety factor $u=0.8$ and the prefill chunk size $C=256$.
For an A100-40GB GPU, after loading one model replica, we compute the remaining usable KV-cache budget under the safety margin and set the maximum batch size as
\[
    B=\left\lfloor \frac{uM_{\mathrm{GPU}}-M_{\mathrm{model}}}{\bar m_{\mathrm{KV}}}\right\rfloor,
\]
where \(\bar m_{\mathrm{KV}}\) is the average per-request KV-cache memory proxy computed from the combined code and conversation trace statistics.
The same calibrated primitives then determine service rates.
With chunk size \(C=256\), the mixed iteration time is \(\tau_{\mathrm{mix}}(C)=\alpha+\beta C\), so class-\(i\) prefill and mixed-decode rates are \(C/(P_i\tau_{\mathrm{mix}}(C))\) and \(1/(D_i\tau_{\mathrm{mix}}(C))\), respectively.
For solo decode, the fitted relation \(T_{\mathrm{solo}}(K)=a_s+b_sK\) has a small KV-cache slope relative to the intercept, so we use the fitted iteration time at the token-equivalent KV load implied by the same safety-adjusted memory budget, and set the class-\(i\) solo-decode rate to \(1/(D_i\tau_{\mathrm{solo}})\).
}

\begin{figure*}[t]
  \centering
  \includegraphics[width=0.95\linewidth]{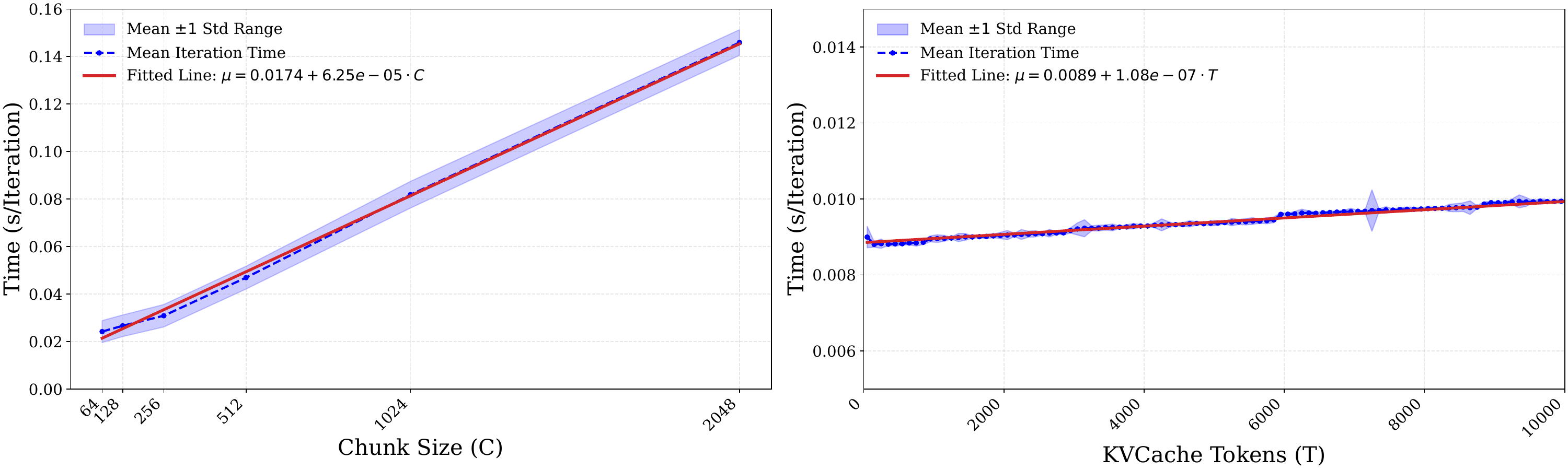}
  \caption{Iteration-time calibration for Qwen3-8B.}
  \label{fig:prefill_calibration}
\end{figure*}

\subsection{Trace-Driven Online Policy Evaluation}
\label{subsec:real-trace-online-eval}

{\color{tc-ins}
We next evaluate the proposed policy in a finite-system setting driven by real workload traces rather than synthetic exponential workloads.
The 2023 replay uses the Azure code and conversation traces released with Splitwise~\citep{Patel2024Splitwise}, including empirical arrival timestamps and request-level prompt/output lengths.
We also evaluate a 2024 Azure replay released with DynamoLLM~\citep{stojkovic2025dynamollm}.
The replays themselves do not impose abandonment; the impatience parameter is used only in the online planning LP.
Each finite-horizon simulation stops at the timestamp of the last prompt arrival in the corresponding replay.

\paragraph{\textbf{Online adaptive algorithm.}}
The online implementation does not assume a fixed arrival-rate estimate: it estimates class-level arrival rates from a rolling window, replans the fluid LP periodically, and updates the target mixed/solo GPU split accordingly.
For planning, the class-level prompt and output lengths are treated as known inputs and set to the empirical means of the native code and conversation classes.
Thus, the online component estimates traffic intensities, not token-length statistics.
The A100 memory budget, batch size, chunk size, and iteration-time primitives follow the calibration procedure in Section~\ref{subsec:calibration}.

The online controller proceeds as follows.
Let $A_i(t)$ be the cumulative number of observed class-$i$ arrivals by time $t$, let $W$ be the rolling-window length, let \(\rho\ge 1\) be the arrival-rate safety factor, and let \(\lambda_{\min}>0\) be a lower bound used during cold-start or low-traffic periods.
The small constant \(\epsilon>0\) initializes the effective window length at the beginning of the replay.
In the reported online benchmark setting, we use \(W=30\) seconds, \(\rho=3\), \(\lambda_{\min}=10^{-6}\), and \(\epsilon=10^{-9}\); replanning is performed every 10 seconds, and the planning LP uses a small common impatience parameter \(\theta_i=\theta=3\times 10^{-4}\) for all \(i\).
Because the trace datasets do not record abandonment, this \(\theta\) is used only as a small regularization parameter in the planning LP to keep it well posed; the replay evaluation itself imposes no abandonment.
At each replanning time $t_k$, it computes
\[
  N_i(t_k)=A_i(t_k)-A_i((t_k-W)^+),
  \qquad
  \bar W(t_k)=\min\{W,\max\{t_k,\epsilon\}\},
\]
and forms the conservative per-GPU arrival estimate
\begin{equation}
  \widehat\lambda_i(t_k)
  =
  \max\left\{
    \rho\,\frac{N_i(t_k)}{n\,\bar W(t_k)},
    \lambda_{\min}
  \right\}.
  \label{eq:main-online-arrival-estimate}
\end{equation}
The LP is then re-solved with $\widehat\lambda_i(t_k)$ replacing $\lambda_i$, while the empirical class-mean prompt/output lengths and calibrated service parameters remain fixed.
If $x_i^\star(t_k)$ is the resulting prefill occupancy target, the desired number of mixed GPUs is
\begin{equation}
  M^\star(t_k)=
  \left\lceil n\sum_i x_i^\star(t_k)\right\rceil.
  \label{eq:main-online-mixed-target}
\end{equation}
At replanning epochs, the simulator updates the target mixed/solo split to track \(M^\star(t_k)\); existing service is not preempted, and the new split governs subsequent placement and admission decisions.
The prefill gate then ranks waiting classes by occupancy deviation from the current LP target and admits the most under-target class first, using the prefill backlog targets \(nq_{p,i}^\star(t_k)\) as tie-breaking guidance.
Decode routing is solo-first and work-conserving: decode jobs are sent to solo GPUs whenever possible, and use mixed-GPU decode slots only when solo capacity is unavailable.

\paragraph{\textbf{Experimental results.}}
We compare against three classes of serving baselines.
The vLLM-style baseline uses prefill-first continuous batching without class-aware admission control~\citep{Kwon2023vLLM}.
The Sarathi-style baseline uses decode-first scheduling without class-aware admission control~\citep{agrawal2024sarathiserve}.
For DistServe-style baselines, we use two best-fixed-split comparators motivated by DistServe~\citep{zhong2024distserve}.
The prefill/solo variant uses the best fixed split between prefill-only GPUs and decode-only GPUs.
The mix/solo variant uses the best fixed split between mixed GPUs, which can serve both prefill and decode, and solo decode-only GPUs.
Table~\ref{tab:trace-policy-definitions} summarizes the main architectural differences among these systems-inspired benchmark policies.
All policies are evaluated on 10-GPU replays with batch size $B=16$, chunk size $C=256$, and seed 42; the 2023 and 2024 Azure interarrival times are uniformly compressed by a factor of $0.1$.
This compression is a load-scaling device: at the calibrated A100/Qwen3-8B primitives the unscaled traces leave the 10-GPU system lightly loaded, so we apply the same compression to every policy to reach the congested, prefill--decode contention regime that the policies target and that is representative of larger production clusters.
The trace-driven evaluator is a calibrated scheduling simulator: it uses measured per-GPU execution primitives and empirical request traces, while abstracting from cluster-level networking, scheduler implementation overheads, and KV-cache migration costs.

\begin{table*}[htbp]
\centering
\footnotesize
\setlength{\tabcolsep}{4pt}
\renewcommand{\arraystretch}{1.18}
\caption{Policy definitions for the trace-driven benchmark.}
\label{tab:trace-policy-definitions}
\begin{tabular}{@{}>{\raggedright\arraybackslash}p{0.23\textwidth}>{\raggedright\arraybackslash}p{0.20\textwidth}>{\raggedright\arraybackslash}p{0.23\textwidth}>{\raggedright\arraybackslash}p{0.30\textwidth}@{}}
\toprule
\textbf{Policy} & \textbf{Resource split} & \textbf{Prefill admission} & \textbf{Decode policy} \\
\midrule
\shortstack[l]{Online gate-and-route\\(Ours)} & \shortstack[l]{Online mixed/solo\\replanning from the LP} & \shortstack[l]{Occupancy-based gate\\around class targets} & \shortstack[l]{Solo-first, work-conserving\\decode routing} \\
\addlinespace[0.20em]
\shortstack[l]{vLLM-style\\\strut} & \shortstack[l]{No fixed\\prefill/decode split} & \shortstack[l]{Prefill-first continuous\\batching} & \shortstack[l]{Decodes share mixed batches\\when prefills are admitted} \\
\addlinespace[0.20em]
\shortstack[l]{Sarathi-style\\\strut} & \shortstack[l]{No fixed\\prefill/decode split} & \shortstack[l]{Admits prefills when\\slots are available} & \shortstack[l]{Decode-first local\\execution after prefill} \\
\addlinespace[0.20em]
\shortstack[l]{DistServe best split\\(prefill/solo)} & \shortstack[l]{Best fixed\\prefill/solo split} & \shortstack[l]{Class-agnostic prefill\\admission} & \shortstack[l]{Decode isolated on\\solo GPUs} \\
\addlinespace[0.20em]
\shortstack[l]{DistServe best split\\(mix/solo)} & \shortstack[l]{Best fixed\\mixed/solo split} & \shortstack[l]{Class-agnostic prefill\\admission on mixed GPUs} & \shortstack[l]{Solo-first decode\\with fixed split} \\
\bottomrule
\end{tabular}
\end{table*}

\begin{table*}[htbp]
\centering
\scriptsize
\setlength{\tabcolsep}{2.2pt}
\renewcommand{\arraystretch}{1.08}
\caption{Trace-driven policy comparison on 10-GPU Azure replays.}
\label{tab:main-real-trace-benchmarks}
\textbf{(a) 2023 Azure replay}
\vspace{0.45em}
\resizebox{\linewidth}{!}{%
\begin{tabular}{lrrrrrrrr}
\toprule
Policy & Revenue rate & Completion rate & TTFT mean & TTFT P95 & TTFT P99 & TPOT mean & TPOT P95 & TPOT P99 \\
\midrule
Online gate-and-route (Ours) & \textbf{677.96} & \textbf{0.4263} & \textbf{52.98} & 248.70 & 268.87 & 0.03254 & 0.03515 & 0.03577 \\
Sarathi-style & 560.34 & 0.3860 & 103.85 & \textbf{198.83} & \textbf{208.37} & 0.02776 & 0.03481 & 0.03541 \\
vLLM-style & 442.94 & 0.3079 & 116.02 & 216.79 & 231.67 & 0.03384 & 0.03538 & 0.03592 \\
\midrule
\textit{DistServe best split (prefill/solo)} & 416.51 & 0.2878 & 125.46 & 225.82 & 236.56 & \textbf{0.01955} & \textbf{0.02021} & \textbf{0.02062} \\
\textit{DistServe best split (mix/solo)} & 418.59 & 0.2899 & 112.04 & 229.62 & 237.17 & 0.02685 & 0.03486 & 0.03552 \\
\bottomrule
\end{tabular}%
}
\vspace{1.25em}

\textbf{(b) 2024 Azure replay}
\vspace{0.45em}
\resizebox{\linewidth}{!}{%
\begin{tabular}{lrrrrrrrr}
\toprule
Policy & Revenue rate & Completion rate & TTFT mean & TTFT P95 & TTFT P99 & TPOT mean & TPOT P95 & TPOT P99 \\
\midrule
Online gate-and-route (Ours) & \textbf{734.71} & 0.1072 & \textbf{40.49} & \textbf{132.67} & 165.17 & 0.03355 & 0.03658 & 0.03767 \\
Sarathi-style & 658.89 & 0.1048 & 78.50 & 150.53 & 158.09 & 0.03311 & 0.03730 & 0.03806 \\
vLLM-style & 696.70 & \textbf{0.1107} & 79.18 & 150.21 & \textbf{157.01} & 0.03616 & 0.03794 & 0.03858 \\
\midrule
\textit{DistServe best split (prefill/solo)} & 407.26 & 0.0640 & 82.84 & 157.58 & 166.18 & \textbf{0.02037} & \textbf{0.02200} & \textbf{0.02267} \\
\textit{DistServe best split (mix/solo)} & 362.05 & 0.0573 & 82.65 & 158.71 & 166.82 & 0.03146 & 0.03766 & 0.03845 \\
\bottomrule
\end{tabular}%
}
\end{table*}

Table~\ref{tab:main-real-trace-benchmarks} is best read as a revenue-tradeoff table: because the gate-and-route controller optimizes a token-weighted completion objective, it leads on revenue while making the implied tail-SLI tradeoffs explicit.
Figure~\ref{fig:main-real-trace-tpot-frontiers} then shows how the same control architecture can move along a TPOT--revenue frontier when stricter tail control is required.

Within this framing, Table~\ref{tab:main-real-trace-benchmarks} and Figure~\ref{fig:main-real-trace-revenue-timeseries} show that online gate-and-route attains the highest revenue rate in both replay slices and maintains a strong revenue trajectory over time.
Sarathi-style immediate decoding keeps a request on the same GPU after prefill and admits new prefills whenever slots are available; this local rule does not match aggregate prefill and decode rates, so it can either leave decode opportunities unused or admit too much prefill work.
vLLM-style prefill-first scheduling similarly lacks class-aware regulation and can overfeed downstream decode.
The DistServe-style rows serve as best-fixed-split comparators rather than adaptive online scheduling policies.
The DistServe prefill/solo variant has the lowest TPOT by isolating decode, but even at its best fixed split it has low completion and revenue; the mix/solo variant adds flexibility but remains static and class-agnostic.
Overall, these replays suggest that class-aware prefill admission and online mixed/solo replanning are important for improving revenue under nonstationary traffic, rather than optimizing a single latency or completion metric in isolation.

\begin{figure*}[htbp]
    \centering
    \includegraphics[width=0.455\textwidth]{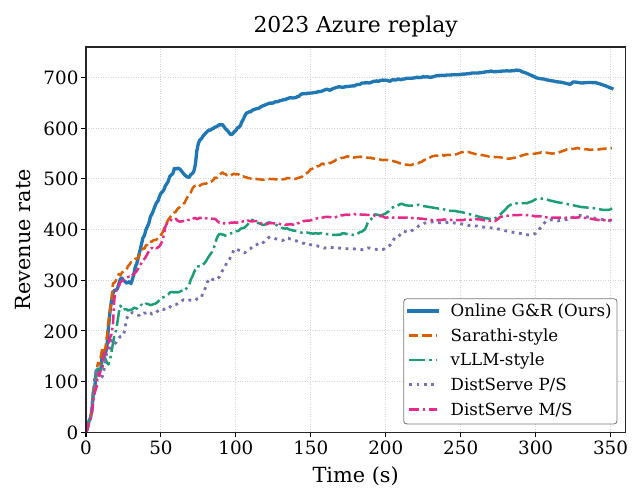}
    \hspace{0.07\textwidth}
    \includegraphics[width=0.455\textwidth]{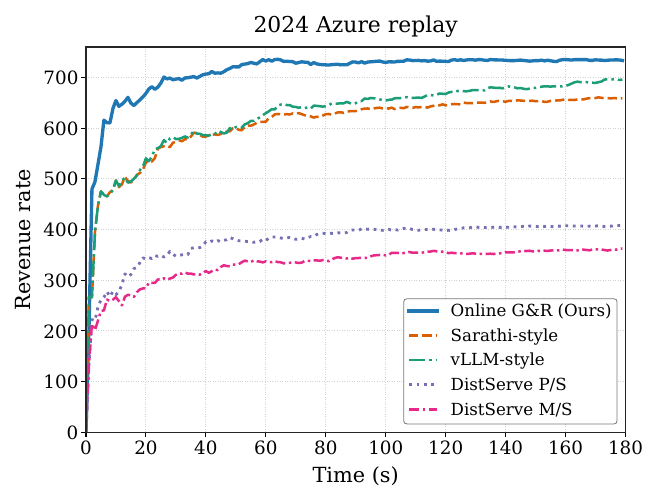}
    \caption{Running revenue rate on 10-GPU Azure replays.}
    \label{fig:main-real-trace-revenue-timeseries}
\end{figure*}

Table~\ref{tab:main-real-trace-benchmarks} also shows that online gate-and-route has a clear advantage in mean TTFT over the tested benchmarks, whereas TPOT is not uniformly minimized.
This indicates that first-token latency and per-token latency need not move together under trace-driven replay.
Figure~\ref{fig:main-real-trace-tpot-frontiers} therefore plots the TPOT--revenue frontier obtained by adding TPOT-aware planning to the same online gate-and-route architecture on the 2023 and 2024 replays.
In both panels, the star-marked Online gate-and-route point is the benchmark with no additional SLI control, and the remaining solid points are generated by varying the TPOT-control parameter \(\eta_3\) within the same controller.
Moving left from the star lowers TPOT at the cost of revenue.
The dashed segment to the right is only a visual reference: the benchmark Online gate-and-route point already represents the highest-revenue operating point within this SLI-aware online-control family, so increasing TPOT further does not produce a better frontier point.
The figure compares against Online gate-and-route, Sarathi-style, vLLM-style, and DistServe best split (mix/solo); we omit DistServe best split (prefill/solo) because it disallows mixed-mode execution and therefore does not represent the same mixed/solo tradeoff.
The benchmark locations also help interpret the tradeoff. A vLLM-style prefill-first rule tends to keep a larger share of GPUs in mixed mode, which raises TPOT, but in these no-abandonment replays it can still sustain relatively high throughput and hence competitive revenue. A Sarathi-style decode-first rule tends to reduce TPOT by draining decode more aggressively, but it can also leave the system short of decode work when the downstream supply is insufficient, which lowers throughput and revenue. The proposed gate-and-route controller can move systematically along the frontier between these regimes, allowing the operator to trade revenue against TPOT according to the service objective.

\begin{figure*}[htbp]
    \centering
    \includegraphics[width=0.455\textwidth]{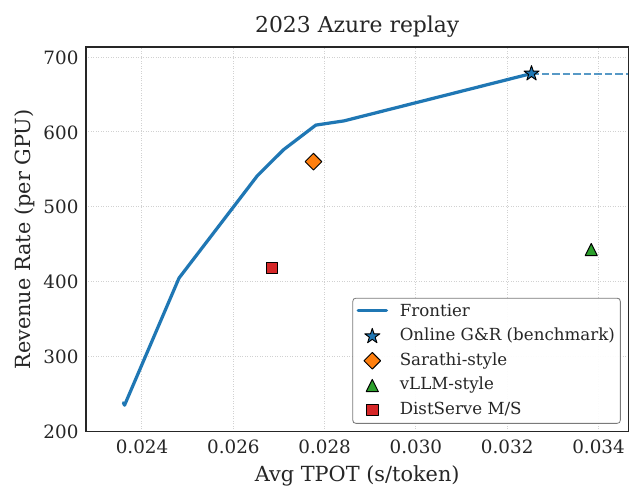}
    \hspace{0.07\textwidth}
    \includegraphics[width=0.455\textwidth]{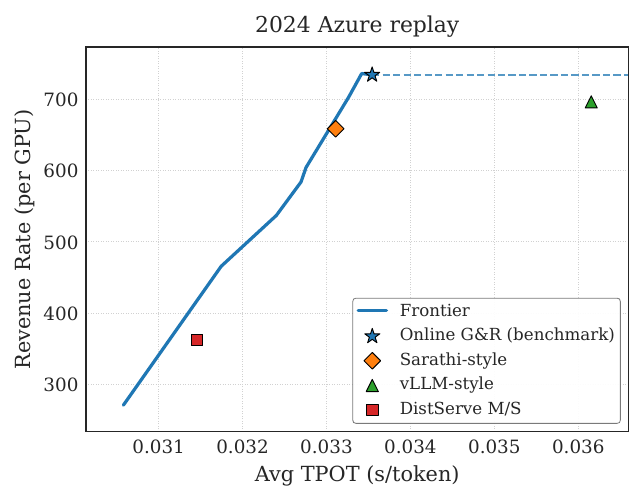}
    \caption{Operating frontiers for average TPOT versus revenue on the 10-GPU Azure replays.}
    \label{fig:main-real-trace-tpot-frontiers}
\end{figure*}
}

\subsection{The Cost of Service Quality and Hardware Resources}

To investigate the cost of service quality and fairness, we analyze the trade-off between total revenue and specific SLIs.
We use the same two-class synthetic instance as in the convergence analysis reported in Appendix~\ref{subsec:convergence}, which is representative since it consists of long-prefill-short-decode and short-prefill-long-decode classes.
We formulate this as a constrained optimization problem where we maximize revenue subject to a strict constraint on exactly one SLI metric at a time: Prefill Fairness, Decode Fairness, or Time Per Output Token (TPOT), while relaxing the others.
This approach allows us to isolate the ``price'' of each specific constraint in terms of lost revenue.

\paragraph{\textbf{The Shadow Price of SLIs}}

Figure \ref{fig:frontiers} illustrates the Pareto frontiers for three operational constraints.
We interpret the slope of these curves as the \textbf{shadow price}, which is the marginal revenue sacrificed to achieve a stricter SLI target.
It reveals distinct economic sensitivities:

\begin{figure}[ht]
    \centering
    \captionsetup{justification=centering}
    \captionsetup[subfigure]{justification=centering, font=small}

    \begin{subfigure}{0.32\textwidth}
        \centering
        \includegraphics[width=\textwidth]{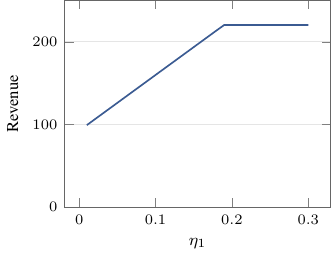}
        \caption{Prefill Fairness ($\eta_1$).}
        \label{fig:frontier1}
    \end{subfigure}
    \hfill
    \begin{subfigure}{0.32\textwidth}
        \centering
        \includegraphics[width=\textwidth]{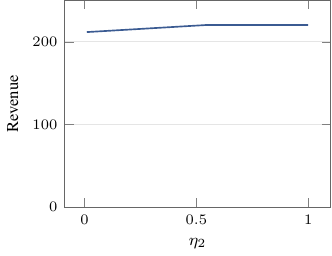}
        \caption{Decode Fairness ($\eta_2$).}
        \label{fig:frontier2}
    \end{subfigure}
    \hfill
    \begin{subfigure}{0.32\textwidth}
        \centering
        \includegraphics[width=\textwidth]{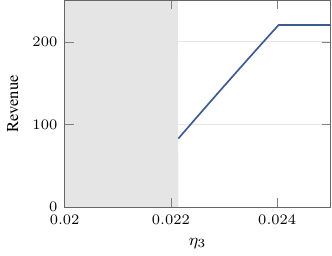}
        \caption{TPOT Limit ($\eta_3$).}
        \label{fig:frontier3}
    \end{subfigure}

    \caption{Pareto frontiers illustrating the shadow price of Service Level Indicators (SLIs).}
    \label{fig:frontiers}
    \vspace{-4mm}
\end{figure}

\begin{itemize}
    \item \textbf{Asymmetric Costs of Fairness (Figs. \ref{fig:frontier1} and \ref{fig:frontier2}):} We observe a stark contrast between the shadow prices of fairness at different stages of the inference pipeline.
    Prefill fairness ($\eta_1$) incurs a steep revenue penalty because the prefill stage acts as the primary bottleneck; imposing rigid class-mix constraints here prevents the scheduler from aligning admissions with the hardware's optimal operating point, leading to a structural mismatch between job arrivals and downstream capacity.
    In contrast, the nearly flat frontier for decode fairness ($\eta_2$) suggests a negligible shadow price.
    This implies that once a request has entered the system, rebalancing its processing speed relative to other classes barely degrades total revenue.
    \item \textbf{The Price of Low TPOT (Fig. \ref{fig:frontier3}):} The shadow price of the TPOT constraint increases significantly as the target latency approaches $0.022\,\text{s}$, a lower bound determined by the solo-decode rate $\gamma$.
    Near this threshold, the feasible region for prefill throughput shrinks rapidly, leading to a substantial decrease in optimal revenue as the system prioritizes meeting the stringent latency requirement over throughput.
\end{itemize}

\paragraph{\textbf{GPU configurations and Revenue}}

We next analyze the sensitivity of the system's performance to the maximum batch size ($B$) and the hyperparameters ($\alpha, \beta, \gamma$).
Figure \ref{fig:gpu_sweeps} illustrates the trade-off between total Revenue (primary objective, left y-axis) and TPOT (latency cost, right y-axis).

\begin{figure}[ht]
\vspace{-3mm}
    \centering
    \includegraphics[width=0.85\textwidth]{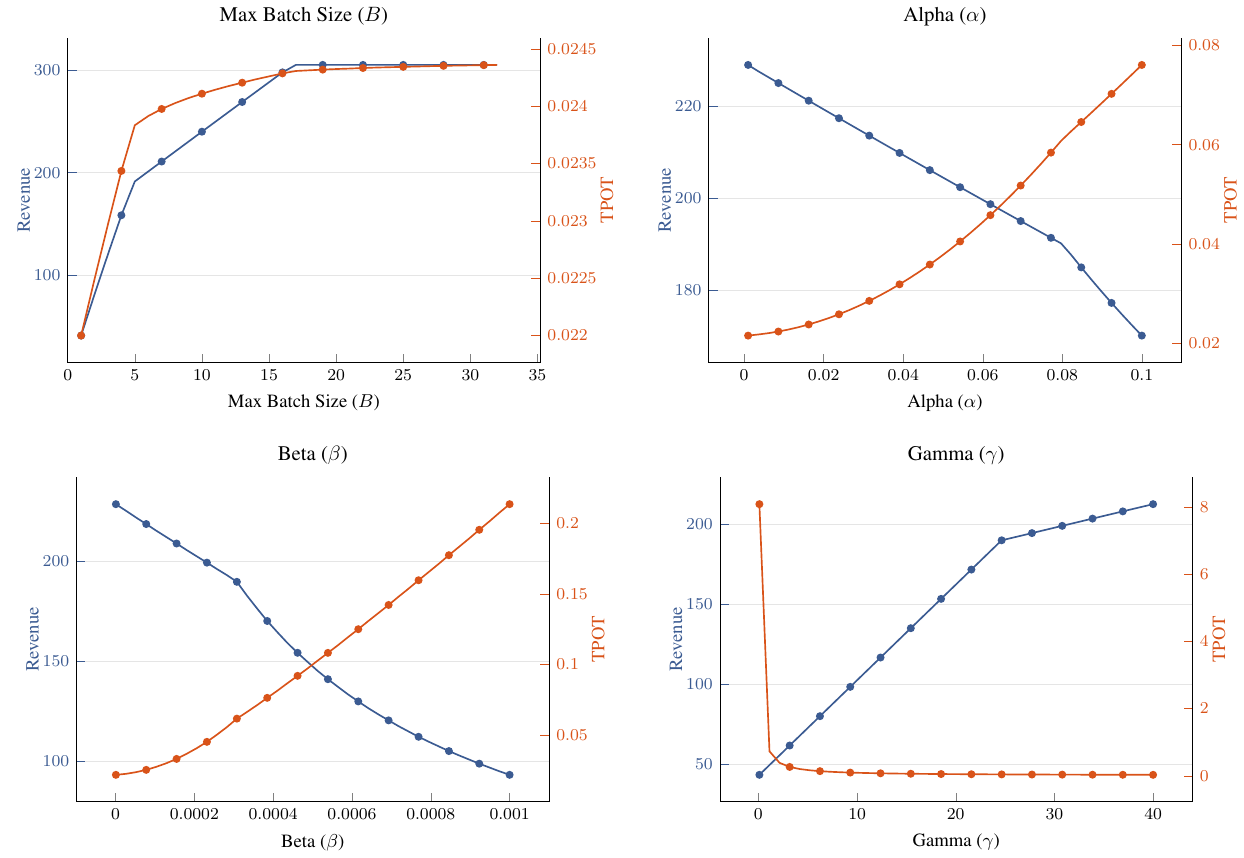} 
    \caption{Parameter sensitivity analysis showing the impact of maximum batch size ($B$), $\alpha$, $\beta$, and $\gamma$ on revenue and TPOT. Blue lines indicate revenue (higher is better); red lines indicate TPOT (lower is better).}
    \label{fig:gpu_sweeps}
\vspace{-3mm}
\end{figure}

The results indicate distinct operational trends.
First, revenue increases with batch size $B$ but saturates around $B=16$, suggesting diminishing returns on memory scaling beyond this point.
Second, the system is highly sensitive to the computational penalty $\beta$; as $\beta$ increases, revenue drops sharply.
Finally, $\gamma$ acts as a strong incentive, positively correlating with higher revenue and lower latency.

Figure \ref{fig:heatmap_tradeoff} visualizes the revenue landscape across the joint configuration space of memory capacity (proxied by $B$) and computational speed (proxied by $\beta$).
This mapping supports infrastructure decisions: operators can overlay GPU prices to identify the highest-ROI configuration, while the gradient reveals whether revenue grows fastest by increasing $B$ or decreasing $\beta$, pinpointing whether memory or compute is the binding bottleneck for upgrades.

\paragraph{\textbf{Optimal Token Pricing}}

We also examine the optimal pricing structure by analyzing the relationship between the prefill price $c_p$ and the decode price $c_d$.
Specifically, if we investigate the common revenue maximization problem subject to a total price constraint $c_p + c_d = k$, the heatmap (Figure~\ref{fig:heatmap_cpcd}) results reveal a striking invariance: regardless of the magnitude of the budget $k$, the revenue-maximizing prices consistently yield a unique, constant ratio $c_p/c_d$.
This indicates that the optimal economic balance between prefill and decode is scale-invariant.
Consequently, for pricing strategy, managers should focus on maintaining this intrinsic cost structure ratio, as the optimal split between prefill and decode prices remains stable even as the total price level varies.

\begin{figure}[ht]
    \centering
    \captionsetup[subfigure]{font=small,justification=centering}
    \begin{subfigure}[b]{0.48\textwidth}
    \centering
    \includegraphics[width=\linewidth]{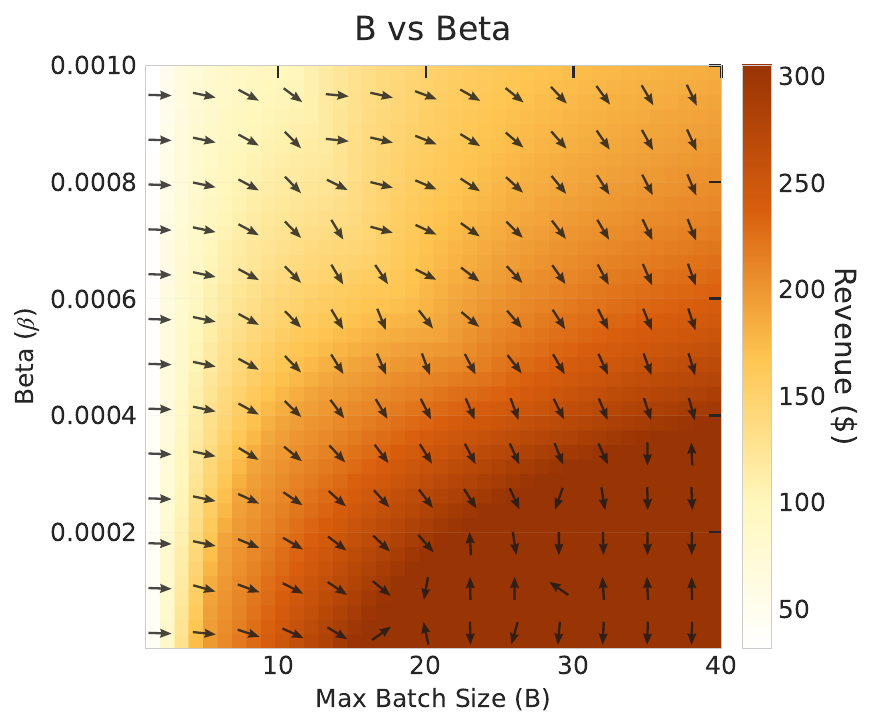}
    \caption{Memory-speed tradeoff ($B$ vs. $\beta$).}
    \label{fig:heatmap_tradeoff}
    \end{subfigure}\hfill
    \begin{subfigure}[b]{0.48\textwidth}
    \centering
    \includegraphics[width=\linewidth]{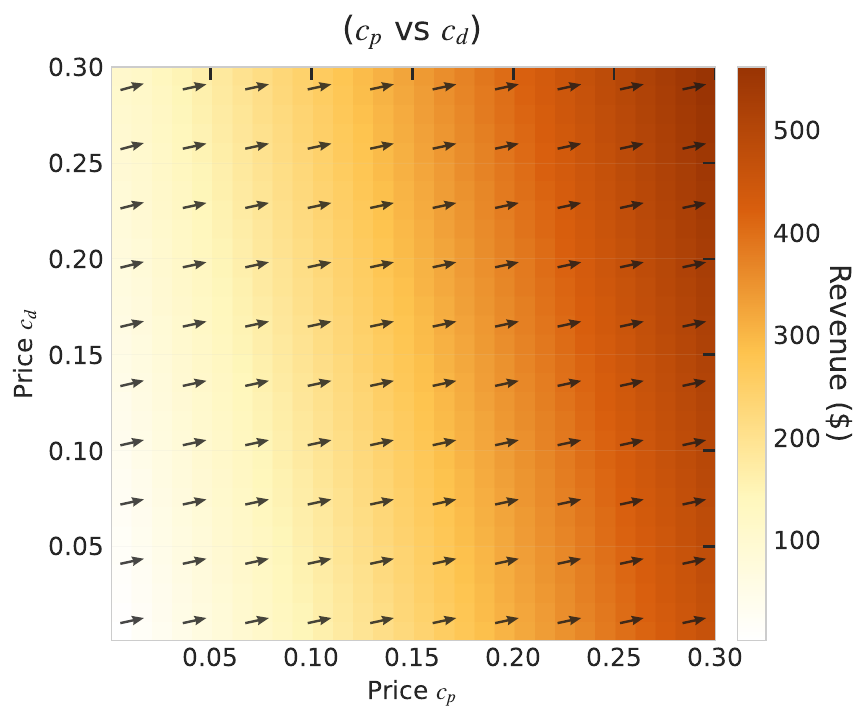}
    \caption{Optimal token pricing ($c_p$ vs. $c_d$).}
    \label{fig:heatmap_cpcd}
    \end{subfigure}
    \caption{Sensitivity analysis of hardware configurations and pricing structures on per-GPU revenue.}
    \label{fig:heatmaps_combined}
    \vspace{-3mm}
    \end{figure}

\section{Conclusions}
\label{sec:conclusion}

Efficient LLM inference at scale hinges on resolving the resource contention between compute-bound prefill and memory-bound decode.
In this work, we connect an empirically grounded iteration-time abstraction with stochastic control to study admission and scheduling in large GPU clusters under token-based revenue objectives and service constraints.
A multiclass many-GPU fluid approximation yields a tractable steady-state linear program that prescribes how to split capacity between mixed and solo modes and how to allocate prefill occupancy across classes.

Building on this planning formulation, we develop a \textit{gate-and-route} control architecture with a prefill admission gate that tracks the fluid occupancy targets and a decode router that keeps downstream capacity work-conserving.
Our analysis establishes a set of structural properties that enable this decomposition, including the existence of an optimal fluid plan with no steady-state decode buffering (Proposition~\ref{prop:qd-zero}) and a corresponding asymptotic optimality guarantee for the proposed policy in the many-GPU limit.

Our numerical evaluation is empirically calibrated and intended to illustrate the mechanisms highlighted by the theory: we calibrate the iteration-time model using real deployments on a modern inference stack, and then run calibrated event-driven simulations to study the resulting control behavior.
In particular, simulations confirm two qualitative predictions from the asymptotic analysis: per-GPU revenue approaches the fluid optimum as the cluster scales, and persistent decode backlogs are avoided under the proposed control.
We further compare against representative heuristic baselines and ablations in the same calibrated setting.

Our results also yield actionable implications for service providers.
First, billing and scheduling objectives need not coincide: while separate charging for prefill and decode tokens is natural for accounting, optimizing the scheduler against a separate objective can encourage overly aggressive prefill admission and shift congestion downstream.
A practical recommendation is therefore to bill by phase if desired, while scheduling against an end-to-end (bundled) completion objective to align incentives with user-perceived performance.
Second, incorporating SLIs as constraints or penalties in the planning problem provides a systematic way to study the revenue implications of fairness and latency requirements; in our setting, enforcing fairness at the prefill stage is typically more revenue-costly than at the decode stage.

Several avenues exist to extend the theoretical depth and practical scope of this work.
First, one can relax the assumption of exponential service times by employing measure-valued processes.
This formulation would accommodate general service distributions, enabling more granular, state-dependent control policies.
Second, to move beyond mean-value analysis, developing diffusion approximations that characterize stochastic variability will provide rigorous guarantees for tail-latency SLIs.
Finally, the model can be generalized to heterogeneous infrastructures, orchestrating inference across clusters composed of diverse GPU generations and distinct agent architectures.

%% file: appendix.tex
\begin{center}
{\Large \textbf{Electronic Companion}}
\end{center}
\vspace{0.2em}

\newcommand{\ecproofheading}[1]{\medskip\noindent\textbf{#1}\par\smallskip}

\section{\textcolor{tc-ins}{Notation and System-to-Model Mapping Tables}}
\label{sec:ec-notation-mapping}

{\color{tc-ins}
This section collects the main notation used in the stochastic model, fluid planning problem, and control policies.
It also provides a mapping from serving-system concepts to the mathematical variables used in the queueing formulation.

\begin{table}[H]
\centering
\scriptsize
\caption{Mapping from LLM serving concepts to model variables.}
\label{tab:system-model-mapping}
\begin{tabularx}{\linewidth}{l X X}
\toprule
Serving-system concept & Modeling representation & Role in the queueing formulation \\
\midrule
GPU cluster & $n$ homogeneous GPU servers & Many-server scale parameter; fluid scaling divides counts by $n$. \\
Request type or workload class & Class $i\in\mathcal I$ with $(P_i,D_i,\lambda_i)$ & Captures heterogeneity in prompt length, decode length, and arrival intensity. \\
Prompt tokens / prefill work & $P_i$ and prefill service rate $\mu_{p,i}=C/(P_i\tau)$ & Determines how much mixed-GPU prefill capacity is consumed by class-$i$ admissions. \\
Generated tokens / decode work & $D_i$ and decode rates $\mu_{m,i},\mu_{s,i}$ & Determines downstream decode load created by completing prefill. \\
Maximum GPU batch size & $B$ & Caps the number of simultaneous decode streams per GPU. \\
Chunked prefill & Chunk size $C$ & One prefill chunk may be processed in a mixed iteration. \\
Decode-only GPU iteration & $\tau_{\mathrm{solo}}$ and $\gamma=1/\tau_{\mathrm{solo}}$ & Determines solo decode speed. \\
Mixed GPU iteration & $\tau_{\mathrm{mix}}(C)=\alpha+\beta C$ & Determines prefill speed and mixed decode speed when prefill and decode share a GPU. \\
GPU running one prefill plus decodes & Mixed mode; prefill mass $x_i$ and mixed decode mass $y_{m,i}$ & Capacity is coupled by $\sum_i y_{m,i}\le (B-1)\sum_i x_i$. \\
GPU running decode only & Solo mode; solo decode mass $y_{s,i}$ & Capacity is $\sum_i y_{s,i}\le B(1-\sum_i x_i)$. \\
Waiting prompt requests & Prefill queue $Q_{p,i}^n$ or fluid $q_{p,i}$ & Source of jobs controlled by the prefill gate. \\
Requests whose prefill is done & Decode queue $Q_{d,i}^n$ or fluid $q_{d,i}$ & Downstream buffer regulated by decode routing and LP planning. \\
Starting a prefill on a GPU & Admission process $U_{p,i}^n$ or $u_{p,i}$ & Control action that consumes mixed capacity and creates future decode work. \\
Assigning decode to a mixed/solo slot & $U_{d,m,i}^n,U_{d,s,i}^n$ or $u_{d,m,i},u_{d,s,i}$ & Decode routing control. \\
Prefill completion changing GPU mode & Mode-switch processes $M_{m\to s,i}^n$ and $M_{s\to m,i}^n$ & Structural coupling: ongoing decodes may change between mixed and solo modes. \\
User cancellation, timeout, or expiration & Abandonment rates $\theta_i$ and processes $B_{p,i}^n,B_{d,i}^n$ & Reduced-form impatience mechanism in the planning model. \\
Completion-based API revenue & Bundled reward $w_i=c_pP_i+c_dD_i$ and objective $R^n$ & Revenue is credited only after full decode completion. \\
Separate token accounting & Separate objective $\tilde R^n$ & Prefill and decode token revenues are credited separately. \\
Operational latency/fairness requirements & SLI constraints or penalties $g(\mathbf x,\mathbf y,\mathbf q)\le\eta$ and $l_k$ & Added to the steady-state LP to form SLI-aware planning targets. \\
Online traffic estimation in trace replay & Rolling estimate $\widehat\lambda_i(t)$ & Replaces fixed $\lambda_i$ in online replanning experiments. \\
\bottomrule
\end{tabularx}
\end{table}
\renewcommand{\arraystretch}{1}
\clearpage

\begin{table}[H]
\centering
\scriptsize
\caption{Summary of main notation.}
\label{tab:notation-summary}
\begin{tabularx}{\linewidth}{l l X}
\toprule
Category & Symbol & Meaning \\
\midrule
Sets and scale & $\mathcal I=\{1,\ldots,I\}$ & Set of request classes. \\
& $i,n,B,C$ & Class index, number of GPUs, maximum decode streams per GPU, and prefill chunk size. \\
\midrule
Workload primitives & $P_i,D_i$ & Representative prompt and decode/output lengths of class $i$. \\
& $\lambda_i,\theta_i$ & Per-GPU class-$i$ arrival rate and patience/abandonment rate. \\
& $c_p,c_d,w_i$ & Prompt/decode token prices and bundled completion reward $w_i=c_pP_i+c_dD_i$. \\
\midrule
Service primitives & $\tau_{\mathrm{solo}},\gamma$ & Decode-only iteration time and solo decode token generation rate. \\
& $\tau_{\mathrm{mix}}(C),\tau$ & Mixed iteration time with a prefill chunk; $\tau$ denotes $\tau_{\mathrm{mix}}(C)$. \\
& $\mu_{p,i},\mu_{m,i},\mu_{s,i}$ & Class-$i$ prefill, mixed decode, and solo decode service rates. \\
\midrule
Stochastic state & $Q_{p,i}^n,X_i^n,Q_{d,i}^n$ & Prefill queue, prefill service, and decode queue counts. \\
& $Y_{m,i}^n,Y_{s,i}^n$ & Class-$i$ decode jobs in mixed and solo service. \\
& $Z_{p,i}^n,Z_{d,i}^n$ & Total class-$i$ prefill-stage and decode-stage contents. \\
\midrule
Cumulative flows & $A_i^n,B_{p,i}^n,B_{d,i}^n$ & Cumulative arrivals and abandonments from prefill/decode queues. \\
& $S_{p,i}^n,S_{d,m,i}^n,S_{d,s,i}^n$ & Cumulative prefill, mixed decode, and solo decode completions. \\
& $U_{p,i}^n,U_{d,m,i}^n,U_{d,s,i}^n$ & Cumulative admissions into prefill, mixed decode, and solo decode. \\
& $M_{s\to m,i}^n,M_{m\to s,i}^n$ & Cumulative decode mode switches induced by GPU-mode changes. \\
\midrule
Fluid and LP variables & $\bar W^n=W^n/n$ & Fluid-scaled version of a stochastic process $W^n$. \\
& $q_{p,i},x_i,q_{d,i},y_{m,i},y_{s,i}$ & Fluid queue and in-service masses; also used as steady-state LP decision variables. \\
& $u_{p,i},u_{d,m,i},u_{d,s,i}$ & Fluid cumulative admission processes. \\
& $s_{p,i},s_{d,m,i},s_{d,s,i},b_{p,i},b_{d,i}$ & Fluid service-completion and abandonment processes. \\
& $(x_i^\star,y_{m,i}^\star,y_{s,i}^\star,q_{p,i}^\star,q_{d,i}^\star)$ & Optimal LP solution used as a planning target. \\
\midrule
Policy and SLI quantities & $\mathcal G_{\mathrm{mix}},\mathcal G_{\mathrm{solo}},M$ & Mixed/solo GPU sets and target number of mixed GPUs. \\
& $\pi^n,\Pi^n$ & A scheduling policy and the set of admissible policies in the $n$-GPU system. \\
& $R^n,\tilde R^n$ & Bundled and separate-charging finite-horizon rewards. \\
& $\eta_1,\eta_2,\eta_3,l_k,\mathcal K,\mathcal F_{\mathcal K}$ & SLI thresholds, penalties, active SLI set, and feasible region. \\
\bottomrule
\end{tabularx}
\end{table}
}
\clearpage


\section{Proof of Proposition~\ref{prop:qd-zero}}


\begin{proof}{Proof of Proposition~\ref{prop:qd-zero}}
We mainly adopt the strategy that given an optimal solution with $q_{d,i} > 0$, we can always construct an optimal solution with $q_{d,i}' = 0$ based on the solution given above.
The proof follows from the following steps:

\textbf{First, we eliminate $q_d$ in the original optimization problem to get a new optimization problem (P$'$) which better characterizes our target.}
Fix a class $i$ with $\theta_i > 0$.
From the second constraint in \eqref{eq:LP-original}, we have:
\[
  \mu_{p,i} x_i - \theta_i q_{d,i}
    = \mu_{m,i} y_{m,i} + \mu_{s,i} y_{s,i},
\]
Considering $q_{d,i} \ge 0$, we obtain the following inequality
\begin{equation}
  \label{eq:decode-ineq}
  \mu_{m,i} y_{m,i} + \mu_{s,i} y_{s,i}
  \;\le\; \mu_{p,i} x_i,
\end{equation}
with
\begin{equation}
  \label{eq:qdi-as-slack}
  q_{d,i}
  \;=\; \frac{\mu_{p,i} x_i - \mu_{m,i} y_{m,i} - \mu_{s,i} y_{s,i}}{\theta_i}
  \;\ge\; 0.
\end{equation}

Here $q_{d,i}$ is exactly the slack variable of \eqref{eq:decode-ineq}.
Without loss of generality, we focus on $\theta_i > 0$.

Define the problem (P$'$):
\begin{equation}
  \label{eq:LP-reduced}
  \begin{aligned}
    \max_{\{x_i,y_{m,i},y_{s,i},q_{p,i}\}_{i=1}^I}\quad
      & \sum_{i=1}^I w_i\bigl(\mu_{m,i} y_{m,i} + \mu_{s,i} y_{s,i}\bigr) \\
    \text{s.t.}\quad
      & \sum_{i=1}^I x_i \;\le\; 1, \\[0.2em]
      & \sum_{i=1}^I y_{m,i} \;\le\; (B-1)\sum_{i=1}^I x_i, \\[0.2em]
      & \sum_{i=1}^I y_{s,i} \;\le\; B\Bigl(1-\sum_{i=1}^I x_i\Bigr), \\[0.2em]
      & \lambda_i - \theta_i q_{p,i} = \mu_{p,i} x_i, \qquad \forall i, \\[0.2em]
      & \mu_{m,i} y_{m,i} + \mu_{s,i} y_{s,i} \;\le\; \mu_{p,i} x_i,
        \qquad \forall i, \\[0.2em]
      & x_i, y_{m,i}, y_{s,i}, q_{p,i} \;\ge\; 0, \qquad \forall i.
  \end{aligned}
\end{equation}
For any feasible solution of \eqref{eq:LP-reduced} we can recover a feasible solution of \eqref{eq:LP-original} by defining $q_{d,i}$ via \eqref{eq:qdi-as-slack}.
Conversely, any feasible solution of \eqref{eq:LP-original} satisfies \eqref{eq:decode-ineq} and thus yields a feasible solution of \eqref{eq:LP-reduced} by dropping $q_{d,i}$.
Hence the two problems are equivalent and share the same optimal value.

In particular, $q_{d,i}=0$ for a given class $i$ is equivalent to the inequality \eqref{eq:decode-ineq} being \emph{tight}:
\[
  q_{d,i} = 0
  \quad\Longleftrightarrow\quad
  \mu_{m,i} y_{m,i} + \mu_{s,i} y_{s,i} = \mu_{p,i} x_i.
\]

\textbf{Second, we characterize the above optimization problem by using the KKT condition:}

Because \eqref{eq:LP-reduced} is a linear program with a nonempty relative interior (Slater point exists: e.g., take $x_i=y_{m,i}=y_{s,i}=0$ and $q_{p,i}=\lambda_i/\theta_i$), the Karush--Kuhn--Tucker (KKT) conditions are necessary and sufficient for optimality.

The Lagrangian of \eqref{eq:LP-reduced} is
\begin{align*}
  \mathcal{L}
  &= \sum_{i=1}^I w_i\bigl(\mu_{m,i} y_{m,i} + \mu_{s,i} y_{s,i}\bigr)
   + \alpha\Bigl(1 - \sum_{i=1}^I x_i\Bigr)
   + \beta\Bigl((B-1)\sum_{i=1}^I x_i - \sum_{i=1}^I y_{m,i}\Bigr) \\
  &\quad
   + \delta\Bigl(B - \sum_{i=1}^I y_{s,i} - B\sum_{i=1}^I x_i\Bigr)
   + \sum_{i=1}^I \phi_i\bigl(\lambda_i - \theta_i q_{p,i} - \mu_{p,i} x_i\bigr) \\
  &\quad
   + \sum_{i=1}^I \eta_i\bigl(\mu_{p,i} x_i - \mu_{m,i} y_{m,i} - \mu_{s,i} y_{s,i}\bigr)
   + \sum_{i=1}^I \bigl(\sigma_{x,i} x_i + \sigma_{m,i} y_{m,i}
                        + \sigma_{s,i} y_{s,i} + \sigma_{p,i} q_{p,i}\bigr).
\end{align*}

The KKT conditions at an optimal primal--dual pair $\bigl(x_i^\star,y_{m,i}^\star,y_{s,i}^\star,q_{p,i}^\star; \alpha^\star,\beta^\star,\delta^\star,\phi_i^\star,\eta_i^\star, \sigma_{x,i}^\star,\sigma_{m,i}^\star,\sigma_{s,i}^\star,\sigma_{p,i}^\star\bigr)$ are:

\paragraph{(i) Stationarity.}
For each $i$,
\begin{align}
  \frac{\partial\mathcal{L}}{\partial y_{m,i}}
  &= w_i\mu_{m,i} - \beta - \eta_i \mu_{m,i} + \sigma_{m,i} = 0, \label{eq:stat-ym} \\
  \frac{\partial\mathcal{L}}{\partial y_{s,i}}
  &= w_i\mu_{s,i} - \delta - \eta_i \mu_{s,i} + \sigma_{s,i} = 0, \label{eq:stat-ys} \\
  \frac{\partial\mathcal{L}}{\partial x_i}
  &= -\alpha + (B-1)\beta - B\delta
     - \phi_i\mu_{p,i} + \eta_i\mu_{p,i} + \sigma_{x,i} = 0, \label{eq:stat-x} \\
  \frac{\partial\mathcal{L}}{\partial q_{p,i}}
  &= -\phi_i\theta_i + \sigma_{p,i} = 0. \label{eq:stat-qp}
\end{align}

\paragraph{(ii) Complementary slackness.}
\begin{align}
  &\alpha\Bigl(1 - \sum_{j=1}^I x_j\Bigr) = 0, \tag{C1} \\
  &\beta\Bigl(\sum_{j=1}^I y_{m,j} - (B-1)\sum_{j=1}^I x_j\Bigr) = 0, \tag{C2} \\
  &\delta\Bigl(\sum_{j=1}^I y_{s,j} + B\sum_{j=1}^I x_j - B\Bigr) = 0, \tag{C3} \\
  &\eta_i\bigl(\mu_{m,i} y_{m,i} + \mu_{s,i} y_{s,i} - \mu_{p,i} x_i\bigr) = 0, \quad \forall i, \tag{C4} \\
  &\sigma_{x,i} x_i = 0,\;
   \sigma_{m,i} y_{m,i} = 0,\;
   \sigma_{s,i} y_{s,i} = 0,\;
   \sigma_{p,i} q_{p,i} = 0, \quad \forall i. \tag{C5}
\end{align}

\paragraph{(iii) Primal feasibility.}
All constraints in \eqref{eq:LP-reduced} hold.

\paragraph{(iv) Dual feasibility.}
\[
  \alpha,\beta,\delta,\eta_i,\sigma_{x,i},\sigma_{m,i},\sigma_{s,i},\sigma_{p,i} \;\ge\; 0,
  \quad \forall i,\qquad
  \phi_i \in \mathbb{R},\ \forall i.
\]

\textbf{Third, starting from the characterization above, we can begin to construct a satisfying optimal solution.}

Let
\[
  (x_i^\star,y_{m,i}^\star,y_{s,i}^\star,q_{p,i}^\star)
\]
be any optimal solution of \eqref{eq:LP-reduced} (hence also of \eqref{eq:LP-original}) and let $(\alpha^\star,\beta^\star,\delta^\star,\phi_i^\star,\eta_i^\star, \dots)$ be a corresponding dual optimal solution satisfying the KKT conditions above.

Suppose that there exists an index $i_0$ such that the following constraint is \emph{slack} at the optimum, i.e.,
\begin{equation}
  \label{eq:slack-i0}
  \mu_{m,i_0} y_{m,i_0}^\star + \mu_{s,i_0} y_{s,i_0}^\star
    < \mu_{p,i_0} x_{i_0}^\star.
\end{equation}
Then, by complementary slackness (C4), we must have
\[
  \eta_{i_0}^\star = 0.
\]

We now show that, under the assumption $\gamma\tau \ge (B-1)/B$, we can \emph{reallocate} prefill and decode occupancy across modes so that:
\begin{itemize}
  \item the global capacity constraints remain feasible,
  \item the per-class decode completion rates $\mu_{m,i} y_{m,i}^\star + \mu_{s,i} y_{s,i}^\star$ remain unchanged for all $i$,
  \item and the slack in \eqref{eq:slack-i0} is reduced, ultimately to $0$,
\end{itemize}
without changing the objective value.
This yields a new optimal solution in which the $i_0$-th decode constraint is tight.
By repeating the same procedure for every class with slack inequality \eqref{eq:slack-i0}, we obtain an optimal solution with
\[
  \mu_{m,i} y_{m,i}^\star + \mu_{s,i} y_{s,i}^\star
  = \mu_{p,i} x_i^\star,\quad \forall i,
\]
which is equivalent to $q_{d,i}^\star = 0$ in the original formulation.

For each $i$ with $\mu_{m,i} y_{m,i}^\star + \mu_{s,i} y_{s,i}^\star < \mu_{p,i} x_i^\star$ (i.e., $q_{d,i}^\star > 0$ in the original variables), define the gap
\[
  \Delta_i := \frac{\mu_{p,i} x_i^\star
                     - \mu_{m,i} y_{m,i}^\star - \mu_{s,i} y_{s,i}^\star}
                    {\mu_{p,i}} \;\ge\; 0.
\]
Consider the modified active prefill and prefill queue
\[
  \tilde x_i := x_i^\star - \Delta_i
  = \frac{\mu_{m,i} y_{m,i}^\star + \mu_{s,i} y_{s,i}^\star}{\mu_{p,i}},
  \qquad
  \tilde q_{p,i} := q_{p,i}^\star + \frac{\mu_{p,i}\Delta_i}{\theta_i}
                  = q_{p,i}^\star + \frac{\mu_{p,i} x_i^\star
                                       - \mu_{m,i} y_{m,i}^\star
                                       - \mu_{s,i} y_{s,i}^\star}{\theta_i}.
\]
(For classes where \eqref{eq:decode-ineq} already holds with equality, set $\Delta_i=0$ and $\tilde x_i = x_i^\star$, $\tilde q_{p,i} = q_{p,i}^\star$.)

By construction,
\begin{align*}
  \mu_{p,i}\tilde x_i
  &= \mu_{p,i}\bigl(x_i^\star - \Delta_i\bigr)
   = \mu_{p,i}x_i^\star
     - \bigl(\mu_{p,i}x_i^\star
             - \mu_{m,i} y_{m,i}^\star - \mu_{s,i} y_{s,i}^\star\bigr) \\
  &= \mu_{m,i} y_{m,i}^\star + \mu_{s,i} y_{s,i}^\star.
\end{align*}
Hence the inequality\eqref{eq:slack-i0} becomes tight:
\[
  \mu_{m,i} y_{m,i}^\star + \mu_{s,i} y_{s,i}^\star
    = \mu_{p,i}\tilde x_i.
\]
Moreover,
\begin{align*}
  \lambda_i - \theta_i \tilde q_{p,i} - \mu_{p,i} \tilde x_i
  &= \lambda_i
     - \theta_i\Bigl(q_{p,i}^\star
                    + \frac{\mu_{p,i}x_i^\star
                            - \mu_{m,i}y_{m,i}^\star
                            - \mu_{s,i}y_{s,i}^\star}{\theta_i}\Bigr)
     - \mu_{p,i}\bigl(x_i^\star - \Delta_i\bigr) \\
  &= (\lambda_i - \theta_i q_{p,i}^\star - \mu_{p,i} x_i^\star)
     + \mu_{m,i} y_{m,i}^\star + \mu_{s,i} y_{s,i}^\star
     - \mu_{m,i} y_{m,i}^\star - \mu_{s,i} y_{s,i}^\star \\
  &= 0,
\end{align*}
so the first constraint remains satisfied.
Nonnegativity holds because $x_i^\star \ge \Delta_i$ is equivalent to $\mu_{p,i} x_i^\star \ge \mu_{m,i} y_{m,i}^\star + \mu_{s,i} y_{s,i}^\star$, which already holds.

At this stage we have decreased $\sum_i x_i$ to $\sum_i \tilde x_i = \sum_i x_i^\star - \sum_i \Delta_i$ while keeping $(y_{m,i},y_{s,i})$ unchanged.
Thus:
\begin{itemize}
  \item The constraint $\sum_i x_i \le 1$ remains feasible;
  \item The solo-decode capacity $\sum_i y_{s,i} \le B\bigl(1-\sum_i x_i\bigr)$ becomes \emph{less} restrictive, since the right-hand side increases when $\sum_i x_i$ decreases;
  \item The mixed-decode capacity
        \[
          \sum_{i=1}^I y_{m,i} \;\le\; (B-1)\sum_{i=1}^I x_i
        \]
        may become \emph{more} restrictive, because the right-hand side decreases as $\sum_i x_i$ decreases.
\end{itemize}
If the mixed-capacity constraint is still satisfied with the new $\tilde x_i$, we are done with this step.
Otherwise, let
\[
  \Delta := \sum_{i=1}^I y_{m,i}^\star - (B-1)\sum_{i=1}^I \tilde x_i > 0
\]
denote the gap of mixed-mode occupancy between the new capacity.

We now show that, thanks to the assumption $\gamma\tau \ge (B-1)/B$, we can reduce the total mixed occupancy by $\Delta$ while increasing solo occupancy accordingly, without violating capacity and without changing decode completion rates.

For each $i$, define perturbations $(\delta y_{m,i},\delta y_{s,i})$ such that
\[
  \sum_{i=1}^I \delta y_{m,i} = -\Delta, \qquad
  \delta y_{m,i} \le 0,\ \forall i.
\]
We choose $\delta y_{s,i}$ to preserve per-class decode completion rates:
\begin{equation}
  \mu_{m,i}\,\delta y_{m,i} + \mu_{s,i}\,\delta y_{s,i} = 0,
  \qquad \forall i.
  \label{eq:per-class-preserve}
\end{equation}
Equation \eqref{eq:per-class-preserve} implies
\[
  \delta y_{s,i}
  = -\frac{\mu_{m,i}}{\mu_{s,i}}\delta y_{m,i}.
\]
From the speed abstraction, we have
\[
  \mu_{m,i} = \frac{1}{D_i\tau}, \qquad
  \mu_{s,i} = \frac{\gamma}{D_i},
  \qquad \Rightarrow\qquad
  \frac{\mu_{m,i}}{\mu_{s,i}} = \frac{1}{\gamma\tau},
\]
which is \emph{independent} of $i$. Hence
\begin{equation}
  \sum_{i=1}^I \delta y_{s,i}
  = -\frac{1}{\gamma\tau}\sum_{i=1}^I \delta y_{m,i}
  = \frac{\Delta}{\gamma\tau}. \label{eq:sum-delta-ys}
\end{equation}
Thus we have decreased total mixed occupancy by $\Delta$ and increased total solo occupancy by $\Delta/(\gamma\tau)$.

\smallskip
\emph{Mixed-decode capacity.}
By construction, the new mixed occupancy is
\[
  \sum_{i=1}^I \bigl(y_{m,i}^\star + \delta y_{m,i}\bigr)
  = \sum_{i=1}^I y_{m,i}^\star - \Delta
  = (B-1)\sum_{i=1}^I \tilde x_i,
\]
so the mixed capacity is now exactly tight and hence feasible.

\smallskip
\emph{Solo-decode capacity.}
Let $X^\star := \sum_{i=1}^I x_i^\star$ and $\tilde X := \sum_{i=1}^I \tilde x_i$.
Then
\[
  \tilde X = X^\star - \sum_{i=1}^I \Delta_i.
\]
Originally, the solo capacity constraint was
\[
  \sum_{i=1}^I y_{s,i}^\star \;\le\; B(1 - X^\star).
\]
Conducting the steps above, with $y_{s,i}$ unchanged and $x_i$ replaced by $\tilde x_i$, we had
\[
  \sum_{i=1}^I y_{s,i}^\star \;\le\; B(1 - X^\star)
  \;\le\; B(1 - \tilde X),
\]
so solo capacity was slack.
After applying $(\delta y_{m,i},\delta y_{s,i})$, the new solo occupancy becomes, using \eqref{eq:sum-delta-ys},
\[
  \sum_{i=1}^I \bigl(y_{s,i}^\star + \delta y_{s,i}\bigr)
  = \sum_{i=1}^I y_{s,i}^\star + \frac{\Delta}{\gamma\tau}.
\]
We want to ensure
\begin{equation}
  \sum_{i=1}^I y_{s,i}^\star + \frac{\Delta}{\gamma\tau}
  \;\le\; B(1 - \tilde X).
  \label{eq:solo-capacity-check}
\end{equation}
Observe that
\begin{align*}
  \Delta
  &= \sum_{i=1}^I y_{m,i}^\star - (B-1)\sum_{i=1}^I \tilde x_i \\
  &= \biggl[\sum_{i=1}^I y_{m,i}^\star - (B-1)\sum_{i=1}^I x_i^\star\biggr]
     + (B-1)\sum_{i=1}^I (x_i^\star - \tilde x_i).
\end{align*}
By feasibility of the original solution,
$\sum_i y_{m,i}^\star \le (B-1)\sum_i x_i^\star$, hence
\[
  \Delta \;\le\; (B-1)\sum_{i=1}^I (x_i^\star - \tilde x_i)
  = (B-1)\sum_{i=1}^I \Delta_i.
\]
Therefore,
\[
  \frac{\Delta}{\gamma\tau}
  \;\le\; \frac{B-1}{\gamma\tau}\sum_{i=1}^I \Delta_i.
\]
Using the assumption $\gamma\tau \ge (B-1)/B$, we obtain
\[
  \frac{B-1}{\gamma\tau} \;\le\; B,
  \qquad\Rightarrow\qquad
  \frac{\Delta}{\gamma\tau}
  \;\le\; B\sum_{i=1}^I \Delta_i.
\]
Finally, note that
\[
  B(1 - \tilde X)
  = B(1 - X^\star) + B\sum_{i=1}^I \Delta_i.
\]
Combining these inequalities,
\begin{align*}
  \sum_{i=1}^I y_{s,i}^\star + \frac{\Delta}{\gamma\tau}
  &\le \sum_{i=1}^I y_{s,i}^\star + B\sum_{i=1}^I \Delta_i \\
  &\le B(1 - X^\star) + B\sum_{i=1}^I \Delta_i \\
  &= B(1 - \tilde X),
\end{align*}
which is exactly \eqref{eq:solo-capacity-check}.
Thus the solo-decode capacity constraint remains feasible after the reallocation.

By \eqref{eq:per-class-preserve}, for each $i$,
\[
  \mu_{m,i}\bigl(y_{m,i}^\star + \delta y_{m,i}\bigr)
  + \mu_{s,i}\bigl(y_{s,i}^\star + \delta y_{s,i}\bigr)
  = \mu_{m,i} y_{m,i}^\star + \mu_{s,i} y_{s,i}^\star.
\]
Hence all class-level decode completion rates are unchanged, and the objective
\[
  \sum_{i=1}^I w_i\bigl(\mu_{m,i} y_{m,i} + \mu_{s,i} y_{s,i}\bigr)
\]
remains the same.
Consequently, we have constructed a new feasible solution with the \emph{same} objective value, but with mixed occupancy reduced to exactly match the capacity available under the updated $\tilde x_i$.

Moreover, for any class $i$ with initial slack in \eqref{eq:decode-ineq}, we now have
\[
  \mu_{m,i} y_{m,i}^\star + \mu_{s,i} y_{s,i}^\star = \mu_{p,i} \tilde x_i,
\]
i.e., its decode inequality is tight.

Applying the above procedure to each class $i$ whose decode inequality \eqref{eq:decode-ineq} is slack at an optimal solution, we obtain another optimal solution of \eqref{eq:LP-reduced} such that
\[
  \mu_{m,i} y_{m,i}^\star + \mu_{s,i} y_{s,i}^\star
    = \mu_{p,i} x_i^\star, \qquad \forall i.
\]
Returning to the original formulation \eqref{eq:LP-original}, this is equivalent (via \eqref{eq:qdi-as-slack}) to
\[
  q_{d,i}^\star = 0, \qquad \forall i.
\]
Thus the steady-state fluid LP admits an optimal solution in which the decode buffer is empty in the fluid limit, completing the proof.
\end{proof}

\section{Proof of Theorem~\ref{thm:fluid-limit}}

\begin{proof}{Proof of Theorem~\ref{thm:fluid-limit}}
Fix $T>0$.
We denote the scaled processes as $\bar W^n:=W^n/n$.
From the model above, we have:
\[
A_i^n(t)=N_{a,i}(\lambda_i n t),\quad
B_{p,i}^n(t)=N_{b_p,i}\Big(\!\int_0^t \theta_{i} Q_{p,i}^n(s)\,ds\Big),\quad
B_{d,i}^n(t)=N_{b_d,i}\Big(\!\int_0^t \theta_{i} Q_{d,i}^n(s)\,ds\Big),
\]
\[
S_{p,i}^n(t)=N_{s_p,i}\Big(\!\int_0^t \mu_{p,i} X_i^n(s)\,ds\Big),\quad
S_{d,m,i}^n(t)=N_{s_m,i}\Big(\!\int_0^t \mu_{m,i} Y_{m,i}^n(s)\,ds\Big),\quad
\]
\[
S_{d,s,i}^n(t)=N_{s_s,i}\Big(\!\int_0^t \mu_{s,i} Y_{s,i}^n(s)\,ds\Big)
\]

First we prove the tightness.
The GPU capacity implies that $0\le X^n\le n$, $0\le Y_m^n\le (B-1)X^n\le (B-1)n$, $0\le Y_s^n\le B(n-X^n)\le Bn$, so the time-changes for $S_{p,i}^n,S_{d,m,i}^n,S_{d,s,i}^n$ are bounded.
Since $Q_{p,i}^n(s)\le Q_{p,i}^n(0)+A_i^n(s)$ and $Q_{d,i}^n(s)\le Q_{d,i}^n(0)+S_{p,i}^n(s)$, we have:
\[
\frac{1}{n}\!\int_0^t \theta_{p,i} Q_{p,i}^n(s)\,ds \;\le\; \theta_{p,i}\Big(t\,\bar Q_{p,i}^n(0)+t\,\bar A_i^n(t)\Big),\quad
\frac{1}{n}\!\int_0^t \theta_{d,i} Q_{d,i}^n(s)\,ds \;\le\; \theta_{d,i}\Big(t\,\bar Q_{d,i}^n(0)+t\,\bar S_{p,i}^n(t)\Big),
\]
hence these time-changes are stochastically bounded on $[0,T]$.
Controls are nondecreasing and satisfy
\[
U_{p,i}^n(t)\le Q_{p,i}^n(0)+A_i^n(t),\qquad
U_{d,m,i}^n(t)+U_{d,s,i}^n(t)\le Q_{d,i}^n(0)+S_{p,i}^n(t),
\]
and structural transfers obey $M_{m\to s}^n\le B S_p^n$, $M_{s\to m}^n\le B U_p^n$.
By the FSLLN for Poisson processes, the family $\{\bar A_i^n\}_n$ is tight.

Then by the random time-change theorem, together with the bounds above, we have the tightness of all other time-changed Poisson coordinates.
Also, by Billingsley's criterion for nondecreasing processes, the full vector $\bar{\mathcal{X}}^n$ is tight in $\mathbb{D}([0,T],\mathbb{R}^d)$ under the $J_1$-topology.

Second, we focus on the subsequence convergence and continuity.
By tightness, we may pick a subsequence $\bar{\mathcal{X}}^n\Rightarrow \bar{\mathcal{X}}$.
By Skorokhod representation theorem, we may assume a.s. convergence u.o.c. on $[0,T]$.
Since all jumps are $\frac{1}{n}$, the limit $\bar{\mathcal{X}}$ is a.s. continuous.

Finally, we prove that the stochastic processes satisfy the fluid model constraints.
The FSLLN gives $\bar A_i^n\to a_i(t)=\lambda_i t$ u.o.c.
For the time-changed Poisson processes, if $L^n(t):=\frac1n\int_0^t \ell^n(s)\,ds \to L(t)$ u.o.c., then $N(L^n(\cdot)n)/n\to L(\cdot)$ u.o.c.
Hence
\[
\bar B_{p,i}^n \To b_{p,i}(t)=\int_0^t \theta_{p,i}\,q_{p,i}(s)\,ds,\quad
\bar B_{d,i}^n \To b_{d,i}(t)=\int_0^t \theta_{d,i}\,q_{d,i}(s)\,ds,
\]
\[
\bar S_{p,i}^n \To s_{p,i}(t)=\int_0^t \mu_{p,i}\,x_i(s)\,ds,\quad
\bar S_{d,m,i}^n \To s_{d,m,i}(t)=\int_0^t \mu_{m,i}\,y_{m,i}(s)\,ds,\quad
\]
\[
\bar S_{d,s,i}^n \To s_{d,s,i}(t)=\int_0^t \mu_{s,i}\,y_{s,i}(s)\,ds.
\]
Taking limits in the flow balances yields \eqref{eq:fluid-qp}–\eqref{eq:fluid-ys} with \eqref{eq:fluid-a}–\eqref{eq:fluid-dd}.

The remaining inequalities are easy to check since the prelimit constraints satisfy:
\[
0\le \bar X^n\le 1,\quad 0\le \bar Y_m^n\le (B-1)\bar X^n,\quad 0\le \bar Y_s^n\le B(1-\bar X^n),
\]
and the inequalities are preserved under the limit. Hence \eqref{eq:fluid-cap-x}–\eqref{eq:fluid-cap-ys} hold.
\end{proof}
\section{Proof of Theorem~\ref{thm:asymptotic_optimality_occ}}
\ecproofheading{Prefill convergence for Theorem~\ref{thm:asymptotic_optimality_occ}.}

We focus on the LP solutions where $q_{p,i}^*>0$.
The proof for $q_{p,i}^*=0$ is straightforward.
\begin{lemma}
    \label{lem:xi-upper}
    There exists time $t_0$, such that for $t \geq t_0$, we have: $x_i(t) \leq x_i^*$ for all $i$.
\end{lemma}

\begin{proof}{Proof of Lemma~\ref{lem:xi-upper}}
    Suppose not, then for any $t_0$, there exists $t_1$, such that $x_i(t_1) > x_i^*$.
    Take $\Tilde{t} = \sup\{\Tilde{t}< t_0: x_i(\Tilde{t}) \leq x_i^*\}$.
    Then it is clear that $x_i(\Tilde{t}) = x_i^*$ and $x_i(t)> x_i^*$ for $t\in (\Tilde{t},t_1]$.

    But by our policy, if $x_i(t)>x_i^*$, we must have $x_i'(t)<0$.
    This means that $x_i(t_1)= x_i(\Tilde{t}) + \int_{\Tilde{t}}^{t_1}dx_i(t)dt < x_i(\Tilde{t}) = x_i^*$, contradiction.
\end{proof}

\begin{lemma}\label{lem:qp-positive-eventually}
    In the fluid system, there exists time $t_i$, such that $q_{p,i}(t_i)>0$.
    And $q_{p,i}(t)>0$ for all $t \geq t_i$ and $i \in I$.
\end{lemma}

\begin{proof}{Proof of Lemma~\ref{lem:qp-positive-eventually}}
    Suppose not, then for some $i$, we have $q_{p,i}(t)=0$ for all $t$.
    In which case, $x_i(t) \leq x_i^*$.
    By the balance equation, we have:
    \begin{align*}
        \lambda_i &= \theta_iq_{p,i}(t) + \mu_{p,i}x_i(t)\\
        &= \mu_{p,i}x_i(t)\\
        &\leq \mu_{p,i}x_i^*
    \end{align*}
    However, by the structure of the LP solution, $\lambda_i = \theta_iq_{p,i}^* + \mu_{p,i}x_i^*$, contradiction.

    If there exists some index $i$ and $\Tilde{t}_i$ such that $q_{p,i}(\Tilde{t}_i) = 0$, we can take a small enough interval $(\Tilde{t}-\epsilon, \Tilde{t})$ such that $q_{p,i}(t)$ is decreasing in the interval and $q_{p,i}(t) < q_{p,i}^*$.
    Since in the interval $q_{p,i}(t)>0$, $x_i(t) = x_i^*$.
    Then by the balance equations, we have:
    \begin{align*}
        \dot q_{p,i}(t) = \theta_i[q_{p,i}^* - q_{p,i}(t)] > 0
    \end{align*}
    contradiction!
    Then $q_{p,i}(t)>0$ for all $t \geq t_i$ and $i \in I$.
\end{proof}

\begin{lemma}\label{lem:prefill-convergence-fluid}
    Under our policy, $q_{p,i}(t) \rightarrow q_{p,i}^*$ and $x_i(t) \rightarrow x_i^*$ as $t \rightarrow \infty$.
\end{lemma}

\begin{proof}{Proof of Lemma~\ref{lem:prefill-convergence-fluid}}
    By the above lemma, we can assume that $q_{p,i}(t)>0$.
    $x_i(t) = x_i^*$ as a result.
    Then the convergence of $x_i(t)$ follows.
    By the balance equation,
    \begin{align*}
        \dot q_{p,i}(t) &= \lambda_i - \mu_{p,i}x_i^* - \theta_i q_{p,i}(t) \\
        &= \theta_i (q_{p,i}^* - q_{p,i}(t))
    \end{align*}
    Solving the above ODE, the convergence of $q_{p,i}(t)$ follows in standard.
\end{proof}

\ecproofheading{Decode convergence for Theorem~\ref{thm:asymptotic_optimality_occ}.}
 
We now show that, under these conditions, the aggregate decode buffer $q_d(t) := \sum_{i=1}^I q_{d,i}(t)$ converges to zero along any fluid trajectory of the gate-and-route policy.

\medskip

\begin{lemma}\label{lem:decode-cap-sat}
Consider the fluid model under the gate-and-route policy with the static GPU partition induced by the LP solution: a fraction \(x^\star := \sum_{i=1}^I x_i^\star\) of GPUs are mixed and a fraction \(1-x^\star\) are solo.
Let
\[
  y_m(t) := \sum_{i=1}^I y_{m,i}(t),\qquad
  y_s(t) := \sum_{i=1}^I y_{s,i}(t).
\]
If at some regular time \(t\) we have \(q_d(t) > 0\), then decode capacity is fully utilized:
\[
  y_m(t) = x^\star (B-1),\qquad
  y_s(t) = (1-x^\star) B.
\]
\end{lemma}

\begin{proof}{Proof of Lemma~\ref{lem:decode-cap-sat}}
By construction of the static planning step, exactly a fraction \(x^\star\) of the GPUs are permanently assigned to the mixed group, and the remaining fraction \(1-x^\star\) to the solo group.
Each mixed GPU holds at most \((B-1)\) decodes, while each solo GPU holds at most \(B\) decodes.
After fluid scaling by \(n\), the maximal aggregate decode occupancies are
\[
  y_m(t) \le x^\star (B-1),\qquad
  y_s(t) \le (1-x^\star) B.
\]

The gate-and-route decode router is work-conserving: as long as the global decode buffer is nonempty (i.e., \(q_d(t) > 0\)), any free decode slot in either group is immediately filled with a waiting job.
Hence no decode slot can be idle whenever \(q_d(t)>0\).
Therefore,
\[
  y_m(t) = x^\star (B-1),\qquad
  y_s(t) = (1-x^\star) B,
\]
whenever \(q_d(t)>0\).
\end{proof}

\medskip

We next introduce a Lyapunov function that measures the total remaining decode work in ``mixed-mode time units.''

\begin{definition}[Weighted decode work]\label{def:Wd}
Define
\[
  W_d(t)
  := \sum_{i=1}^I \frac{q_{d,i}(t) + y_{m,i}(t) + y_{s,i}(t)}{\mu_{m,i}}.
\]
\end{definition}

Here \(1/\mu_{m,i}\) is the mean decode service time of a type-\(i\) job in mixed mode.
Thus \(W_d(t)\) is the total remaining decode work expressed in mixed-mode time units.

We assume the mode speed ratio
\[
  \mu_{s,i} = \gamma\tau\,\mu_{m,i} =: \kappa\,\mu_{m,i}
  \qquad \text{for all } i \in \{1,\dots,I\},
\]
with \(\kappa = \gamma\tau > 0\) constant across classes.

\begin{lemma}\label{lem:Wd-drift}
Along any fluid trajectory, for \(t \ge 0\),
\begin{equation}\label{eq:Wd-drift-general}
  \dot W_d(t)
  = \sum_{i=1}^I \frac{\mu_{p,i}}{\mu_{m,i}}\,x_i(t)
    \;-\;\bigl( y_m(t) + \kappa y_s(t) \bigr)
    \;-\;\sum_{i=1}^I \frac{\theta_i}{\mu_{m,i}}\,q_{d,i}(t).
\end{equation}
\end{lemma}

\begin{proof}{Proof of Lemma~\ref{lem:Wd-drift}}
We use the fluid balance equations and the random time-change representation for arrivals, completions, and abandonments; all fluid coordinates are absolutely continuous, so their derivatives exist.

\emph{Arrivals into decode.}
For class \(i\), prefills complete at rate \(\mu_{p,i} x_i(t)\).
Each completion creates one decode job of that class, with mixed-mode work \(1/\mu_{m,i}\).
Hence the instantaneous inflow into \(W_d\) from class \(i\) is
\[
  \mu_{p,i} x_i(t) \cdot \frac{1}{\mu_{m,i}},
\]
and summing over \(i\) yields the first term on the right-hand side of \eqref{eq:Wd-drift-general}.

\emph{Service in mixed mode.}
For class \(i\), mixed decodes are in service at rate \(\mu_{m,i} y_{m,i}(t)\), each completion removing \(1/\mu_{m,i}\) units of \(W_d\).
Thus the mixed-mode drain is
\[
  \sum_{i=1}^I \mu_{m,i} y_{m,i}(t) \cdot \frac{1}{\mu_{m,i}}
  = \sum_{i=1}^I y_{m,i}(t) = y_m(t).
\]

\emph{Service in solo mode.}
Using \(\mu_{s,i} = \kappa \mu_{m,i}\), solo-mode completions for class \(i\) occur at rate \(\mu_{s,i} y_{s,i}(t)\), each also removing \(1/\mu_{m,i}\) units of work.
Thus
\[
  \sum_{i=1}^I \mu_{s,i} y_{s,i}(t) \cdot \frac{1}{\mu_{m,i}}
  = \sum_{i=1}^I \kappa \mu_{m,i} y_{s,i}(t) \cdot \frac{1}{\mu_{m,i}}
  = \kappa \sum_{i=1}^I y_{s,i}(t) = \kappa y_s(t),
\]
which gives the second term in \eqref{eq:Wd-drift-general}.

\emph{Abandonments.}
Class-\(i\) decode abandonments occur at rate \(\theta_i q_{d,i}(t)\).
Each abandonment removes \(1/\mu_{m,i}\) units of \(W_d\).
Thus the total drain from abandonments is
\[
  \sum_{i=1}^I \theta_i q_{d,i}(t) \cdot \frac{1}{\mu_{m,i}},
\]
which yields the third term in \eqref{eq:Wd-drift-general} with a minus sign.

Combining these three contributions establishes \eqref{eq:Wd-drift-general}.
\end{proof}

\begin{lemma}\label{lem:Wd-drift-positive}
Assume that under the gate-and-route policy, the prefill occupancies satisfy \(x_i(t) \to x_i^\star\) as \(t\to\infty\) for each \(i\), and let \(x^\star := \sum_{i=1}^I x_i^\star\).
Define the per-GPU decode capacity (in mixed-mode time units)
\[
  C^\star := \kappa (1-x^\star) B + x^\star (B-1).
\]
Then there exists \(T_0<\infty\) such that for all \(t\ge T_0\),
\[
  q_d(t) > 0 \;\Longrightarrow\;
  \dot W_d(t)
  \le - \sum_{i=1}^I \frac{\theta_i}{\mu_{m,i}}\,q_{d,i}(t).
\]
\end{lemma}

\begin{proof}{Proof of Lemma~\ref{lem:Wd-drift-positive}}
Since \(x_i(t) \to x_i^\star\), we have
\[
  A(t) := \sum_{i=1}^I \frac{\mu_{p,i}}{\mu_{m,i}}\,x_i(t)
  \;\longrightarrow\;
  A^\star := \sum_{i=1}^I \frac{\mu_{p,i}}{\mu_{m,i}}\,x_i^\star.
\]
We claim that \(A^\star \le C^\star\).
Indeed, by Proposition~\ref{prop:qd-zero} we can choose an LP-optimal solution with \(q_{d,i}^\star=0\) for all \(i\).
Summing the LP decode flow-balance constraints and using \(\mu_{s,i}=\kappa\mu_{m,i}\) gives
\[
  A^\star
  = \sum_{i=1}^I y_{m,i}^\star + \kappa \sum_{i=1}^I y_{s,i}^\star
  \le x^\star(B-1) + \kappa (1-x^\star)B
  = C^\star,
\]
where the inequality is exactly the LP mixed/solo decode capacity constraints.

By Lemma~\ref{lem:xi-upper}, there exists \(t_0\) such that \(x_i(t)\le x_i^\star\) for all \(i\) and all \(t\ge t_0\).
Hence
\[
  A(t) \le A^\star \le C^\star,\qquad t\ge t_0.
\]

By Lemma~\ref{lem:decode-cap-sat}, if \(q_d(t)>0\) then decode capacity is fully utilized and
\(
  y_m(t) = x^\star (B-1),\;
  y_s(t) = (1-x^\star)B.
\)
Substituting into \eqref{eq:Wd-drift-general} gives, for such \(t\),
\[
  \dot W_d(t)
  = A(t)
    - \Bigl[x^\star (B-1) + \kappa (1-x^\star)B\Bigr]
    - \sum_{i=1}^I \frac{\theta_i}{\mu_{m,i}} q_{d,i}(t).
\]

Then by using the inequality $A(t) \le A^\star \le C^\star$, we have:
\[
  \dot W_d(t)
  \le - \sum_{i=1}^I \frac{\theta_i}{\mu_{m,i}} q_{d,i}(t),
\]
as claimed.
\end{proof}

\medskip

We are now ready to show that the decode buffer vanishes in the fluid limit.

\begin{proposition}\label{prop:qd-to-zero}
Assume \(\theta_i>0\) for all \(i\), and that under the gate-and-route policy the prefill occupancies satisfy \(x_i(t)\to x_i^\star\) as \(t\to\infty\), with the associated capacity constant $C^\star := \kappa (1-x^\star) B + x^\star (B-1)$.
Then along any fluid trajectory,
\[
  \lim_{t\to\infty} q_d(t) = 0.
\]
\end{proposition}

\begin{proof}{Proof of Proposition~\ref{prop:qd-to-zero}}
Define
\[
  \underline{\theta}
  := \min_{1\le i\le I} \frac{\theta_i}{\mu_{m,i}} > 0,
\]
using the assumption \(\theta_i>0\) for all \(i\). Then
\[
  \sum_{i=1}^I \frac{\theta_i}{\mu_{m,i}} q_{d,i}(t)
  \;\ge\;
  \underline{\theta} \sum_{i=1}^I q_{d,i}(t)
  = \underline{\theta}\, q_d(t).
\]

Then by Lemma~\ref{lem:Wd-drift-positive}, we have:

\begin{equation}\label{eq:Wd-drift-ineq}
  \dot W_d(t)
  \le - \underline{\theta}\, q_d(t).
\end{equation}
which means that $W_d(t)$ is nonincreasing whenever $q_d(t) \geq 0$.

We now argue by contradiction.
Suppose \(q_d(t)\) does not converge to zero.
Then there exists \(\delta>0\) and a sequence \(t_k\to\infty\) such that \(q_d(t_k) \ge \delta\) for all \(k\).
Because of the Lipschitz property of the fluid processes, there exists \(\eta>0\) and subintervals \([t'_k, t''_k]\) with
\[
  t'_k \le t_k \le t''_k,\qquad
  t''_k - t'_k \ge \eta,
\]

Then for all \(t\in[t'_k,t''_k]\), \eqref{eq:Wd-drift-ineq} gives
\(
  \dot W_d(t) \le -\underline{\theta}\,\delta/2 =: -c<0
\).
Take the integral of the inequality, we have:
\[
  W_d(t''_k) - W_d(t'_k)
  \le -c (t''_k - t'_k)
  \le -c \eta < 0,
\]

which means that in each interval determined by the sequence $t_k$, the decrease of $W_d(t)$ is at least $c \eta$, which contradicts the fact that $W_d(t) \geq 0$.

Therefore, our assumption is false, and we must have
\[
  \lim_{t\to\infty} q_d(t) = 0.
\]
\end{proof}

\medskip

\begin{proof}{Proof of Theorem~\ref{thm:asymptotic_optimality_occ}}

Combining Proposition~\ref{prop:qd-to-zero} with the previously established convergence of the prefill occupancies $x_i(t)\to x_i^\star$ and the fixed decode capacities in Lemma~\ref{lem:decode-cap-sat}, the aggregate decode workload is stable and the decode buffer vanishes.
Although the solo-first router need not reproduce the class-by-mode LP split $(y_{m,i}^\star, y_{s,i}^\star)$, it attains the same aggregate decode completion rate and hence the same fluid reward $R^\star$.
Therefore the per-GPU reward under the gate-and-route policy converges to the LP optimum, which is the conclusion of Theorem~\ref{thm:asymptotic_optimality_occ}.

\end{proof}

\section{Proof of Theorem~\ref{thm:sep-optimality}}

\begin{proof}{Proof of Theorem~\ref{thm:sep-optimality}}
Let $\tilde\pi^{n,\star}$ be the prioritize-and-route policy defined in Section~\ref{sec:separate-charging}, parameterized by an optimal solution of the steady-state fluid LP \eqref{eq:LP-separate}.

Under separate charging, the steady-state objective \eqref{eq:LP-separate} by which the reward only depends on aggregate occupancies $(\sum_i x_i,\sum_i y_{m,i},\sum_i y_{s,i})$.
What is left is to prove that the prioritize-and-route policy reaches the maximal workload of the system.

The prioritize-and-route policy is work-conserving at the prefill side: whenever at least one prefill queue is nonempty, every reserved prefill slot is filled.
If the fluid system is underloaded so that prefill queues drain in the limit, no admissible policy can keep more capacity busy, and the LP optimum is attained on the slack of the prefill flow-balance constraint; the argument below focuses on the binding-capacity case in which the priority rule has nontrivial bite, so we may assume prefill occupancy is full.
The priority index $\phi_i=D_i/P_i$ is proportional to the decode-work generation rate per unit prefill occupancy, since
\[
  \frac{\mu_{p,i}}{\mu_{m,i}}
  \;=\;
  \frac{C/(\tau P_i)}{1/(\tau D_i)}
  \;=\;
  C\,\frac{D_i}{P_i}
  \;=\;
  C\,\phi_i.
\]
Therefore, among all feasible ways to allocate a given amount of prefill occupancy across classes, the static priority rule maximizes the instantaneous inflow of downstream decode workload.
This ensures that the aggregate decode occupancies are maximized whenever sufficient workload is available; if workload is insufficient, no policy can keep more decode capacity busy.

Consequently, the limiting fluid reward under the prioritize-and-route policy achieves the optimal value $\tilde R^\star$ of the steady-state LP \eqref{eq:LP-separate}, which yields the asymptotic optimality of the policy, i.e.
$$
\liminf_{T\to\infty}\;\liminf_{n\to\infty}\; \tilde R_n(T;\tilde \pi^{n,\star}) \;=\; \tilde R^\star.
$$
which is exactly Theorem~\ref{thm:sep-optimality}.
\end{proof}

\section{Proof of Theorem~\ref{thm:slo_optimality}}

In this section, we prove the optimality of the SLI-aware policy.
Since the SLI-aware policy deals with the prefill stage the same as the Gate-and-Route policy, we will only prove the optimality of the SLI-aware policy for the decode stage.
The optimality of the prefill stage follows from the optimality of the Gate-and-Route policy.

\ecproofheading{Decode-occupancy convergence under the SLI-aware policy.}
\label{sec:EC-sli-aware-decode-conv}

We focus here on the convergence of the decode occupancies under the SLI-aware policy.

\begin{proposition}
\label{prop:sli-aware-y-conv}
Fix an optimal solution $(x_i^\star,y_{m,i}^\star,y_{s,i}^\star,q_{p,i}^\star,q_{d,i}^\star)_{i\in\mathcal I}$ of the SLI-aware steady-state program \eqref{eq:slo_objective} with $q_{d,i}^\star=0$ for all $i$.
Let $x^\star:=\sum_{i=1}^I x_i^\star$ and consider the corresponding static planning of the system.
Denote $q_{d,m,i}(t)$ and $q_{d,s,i}(t)$ for the mixed and solo decode buffers under the virtual buffer split of Section~\ref{subsec:sli-aware-policy} respectively.
Then along any fluid trajectory of the SLI-aware policy,
\[
  \lim_{t\to\infty} y_{m,i}(t) = y_{m,i}^\star,
  \qquad
  \lim_{t\to\infty} y_{s,i}(t) = y_{s,i}^\star,
  \qquad
  \lim_{t\to\infty} q_{d,s,i}(t) = 0,
  \qquad
  \lim_{t\to\infty} q_{d,m,i}(t) = 0,
  \qquad
  \qquad \forall i\in\mathcal I.
\]
\end{proposition}

\begin{proof}{Proof of Proposition~\ref{prop:sli-aware-y-conv}}
Denote the total mixed and solo decode capacities as
\[
  C_m := x^\star(B-1),\qquad C_s := (1-x^\star)B.
\]

First we show that the decode buffer converges to zero as $t\to\infty$.
Define the weighted decode work:
\[
  W_m(t) := \sum_{i=1}^I \frac{q_{d,m,i}(t)+y_{m,i}(t)}{\mu_{m,i}},
  \qquad
  W_s(t) := \sum_{i=1}^I \frac{q_{d,s,i}(t)+y_{s,i}(t)}{\mu_{s,i}}.
\]
Take the derivative of the formulas above, we have:
\begin{align}
  \dot W_m(t)
  &= \sum_{i=1}^I \frac{(1-p_{s,i})\,\mu_{p,i}}{\mu_{m,i}}\,x_i(t)
     \;-\; y_m(t)
     \;-\;\sum_{i=1}^I \frac{\theta_i}{\mu_{m,i}}\,q_{d,m,i}(t),
  \label{eq:EC-sli-Wm-drift}\\
  \dot W_s(t)
  &= \sum_{i=1}^I \frac{p_{s,i}\,\mu_{p,i}}{\mu_{s,i}}\,x_i(t)
     \;-\; y_s(t)
     \;-\;\sum_{i=1}^I \frac{\theta_i}{\mu_{s,i}}\,q_{d,s,i}(t).
  \label{eq:EC-sli-Ws-drift}
\end{align}
Here $y_m(t):=\sum_i y_{m,i}(t)$ and $y_s(t):=\sum_i y_{s,i}(t)$.

Whenever $q_{d,m}(t)>0$, the mixed group is work-conserving, so all mixed decode slots are busy and $y_m(t)=C_m$; similarly, $q_{d,s}(t)>0$ implies $y_s(t)=C_s$.
Since $x_i(t)\to x_i^\star$ and the LP feasibility constraints include $\sum_i y_{m,i}^\star\le C_m$ and $\sum_i y_{s,i}^\star\le C_s$, the constants
\[
  A_m^\star := \sum_{i=1}^I \frac{(1-p_{s,i})\,\mu_{p,i}}{\mu_{m,i}}\,x_i^\star,
  \qquad
  A_s^\star := \sum_{i=1}^I \frac{p_{s,i}\,\mu_{p,i}}{\mu_{s,i}}\,x_i^\star
\]
satisfy $A_m^\star\le C_m$ and $A_s^\star\le C_s$, since they are both the admission limit of the policy which should be less than the total decode capacity.
Therefore,
\[
  q_{d,m}(t)>0 \ \Longrightarrow\ \dot W_m(t)\le  - \underline\theta_m\,q_{d,m}(t),
  \qquad
  q_{d,s}(t)>0 \ \Longrightarrow\ \dot W_s(t)\le  - \underline\theta_s\,q_{d,s}(t),
\]
where $\underline\theta_m:=\min_i \theta_i/\mu_{m,i}>0$ and $\underline\theta_s:=\min_i \theta_i/\mu_{s,i}>0$.
The same contradiction argument as in Proposition~\ref{prop:qd-to-zero} implies
\[
  \lim_{t\to\infty} q_{d,m}(t)=0,\qquad \lim_{t\to\infty} q_{d,s}(t)=0.
\]

Now we are ready to show the convergence of the decode occupancies.
By our policy, we can compute the instant admission rates of the mixed and solo decode pools:
\[
  \dot u_{d,m,i}(t) = (1-p_{s,i})\,(\mu_{p,i}\,x_i(t)-\theta_iq_{d,s,i}(t)),\qquad
  \dot u_{d,s,i}(t) = p_{s,i}\,(\mu_{p,i}\,x_i(t)-\theta_iq_{d,m,i}(t))
\]

Take the derivative of the decode occupancies of the balance equations, we have:
\[
  \dot y_{m,i}(t) = (1-p_{s,i})\,\mu_{p,i}\,x_i(t) - \mu_{m,i}\,y_{m,i}(t),
  \qquad
  \dot y_{s,i}(t) = p_{s,i}\,\mu_{p,i}\,x_i(t) - \mu_{s,i}\,y_{s,i}(t),
\]

Since $x_i(t)\to x_i^\star$ and $q_{d,m,i}(t)\to 0$ and $q_{d,s,i}(t)\to 0$ and the LP feasibility condition:
$$\mu_{p,i}x_i^\star=\mu_{m,i}y_{m,i}^\star+\mu_{s,i}y_{s,i}^\star$$
then for any $\epsilon>0$, there exists $T_\epsilon$ such that for all $t\ge T_\epsilon$,
\begin{align*}
   & |\dot \mu_{d,s,i}(t)-\mu_{s,i}y_{s,i}^\star| \le \epsilon \\
   & |\dot \mu_{d,m,i}(t)-\mu_{m,i}y_{m,i}^\star| \le \epsilon 
\end{align*}

Then we have:
\begin{align*}
   & \mu_{m,i}y_{m,i}^* - \epsilon \le \dot y_{s,i}(t) + \mu_{m,i}y_{m,i}(t) \le \mu_{m,i}y_{m,i}^* + \epsilon \\
   & \mu_{s,i}y_{s,i}^* - \epsilon \le \dot y_{s,i}(t) + \mu_{s,i}y_{s,i}(t) \le \mu_{s,i}y_{s,i}^* + \epsilon 
\end{align*}

Solving the above inequalities and let $\epsilon\to 0$ and $t\to\infty$, we have:
\[
  \lim_{t\to\infty} y_{m,i}(t) = y_{m,i}^\star,
  \qquad
  \lim_{t\to\infty} y_{s,i}(t) = y_{s,i}^\star,
  \qquad \forall i\in\mathcal I.
\]
which means that the decode occupancies converge to the LP targets.
\end{proof}

\section{General SLI-aware Extension} \label{sec:sli-aware-policy-general}

In this section, we analyze the general SLI-aware setting under mild regularity on the optimization problem (bounded penalties, etc.), in which case the SLI-aware LP may admit an optimal solution with $q_{d,i}^\star>0$ for some $i$.

\ecproofheading{Policy definition.}
Fix an optimal solution of the SLI-aware LP $(x_i^\star,y_{m,i}^\star,y_{s,i}^\star,q_{p,i}^\star,q_{d,i}^\star)_{i\in\mathcal I}$, allowing $q_{d,i}^\star>0$.
As in Section~\ref{subsec:sli-aware-policy}, we partition GPUs into $G_{\mathrm{mix}}^{(n)}$ and $G_{\mathrm{solo}}^{(n)}$ according to the static planning.
We split the decode buffer into mixed buffer and solo buffer and track the queue lengths $Q_{d,m,i}^n(t)$ and $Q_{d,s,i}^n(t)$ accordingly.
Upon each class-$i$ prefill completion, we route the resulting decode job to the solo buffer with probability $p_{s,i}$ and to the mixed buffer with probability $1-p_{s,i}$.

We set the pool-selection probabilities
\[
p_{s,i}
\;:=\;
\begin{cases}
\dfrac{\mu_{s,i}\, y_{s,i}^\star}{\mu_{m,i}\, y_{m,i}^\star + \mu_{s,i}\, y_{s,i}^\star},
& \text{if }\mu_{m,i} y_{m,i}^\star+\mu_{s,i}y_{s,i}^\star>0,\\[8pt]
1, & \text{otherwise.}
\end{cases}
\]
When $q_{d,i}^\star>0$, the LP flow-balance implies
\[
\mu_{m,i}y_{m,i}^\star+\mu_{s,i}y_{s,i}^\star
\;=\;
\mu_{p,i}x_i^\star-\theta_i q_{d,i}^\star,
\]
so the definition of $p_{s,i}$ above can be rewritten as
\[
p_{s,i}
\;=\;
\frac{\mu_{s,i}y_{s,i}^\star}{\mu_{p,i}x_i^\star-\theta_i q_{d,i}^\star},
\]
i.e., $p_{s,i}$ is the solo share of the \emph{net} decode completion rate. Accordingly, we have the consistent pool-level queue split
\[
q_{d,s,i}^\star := p_{s,i}\,q_{d,i}^\star,\qquad
q_{d,m,i}^\star := (1-p_{s,i})\,q_{d,i}^\star.
\]

We additionally define pool class selection weights (with the convention $0/0:=0$)
\[
\varpi_{m,i}
\;:=\;
\frac{\mu_{m,i}y_{m,i}^\star}{\sum_{j\in\mathcal I}\mu_{m,j}y_{m,j}^\star},
\qquad
\varpi_{s,i}
\;:=\;
\frac{\mu_{s,i}y_{s,i}^\star}{\sum_{j\in\mathcal I}\mu_{s,j}y_{s,j}^\star}.
\]

Different from the baseline SLI-aware policy, in the general SLI-aware policy, we introduce a pool class selection rule within each pool.

\textbf{(i) Solo-pool class selection.}
Whenever a solo decode slot becomes available and $\sum_{i}Q_{d,s,i}^n(t^-)>0$, the decode scheduler selects a class according to the weights $\{\varpi_{s,i}\}_{i\in\mathcal I}$ restricted to currently nonempty solo decode buffer:
\[
\mathbb{P}\!\left(\text{select class }i \,\middle|\, \{Q_{d,s,j}^n(t^-)\}_j\right)
\;=\;
\frac{\varpi_{s,i}\,\mathbf{1}\{Q_{d,s,i}^n(t^-)>0\}}
{\sum_{j\in\mathcal I}\varpi_{s,j}\,\mathbf{1}\{Q_{d,s,j}^n(t^-)>0\}}.
\]
It then routes the head-of-line job of that class to the freed solo slot.
If the solo decode buffer is empty, the slot idles.

\textbf{(ii) Mixed-pool class selection.}
Whenever a \emph{mixed} decode slot becomes available and $\sum_{i}Q_{d,m,i}^n(t^-)>0$, the scheduler selects a class according to $\{\varpi_{m,i}\}$ restricted to nonempty mixed decode buffer (defined analogously) and routes the corresponding head-of-line job to the freed mixed slot.

This policy is work-conserving within each pool and preserves FCFS within each class buffer.
It differs from the baseline SLI-aware router only in the cross-class selection rule \emph{within} the mixed/solo decode pools.
Then we focus on the fluid dynamics induced by the general SLI-aware policy.
We focus on the case where the SLI-aware LP solution satisfies $q_{d,i}^\star>0$ for all $i$ and both decode capacity constraints bind.
The proof for boundary cases (some $q_{d,i}^\star=0$, etc.) is analogous.

First, we introduce some conventions for the general SLI-aware policy.
Let $q_{d,m,i}(t)$ and $q_{d,s,i}(t)$ denote the mixed/solo decode queue contents, and $u_{d,m,i}(t)$ and $u_{d,s,i}(t)$ denote the cumulative admissions into mixed/solo decode service.
Abandonments occur from each pool buffer at rate $\theta_i$, i.e.,
\[
  \dot b_{d,m,i}(t)=\theta_i q_{d,m,i}(t),\qquad \dot b_{d,s,i}(t)=\theta_i q_{d,s,i}(t),
\]
and the total queue length is $q_{d,i}(t)=q_{d,m,i}(t)+q_{d,s,i}(t)$.

\begin{lemma}
\label{lem:sli_qpos_fluid_adm}
Fix a regular time $t$ at which all fluid coordinates are differentiable.
Define the aggregate decode admission and completion rates
\[
\dot u_{d,m}(t):=\sum_{i\in\mathcal I}\dot u_{d,m,i}(t),\quad
\dot u_{d,s}(t):=\sum_{i\in\mathcal I}\dot u_{d,s,i}(t),\quad
\dot d_{d,m}(t):=\sum_{i\in\mathcal I}\mu_{m,i}y_{m,i}(t),\quad
\dot d_{d,s}(t):=\sum_{i\in\mathcal I}\mu_{s,i}y_{s,i}(t).
\]
Define the pool-level decode backlogs
\[
q_{d,m}(t):=\sum_{i\in\mathcal I} q_{d,m,i}(t),\qquad
q_{d,s}(t):=\sum_{i\in\mathcal I} q_{d,s,i}(t).
\]
If the corresponding pool buffer is positive at time $t$, then that pool is work-conserving:
\[
q_{d,m}(t)>0 \ \Longrightarrow\ \dot u_{d,m}(t)=\dot d_{d,m}(t),\qquad
q_{d,s}(t)>0 \ \Longrightarrow\ \dot u_{d,s}(t)=\dot d_{d,s}(t).
\]
Moreover, under the within-pool class selection rule,
\begin{align*}
q_{d,m}(t)>0
&\Longrightarrow\ 
\dot u_{d,m,i}(t)
=
\dot u_{d,m}(t)\,
\frac{\varpi_{m,i}\,\mathbf{1}\{q_{d,m,i}(t)>0\}}
{\sum_{j\in\mathcal I}\varpi_{m,j}\,\mathbf{1}\{q_{d,m,j}(t)>0\}},\\
q_{d,s}(t)>0
&\Longrightarrow\ 
\dot u_{d,s,i}(t)
=
\dot u_{d,s}(t)\,
\frac{\varpi_{s,i}\,\mathbf{1}\{q_{d,s,i}(t)>0\}}
{\sum_{j\in\mathcal I}\varpi_{s,j}\,\mathbf{1}\{q_{d,s,j}(t)>0\}}.
\end{align*}
\end{lemma}

\begin{proof}{Proof of Lemma~\ref{lem:sli_qpos_fluid_adm}}
If $q_{d,\bullet}(t)>0$ for $\bullet\in\{m,s\}$, then the corresponding pool buffer is nonempty and the policy is work-conserving on that pool: If a decode slot is free, then a job will be routed to the slot immediately.
Hence $\dot u_{d,\bullet}(t)=\dot d_{d,\bullet}(t)$.
The instant admission formulas follow directly from the definition of pool selection rule: at the fluid scale, each pool's admissions inherit the fixed weights $\varpi_{\bullet,i}$ restricted to nonempty class buffers.
\end{proof}

\begin{lemma}\label{lem:rank_one_convergence}
  Fix $\mu\in\mathbb{R}^I_{++}$ and $v\in\mathbb{R}^I_{++}$ with $\sum_{i=1}^I v_i=1$. Let
  $D:=\mathrm{diag}(\mu_1,\ldots,\mu_I)$ and $A:=v\mu^\top-D$. Then:
  \begin{enumerate}
      \item[(i)] $0$ is a simple eigenvalue of $A$ and all other eigenvalues have strictly negative real
      parts.
      \item[(ii)] For the ODE $\dot y(t)=Ay(t)$, $\mathbf 1^\top y(t)$ is conserved and
      \[
        y(t)\ \longrightarrow\ y^\infty
        :=\bigl(\mathbf 1^\top y(0)\bigr)\,
        \frac{D^{-1}v}{\mathbf 1^\top D^{-1}v}
        \qquad \text{as } t\to\infty.
      \]
  \end{enumerate}
\end{lemma}
  
\begin{proof}{Proof of Lemma~\ref{lem:rank_one_convergence}}
  (i): By the matrix determinant lemma,
  \[
  \det(\lambda I-A)
  =\det(\lambda I+D-v\mu^\top)
  =\det(\lambda I+D)\Bigl(1-\mu^\top(\lambda I+D)^{-1}v\Bigr).
  \]
  Hence $\lambda$ is an eigenvalue of $A$ if and only if
  \[
  1=\mu^\top(\lambda I+D)^{-1}v=\sum_{i=1}^I v_i\,\frac{\mu_i}{\lambda+\mu_i}.
  \]
  At $\lambda=0$, the right-hand side equals $\sum_i v_i=1$, so $\lambda=0$ is an eigenvalue.
  Its simplicity follows because the derivative of the right-hand side at $\lambda=0$ equals $-\sum_i v_i/\mu_i\neq 0$.
  Next, let $\lambda$ satisfy $\mathrm{Re}(\lambda)\ge 0$ and $\lambda\neq 0$.
  Then for each $i$,
  \[
  \left|\frac{\mu_i}{\lambda+\mu_i}\right|
  <\frac{\mu_i}{\mathrm{Re}(\lambda)+\mu_i}\le 1,
  \]
  with strict inequality because $\lambda\neq 0$.
  Therefore
  \[
  \left|\sum_{i=1}^I v_i\,\frac{\mu_i}{\lambda+\mu_i}\right|
  \le \sum_{i=1}^I v_i \left|\frac{\mu_i}{\lambda+\mu_i}\right|
  <\sum_{i=1}^I v_i=1,
  \]
  so the eigenvalue equation cannot hold.
  This shows that $0$ is the only eigenvalue with nonnegative real part.
  
  (ii): 
  Since $A=v\mu^\top-D$ and $\sum_i v_i=1$, we have
  \[
    \mathbf 1^\top A
    = \mathbf 1^\top (v\mu^\top) - \mathbf 1^\top D
    = (\mathbf 1^\top v)\mu^\top - \mu^\top
    = \mathbf 0^\top.
  \]
  Therefore, along any absolutely continuous solution of $\dot y(t)=Ay(t)$,
  \[
    \frac{d}{dt}\,\mathbf 1^\top y(t)
    = \mathbf 1^\top \dot y(t)
    = \mathbf 1^\top A y(t)
    = 0
  \]
  so the total mass $M:=\mathbf 1^\top y(t)$ is conserved: $M=\mathbf 1^\top y(0)$ for all $t$.
  The equilibrium set is the kernel of $A$.
  If $Ay=0$, then
  \[
    0 = Ay = v(\mu^\top y) - Dy
    \quad\Longleftrightarrow\quad
    Dy = (\mu^\top y)\,v.
  \]
  Let $s:=\mu^\top y$; since $D$ is invertible, this implies
  \[
    y = s\,D^{-1}v.
  \]
  Hence $\mathrm{Ker}(A)=\mathrm{span}\{r\}$ where $r:=D^{-1}v\in\mathbb{R}^I_{++}$.
  Imposing the conserved mass $M=\mathbf 1^\top y$ determines $s$ uniquely:
  \[
    M = \mathbf 1^\top y = s\,\mathbf 1^\top D^{-1}v
    \quad\Longrightarrow\quad
    s = \frac{M}{\mathbf 1^\top D^{-1}v}.
  \]
  Therefore the hyperplane $\{\mathbf 1^\top y=M\}$ contains a \emph{unique} equilibrium, namely
  \[
    y^\infty
    = M\,\frac{D^{-1}v}{\mathbf 1^\top D^{-1}v}
    = \bigl(\mathbf 1^\top y(0)\bigr)\,\frac{D^{-1}v}{\mathbf 1^\top D^{-1}v}.
  \]
  
  By part (i), all eigenvalues of $A$ other than $0$ have strictly negative real parts, and $0$ is a simple eigenvalue.
  Let $r:=D^{-1}v$ be the (right) eigenvector associated with eigenvalue $0$, and note that $\mathbf 1^\top$ is a (left) eigenvector since $\mathbf 1^\top A=\mathbf 0^\top$.
  The corresponding spectral projector is
  \[
    P := \frac{r\,\mathbf 1^\top}{\mathbf 1^\top r}
    = \frac{D^{-1}v\,\mathbf 1^\top}{\mathbf 1^\top D^{-1}v}.
  \]
  Standard linear-systems theory then implies $e^{At}\to P$ as $t\to\infty$.
  Consequently,
  \[
    y(t)=e^{At}y(0)\ \longrightarrow\ Py(0)
    = \bigl(\mathbf 1^\top y(0)\bigr)\,\frac{D^{-1}v}{\mathbf 1^\top D^{-1}v}
    = y^\infty,
  \]
\end{proof}

\begin{proposition}
\label{prop:sli_qpos_global}
Assume $\theta_i>0$ for all $i\in\mathcal I$.
Under the refined SLI-aware policy above, every fluid solution satisfies
\[
\lim_{t\to\infty} x_i(t)=x_i^\star,\qquad
\lim_{t\to\infty} y_{m,i}(t)=y_{m,i}^\star,\qquad
\lim_{t\to\infty} y_{s,i}(t)=y_{s,i}^\star,\qquad
\lim_{t\to\infty} q_{d,i}(t)=q_{d,i}^\star,
\]
for all $i\in\mathcal I$.
\end{proposition}

\begin{proof}{Proof of Proposition~\ref{prop:sli_qpos_global}}
The prefill part mainly follows from the proof of Theorem~\ref{thm:asymptotic_optimality_occ}.
Therefore the prefill argument in the proof of Theorem~\ref{thm:asymptotic_optimality_occ} applies verbatim, and there exists $T_p<\infty$ such that for all $t\ge T_p$ and all $i\in\mathcal I$,
\begin{equation}\label{eq:sli_qpos_prefill_lock_ec}
q_{p,i}(t)>0,\qquad x_i(t)=x_i^\star.
\end{equation}
We analyze the decode dynamics on $[T_p,\infty)$ and shift the time origin so that \eqref{eq:sli_qpos_prefill_lock_ec} holds for all $t\ge 0$.

Then the resulting constant pool arrival rates are
\[
\alpha_{s,i}:=p_{s,i}\mu_{p,i}x_i^\star,\qquad
\alpha_{m,i}:=(1-p_{s,i})\mu_{p,i}x_i^\star,\qquad i\in\mathcal I,
\]
and the pool capacities
\[
y_s^\star:=\sum_{i\in\mathcal I}y_{s,i}^\star,\qquad
y_m^\star:=\sum_{i\in\mathcal I}y_{m,i}^\star,
\]
which coincide with the binding decode capacity constraints at the LP optimum.

The differential version of the balance equations give:
\[
  \dot q_{d,s,i}(t)=\alpha_{s,i}-\dot u_{d,s,i}(t)-\theta_i q_{d,s,i}(t),\qquad
  \dot q_{d,m,i}(t)=\alpha_{m,i}-\dot u_{d,m,i}(t)-\theta_i q_{d,m,i}(t),
\]
together with the corresponding in-service dynamics
\(
\dot y_{s,i}(t)=\dot u_{d,s,i}(t)-\mu_{s,i}y_{s,i}(t)
\)
and
\(
\dot y_{m,i}(t)=\dot u_{d,m,i}(t)-\mu_{m,i}y_{m,i}(t)
\)
under the static planning.

We first show that there exists $T_0<\infty$ such that for all $t\ge T_0$,
\begin{equation}\label{eq:sli_qpos_pool_sat_ec}
q_{d,s}(t)>0,\ q_{d,m}(t)>0,\qquad
y_s(t)=y_s^\star,\ y_m(t)=y_m^\star,
\end{equation}
where $q_{d,s}(t):=\sum_i q_{d,s,i}(t)$ and $q_{d,m}(t):=\sum_i q_{d,m,i}(t)$, and
$y_s(t):=\sum_i y_{s,i}(t)$ and $y_m(t):=\sum_i y_{m,i}(t)$.

Fix the solo pool (the mixed pool is identical).
Using the LP balances and the definitions of $q_{d,s,i}^\star:=p_{s,i}q_{d,i}^\star$, we have for each $i$,
\[
\alpha_{s,i}=\mu_{s,i}y_{s,i}^\star+\theta_i q_{d,s,i}^\star
\quad\Rightarrow\quad
\frac{\alpha_{s,i}}{\mu_{s,i}}\ge y_{s,i}^\star,
\]
with strict inequality for any $i$ with $p_{s,i}>0$ (equivalently $y_{s,i}^\star>0$).
Summing gives
\begin{equation}\label{eq:sli_qpos_overload_sum_s_ec}
\sum_{i\in\mathcal I}\frac{\alpha_{s,i}}{\mu_{s,i}} > \sum_{i\in\mathcal I} y_{s,i}^\star = y_s^\star.
\end{equation}

Suppose, toward a contradiction, that $q_{d,s}(t)=0$ for $t\ge T'$ for some $T'<\infty$.
Then $q_{d,s,i}(t)=0$ for $t\ge T'$ and all $i$.
The solo-buffer balance implies, for $t\ge T'$,
\[
0=\dot q_{d,s,i}(t)=\alpha_{s,i}-\dot u_{d,s,i}(t)-\theta_i q_{d,s,i}(t)=\alpha_{s,i}-\dot u_{d,s,i}(t),
\]
so $\dot u_{d,s,i}(t)=\alpha_{s,i}$ for $t\ge T'$.
Therefore for $t\ge T'$,
\[
\dot y_{s,i}(t)=\dot u_{d,s,i}(t)-\mu_{s,i}y_{s,i}(t)
=\alpha_{s,i}-\mu_{s,i}y_{s,i}(t).
\]
Solving this linear ODE yields $y_{s,i}(t)\to \alpha_{s,i}/\mu_{s,i}$ as $t\to\infty$.
Taking sums and using \eqref{eq:sli_qpos_overload_sum_s_ec} implies
\(
y_s(t)=\sum_i y_{s,i}(t)\to \sum_i \alpha_{s,i}/\mu_{s,i} > y_s^\star
\),
contradicting the capacity constraint $y_s(t)\le y_s^\star$ for all $t$.
Therefore the solo-pool buffer cannot be eventually empty.
Moreover, by \eqref{eq:sli_qpos_overload_sum_s_ec}, there exists $T_s<\infty$ such that
\[
q_{d,s}(t)>0,\qquad \forall\,t\ge T_s.
\]

Applying the same reasoning to the mixed pool yields $T_m<\infty$ such that $q_{d,m}(t)>0$ for all $t\ge T_m$.
Let $T_0:=\max\{T_s,T_m\}$.
Then for all $t\ge T_0$, both pool buffers are positive.
By the work-conserving property of the policy within each pool, $y_s(t)=y_s^\star$ and $y_m(t)=y_m^\star$ for all $t\ge T_0$, proving \eqref{eq:sli_qpos_pool_sat_ec}.

Second, we focus on the convergence of the decode occupancies.
We consider $t\ge T_0$ for which \eqref{eq:sli_qpos_pool_sat_ec} holds.
Since $q_{d,m}(t)>0$, Lemma~\ref{lem:sli_qpos_fluid_adm} gives $\dot u_{d,m}(t)=\dot d_{d,m}(t)$.
Moreover, since $q_{d,m}(t)>0$ for all $t\ge T_0$ by \eqref{eq:sli_qpos_pool_sat_ec}, the selection rule is active for all such $t$.
For any class $i$ with $\alpha_{m,i}>0$, if $q_{d,m,i}(t)=0$ when $t\ge T_0$, then no class-$i$ job can be pulled from the mixed buffer at time $t$, so $\dot u_{d,m,i}(t)=0$ and the mixed-queue balance gives
\[
\dot q_{d,m,i}(t)=\alpha_{m,i}-\dot u_{d,m,i}(t)-\theta_i q_{d,m,i}(t)=\alpha_{m,i}>0.
\]
Therefore $q_{d,m,i}(t)>0$ for $t\ge T_0$ whenever $\alpha_{m,i}>0$.
Hence the class-level admission rate simplifies to the unconditional proportion $\varpi_{m,i}$ for almost every $t\ge T_0$.
\[
\dot u_{d,m,i}(t)=\varpi_{m,i}\,\dot u_{d,m}(t)
=\varpi_{m,i}\sum_{j\in\mathcal I}\mu_{m,j}y_{m,j}(t).
\]
Under the static planning, the mixed in-service masses satisfy
\[
\dot y_{m,i}(t)=\dot u_{d,m,i}(t)-\mu_{m,i}y_{m,i}(t)
=\varpi_{m,i}\sum_{j\in\mathcal I}\mu_{m,j}y_{m,j}(t)-\mu_{m,i}y_{m,i}(t).
\]
This is a linear system $\dot y_m(t)=A_m\,y_m(t)$, where
\[
  A_m \;:=\; \varpi_m\,\mu_m^\top - \mathrm{diag}(\mu_{m,1},\ldots,\mu_{m,I}),
  \qquad
  \varpi_m:=(\varpi_{m,i})_{i\in\mathcal I},\ \mu_m:=(\mu_{m,i})_{i\in\mathcal I}.
\]

Applying Lemma~\ref{lem:rank_one_convergence} to $A_m$ yields $y_{m,i}(t)\to y_{m,i}^\star$ as $t\to\infty$.
The same argument applies to the solo pool which yields $y_{s,i}(t)\to y_{s,i}^\star$.

Finally, we prove the convergence of the decode queues.
For each class $i$, the pool-specific decode queues satisfy, for $t\ge T_0$,
\[
\dot q_{d,s,i}(t)=\alpha_{s,i}-\dot u_{d,s,i}(t)-\theta_i q_{d,s,i}(t),\qquad
\dot q_{d,m,i}(t)=\alpha_{m,i}-\dot u_{d,m,i}(t)-\theta_i q_{d,m,i}(t).
\]
Since both pools are work-conserving for all $t\ge T_0$, the admission rates satisfy
\[
  \dot u_{d,s,i}(t)=\varpi_{s,i}\sum_{j\in\mathcal I}\mu_{s,j}y_{s,j}(t),\qquad
  \dot u_{d,m,i}(t)=\varpi_{m,i}\sum_{j\in\mathcal I}\mu_{m,j}y_{m,j}(t),
\]
for $t\ge T_0$.
Substituting gives, for $t\ge T_0$,
\[
  \dot q_{d,s,i}(t)=\alpha_{s,i}-\varpi_{s,i}\sum_{j\in\mathcal I}\mu_{s,j}y_{s,j}(t)-\theta_i q_{d,s,i}(t),\qquad
  \dot q_{d,m,i}(t)=\alpha_{m,i}-\varpi_{m,i}\sum_{j\in\mathcal I}\mu_{m,j}y_{m,j}(t)-\theta_i q_{d,m,i}(t).
\]
Since $y_{s,i}(t)\to y_{s,i}^\star$ and $y_{m,i}(t)\to y_{m,i}^\star$, the pool-level LP balances satisfy
\[
p_{s,i}\mu_{p,i}x_i^\star-\mu_{s,i}y_{s,i}^\star=\theta_i q_{d,s,i}^\star,\qquad
(1-p_{s,i})\mu_{p,i}x_i^\star-\mu_{m,i}y_{m,i}^\star=\theta_i q_{d,m,i}^\star.
\]
Using the definition of $\varpi_{\bullet,i}$, we also have $\varpi_{s,i}\sum_{j\in\mathcal I}\mu_{s,j}y_{s,j}^\star=\mu_{s,i}y_{s,i}^\star$ and $\varpi_{m,i}\sum_{j\in\mathcal I}\mu_{m,j}y_{m,j}^\star=\mu_{m,i}y_{m,i}^\star$.

Solving the standard linear ODE, we have for any $t\ge T_0$,
\[
q_{d,s,i}(t)
=
e^{-\theta_i (t-T_0)}\,q_{d,s,i}(T_0)
\;+\;
\int_{T_0}^t e^{-\theta_i (t-u)}
\Bigl(\alpha_{s,i}-\varpi_{s,i}\sum_{j\in\mathcal I}\mu_{s,j}y_{s,j}(u)\Bigr)\,du.
\]
Since $y_{s,i}(u)\to y_{s,i}^\star$ for each $i$, the integrand converges to $\alpha_{s,i}-\varpi_{s,i}\sum_{j\in\mathcal I}\mu_{s,j}y_{s,j}^\star =\alpha_{s,i}-\mu_{s,i}y_{s,i}^\star=\theta_i q_{d,s,i}^\star$.
Standard arguments for linear convolution with an exponentially decaying kernel (or directly applying dominated convergence) then yield $q_{d,s,i}(t)\to q_{d,s,i}^\star$ as $t\to\infty$.
The same argument applied to the mixed pool queue gives $q_{d,m,i}(t)\to q_{d,m,i}^\star$.
Summing gives $q_{d,i}(t)=q_{d,m,i}(t)+q_{d,s,i}(t)\to q_{d,i}^\star$.
\end{proof}

\section{Additional Experiments and Empirical Validation}
\label{sec:ec_exponential_fitting}

\textcolor{tc-ins}{This section provides additional empirical validations that complement the main text.
We first document workload heterogeneity, earlier exponential-fit diagnostics, and additional Markovian decoding checks.
We then provide matched synthetic-versus-trace and class-refinement experiments that test the robustness of the Markovian planning abstraction, followed by controlled convergence and ablation experiments for the synthetic setting.}

\paragraph{Computing infrastructure.}
Table~\ref{tab:computing-infrastructure} summarizes the hardware and software environment used for the empirical measurements and numerical experiments in the paper.

\begin{table}[htbp]
\centering
\caption{Computing infrastructure specifications.}
\label{tab:computing-infrastructure}
\scriptsize
\setlength{\tabcolsep}{3pt}
\renewcommand{\arraystretch}{1.05}
\begin{tabularx}{\linewidth}{lX}
\toprule
\textbf{Component} & \textbf{Specification} \\
\midrule
\multicolumn{2}{l}{\textbf{Hardware}} \\
CPU & AMD EPYC 7H12 64-core processor \\
Cores / threads & 2 sockets $\times$ 64 cores $\times$ 2 threads (256 hardware threads) \\
Clock frequency & 1.5--2.6\,GHz \\
Cache & L1d/L1i: 4\,MiB (128$\times$), L2: 64\,MiB (128$\times$), L3: 512\,MiB (32$\times$) \\
GPU & 4$\times$ NVIDIA A100-SXM4-40GB \\
GPU memory & 40\,GB per GPU (160\,GB total) \\
System memory & 1.0\,TiB RAM \\
Swap & 8.0\,GiB \\
\midrule
\multicolumn{2}{l}{\textbf{Software}} \\
Operating system & Ubuntu 22.04.5 LTS (Jammy Jellyfish) \\
Python & Version 3.9.23 \\
CUDA & Version 12.6 \\
NVIDIA driver & Version 560.28.03 \\
\bottomrule
\end{tabularx}
\end{table}

\subsection{\textcolor{tc-ins}{Workload Heterogeneity, Exponential Fits, and Markovian Checks}}
\label{sec:ec-workload-checks}

\paragraph{Workload heterogeneity and task characteristics.}
Table~\ref{tab:workload_heterogeneity} reports the average prompt (input) and output lengths for eight task categories from the Databricks Dolly-15k instruction-tuning dataset~\citep{DatabricksBlog2023DollyV2}.
The data reveals substantial variation across categories.
For example, summarization and information extraction tasks average over 1,000 input tokens with moderate output.
In contrast, creative writing and open QA require fewer than 100 input tokens yet generate longer outputs on average.

\setlength{\tabcolsep}{3pt}
\begin{table}[ht]
\centering
\scriptsize
\renewcommand{\arraystretch}{1.05}
\newcolumntype{Y}{>{\centering\arraybackslash}X}
\begin{tabularx}{\linewidth}{@{}lYYY@{}}
\toprule
Task Category & Avg. $P$ (tokens) & Avg. $D$ (tokens) & Samples \\
\midrule
Brainstorming & 61 & 331 & 1,764 \\
Classification & 123 & 142 & 2,136 \\
Closed QA & 992 & 182 & 1,738 \\
Creative writing & 89 & 915 & 699 \\
General QA & 69 & 572 & 2,187 \\
Information extraction & 1,139 & 273 & 1,462 \\
Open QA & 45 & 293 & 3,739 \\
Summarization & 1,177 & 436 & 1,150 \\
\bottomrule
\end{tabularx}
\caption{Workload heterogeneity: average prompt ($P$) and output ($D$) lengths across task categories from the Databricks Dolly-15k instruction-tuning dataset~\citep{DatabricksBlog2023DollyV2}.}
\label{tab:workload_heterogeneity}
\end{table}

To further illustrate this heterogeneity, Figure~\ref{fig:scatter_info_creative} plots the input and output length pairs for two representative categories: Information Extraction and Creative Writing.
These two classes exhibit nearly opposite characteristics.
Information Extraction tasks cluster in the high-input and low-output region.
Creative Writing tasks concentrate in the low-input and high-output zone.
This contrast confirms that a uniform scheduling policy cannot efficiently balance resources across diverse workloads, which motivates the multiclass framework developed in this paper.

\begin{figure}[htbp]
    \centering
    \includegraphics[width=0.5\linewidth]{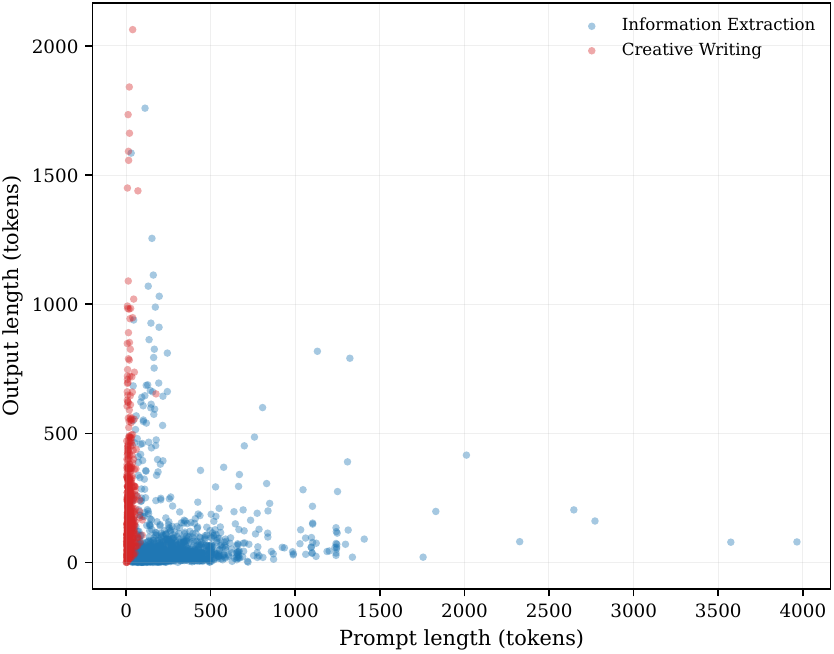}
    \caption{Scatter plot of input versus output lengths for Information Extraction (blue) and Creative Writing (red) in the Databricks Dolly-15k dataset. The two categories occupy nearly opposite regions of the $(P, D)$ space.}
    \label{fig:scatter_info_creative}
\end{figure}

\paragraph{Validation of the exponential distribution assumption.}
Our model assumes exponentially distributed output lengths for each task class.
This assumption directly determines the decode service rate in our stochastic network.
To check whether this approximation is reasonable, we fit exponential distributions to the actual output length data from the Databricks Dolly-15k dataset using Maximum Likelihood Estimation.

Figure~\ref{fig:output_cdf_exponential} compares the empirical CDFs with the fitted exponential CDFs for all eight task categories.
The exponential fit is good for most categories.
Brainstorming, Classification, Closed QA, Information Extraction, Open QA, and Summarization all track closely.
General QA shows reasonable agreement despite some mid-range deviation.
Creative Writing deviates more noticeably in the upper tail.
This is likely because its output distribution is heavier-tailed than a pure exponential, which is consistent with its high mean and variance.

\begin{figure}[htbp]
    \centering
    \includegraphics[width=0.95\linewidth]{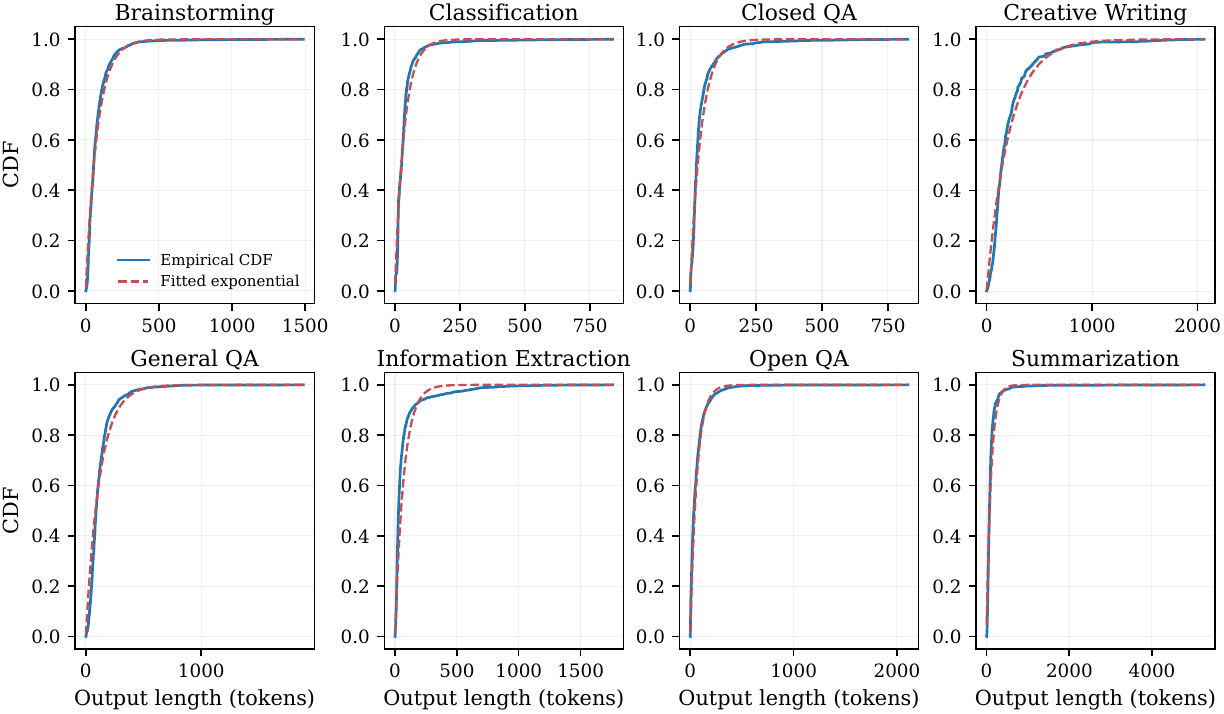}
    \caption{Empirical output length CDFs (solid blue) versus fitted exponential CDFs (dashed red) for the eight task categories in the Databricks Dolly-15k dataset.}
    \label{fig:output_cdf_exponential}
\end{figure}

To complement the visual comparison, Table~\ref{tab:kl_exponential_fit} reports the Kullback-Leibler (KL) divergence between the empirical output-length distribution and its fitted exponential approximation for each category.
Smaller values indicate closer agreement.
Overall, the exponential assumption captures the main characteristics of the distributions well enough to justify its use for analytical tractability.

\begin{table}[htbp]
    \centering
    \scriptsize
    \setlength{\tabcolsep}{3pt}
    \renewcommand{\arraystretch}{1.05}
    \caption{KL divergence between empirical output lengths and fitted exponential distributions across task categories (smaller is better).}
    \label{tab:kl_exponential_fit}
    \begin{tabularx}{\linewidth}{@{}l>{\centering\arraybackslash}X>{\centering\arraybackslash}X>{\centering\arraybackslash}X>{\centering\arraybackslash}X>{\centering\arraybackslash}X>{\centering\arraybackslash}X>{\centering\arraybackslash}X>{\centering\arraybackslash}X@{}}
        \toprule
        \textbf{Category} & \textbf{Open QA} & \textbf{Brainstorming} & \textbf{Closed QA} & \textbf{Summarization} & \textbf{General QA} & \textbf{Classification} & \textbf{\shortstack[c]{Creative\\Writing}} & \textbf{\shortstack[c]{Information\\Extraction}} \\
        \midrule
        \textbf{KL divergence} & 0.0584 & 0.0652 & 0.1245 & 0.1255 & 0.1309 & 0.1491 & 0.1595 & 0.2631 \\
        \bottomrule
    \end{tabularx}
\end{table}

{\color{tc-ins}
\paragraph{Log-survival checks for geometric decay.}
The preceding Dolly analysis checks whether an exponential length approximation is a reasonable distributional simplification within one instruction-tuning dataset.
We next ask a distinct question: whether the token-level stopping mechanism underlying autoregressive decoding is consistent with a Markovian geometric-decay surrogate across additional public query/response datasets.
This Markovian view is consistent with prior work: \citet{zekri2024large} represent autoregressive transformer-based LLMs as Markov chains over a finite state space induced by the finite vocabulary and context window, and the continuous-batching benchmark of \citet{daniel2023continuous} samples generation lengths from an exponential distribution when evaluating LLM serving throughput.
Mechanistically, each decode iteration either emits a regular token and continues or emits an end-of-sequence token and stops, so the stopping event acts like a discrete hazard; a per-class hazard that is approximately stable yields an approximately geometric output length, while a hazard that varies slowly with context still admits a parsimonious exponential-decay surrogate.
For decode lengths, a geometric distribution implies
\[
  \Pr(X\ge k)=(1-p)^{k-1},
\]
so $\log\Pr(X\ge k)$ is linear in $k$.
We therefore compute the empirical survival curve for several public query/response datasets, regress $\log\Pr(X\ge k)$ on $k$ over the main body of the distribution, and report the resulting $R^2$.
This statistic is not an exact goodness-of-fit test, but it measures whether the observed tail follows the exponential-decay shape implied by the Markovian surrogate.
Table~\ref{tab:public-geometric-check} reports representative datasets with high log-survival linearity, and Figure~\ref{fig:public-geometric-qq} provides the corresponding 2-by-2 QQ-plot overview.

\begin{table}[htbp]
\centering
\footnotesize
\caption{Empirical log-survival checks for approximate geometric output-length decay.}
\label{tab:public-geometric-check}
\begin{tabular}{lrr}
\toprule
Dataset & Samples & Log-survival $R^2$ \\
\midrule
Stanford Alpaca & 51,974 & 0.9836 \\
SQuAD & 98,169 & 0.9863 \\
CodeAlpaca 20K & 20,016 & 0.9994 \\
CoQA & 116,630 & 0.9996 \\
\bottomrule
\end{tabular}
\end{table}

\begin{figure}[htbp]
    \centering
    \includegraphics[width=0.80\linewidth]{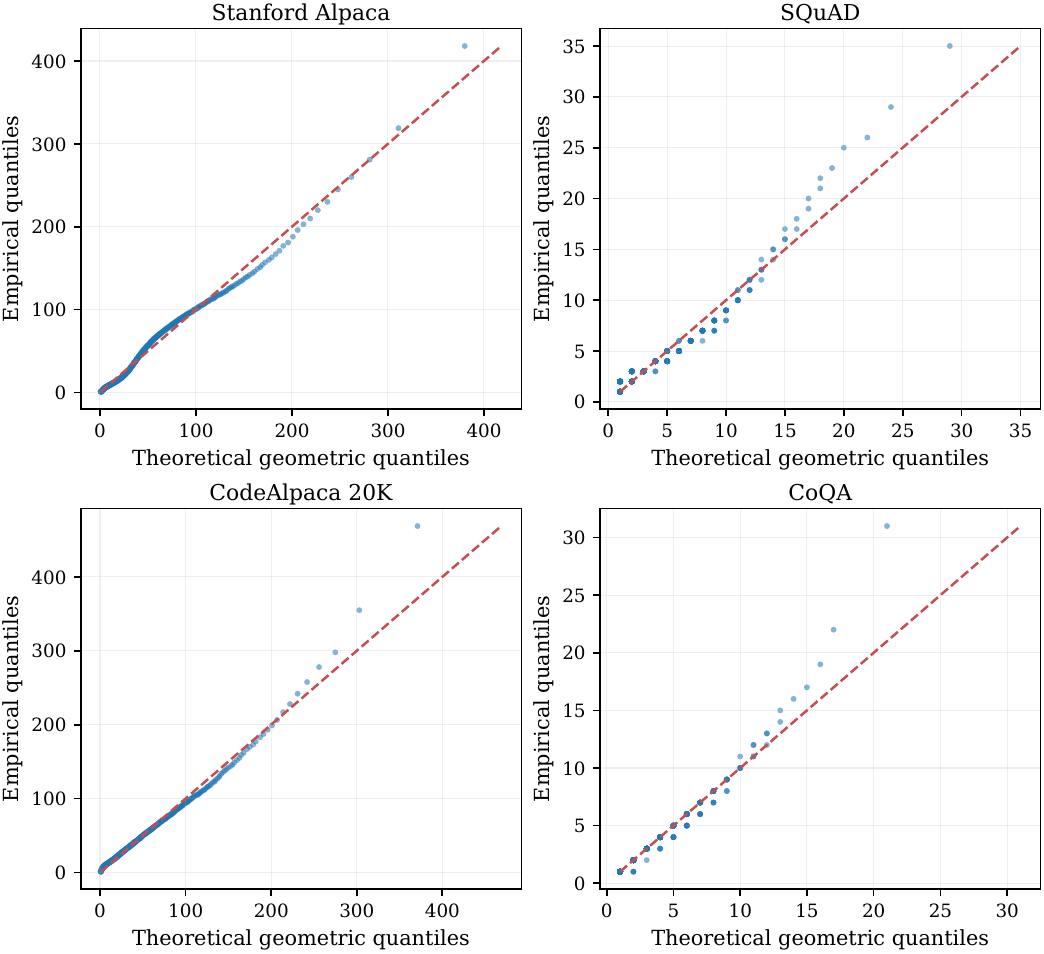}
    \caption{Geometric QQ plots for the four public datasets in Table~\ref{tab:public-geometric-check}. Points closer to the diagonal indicate closer agreement with the fitted geometric output-length distribution.}
    \label{fig:public-geometric-qq}
\end{figure}

These results should be interpreted as supporting evidence for a parsimonious exponential-decay planning model, not as a claim that production traces are exactly geometric.
Indeed, the Azure traces used below are replayed directly: arrivals and token lengths are taken from the data rather than sampled from Poisson or geometric distributions.
}

\subsection{\textcolor{tc-ins}{Matched Synthetic Versus Real Trace}}
\label{sec:ec-matched-synthetic}

{\color{tc-ins}
To quantify the distortion introduced by the Markovian workload abstraction, we construct a matched synthetic workload from the same trace summary.
The synthetic workload has the same two classes, the same class-level mean prompt and output lengths, the same per-GPU arrival-rate calibration, the same simulation horizon, and the same online estimation and replanning parameters as the real-trace replay.
The only change is that arrivals and token lengths are sampled from the corresponding Markovian/exponentialized primitives rather than replayed from the empirical trace.
We repeat this matched comparison across cluster sizes $n \in \{5, 10, 20\}$, holding the per-GPU offered load fixed so that the underlying fluid limit is the same across $n$, which isolates the effect of system scale on the approximation quality.

\begin{table}[htbp]
\centering
\footnotesize
\caption{Matched synthetic-versus-real comparison for the online gate-and-route policy across cluster sizes, holding the per-GPU offered load fixed. At each cluster size the synthetic workload matches the class means, arrival-rate calibration, horizon, and online-control parameters of the real-trace replay. Gap is the matched synthetic revenue rate relative to the real-trace replay.}
\label{tab:matched-synthetic-real}
\begin{tabular}{llrrrrr}
\toprule
$n$ & Scenario & Revenue rate & Gap (\%) & Completion & TTFT mean & TTFT P99 \\
\midrule
\multirow{2}{*}{$5$}  & Real trace replay & 600.70 &          & 0.3841 & 135.88 & 529.24 \\
                      & Matched synthetic & 681.24 & $+13.41$ & 0.4335 & 110.60 & 513.21 \\
\midrule
\multirow{2}{*}{$10$} & Real trace replay & 677.96 &          & 0.4263 & 52.98  & 268.87 \\
                      & Matched synthetic & 697.95 & $+2.95$  & 0.4441 & 47.98  & 238.77 \\
\midrule
\multirow{2}{*}{$20$} & Real trace replay & 617.17 &          & 0.3887 & 31.27  & 129.71 \\
                      & Matched synthetic & 627.39 & $+1.66$  & 0.3953 & 25.69  & 133.94 \\
\bottomrule
\end{tabular}
\end{table}

The matched synthetic workload is slightly optimistic on revenue, but the gap shrinks monotonically as the cluster scales: it is $13.41\%$ at $n=5$, $2.95\%$ at $n=10$, and $1.66\%$ at $n=20$.
Because the real and matched synthetic workloads share the same first-order class means and arrival-rate calibration, their fluid limits coincide, so the finite-system distortion introduced by the Markovian abstraction diminishes as the system scales.
This supports using the Markovian planning model in the many-server regime that motivates our analysis.
}

\subsection{\textcolor{tc-ins}{Benchmark Ranking Across Cluster Scale}}
\label{sec:ec-scale-axis}

{\color{tc-ins}
The trace-driven comparison in Table~\ref{tab:main-real-trace-benchmarks} compresses the Azure interarrival times by a factor of $0.1$ to bring the $10$-GPU system into the congested regime that the policies are designed to manage.
To confirm that the comparison is not an artifact of this particular operating point, we sweep the system along the many-server axis that our analysis targets: we hold the per-GPU offered load fixed and grow the cluster, taking $(n, \text{compression}) \in \{(10, 0.1), (20, 0.05), (40, 0.025)\}$ so that the product of cluster size and compression factor is constant.
The DistServe-style baselines keep their best fixed split at each cluster size.
Table~\ref{tab:ec-scale-axis-benchmarks} reports the same metrics as Table~\ref{tab:main-real-trace-benchmarks} at each cluster size, with online gate-and-route configured for each operating point.

\begin{table*}[htbp]
\centering
\scriptsize
\setlength{\tabcolsep}{2.2pt}
\renewcommand{\arraystretch}{1.08}
\caption{Trace-driven policy comparison across cluster scale on the 2023 Azure replay, holding the per-GPU offered load fixed (cluster size $\times$ compression factor constant).}
\label{tab:ec-scale-axis-benchmarks}
\vspace{0.9em}

\textbf{(a) $n=10$ GPUs, compression $0.1$}
\vspace{0.45em}
\resizebox{\linewidth}{!}{%
\begin{tabular}{lrrrrrrrr}
\toprule
Policy & Revenue rate & Completion rate & TTFT mean & TTFT P95 & TTFT P99 & TPOT mean & TPOT P95 & TPOT P99 \\
\midrule
Online gate-and-route (Ours) & \textbf{677.96} & \textbf{0.4263} & \textbf{52.98} & 248.70 & 268.87 & 0.03254 & 0.03515 & 0.03577 \\
Sarathi-style & 560.34 & 0.3860 & 103.85 & \textbf{198.83} & \textbf{208.37} & 0.02776 & 0.03481 & 0.03541 \\
vLLM-style & 442.94 & 0.3079 & 116.02 & 216.79 & 231.67 & 0.03384 & 0.03538 & 0.03592 \\
\midrule
\textit{DistServe best split (prefill/solo)} & 416.51 & 0.2878 & 125.46 & 225.82 & 236.56 & \textbf{0.01955} & \textbf{0.02021} & \textbf{0.02062} \\
\textit{DistServe best split (mix/solo)} & 418.59 & 0.2899 & 112.04 & 229.62 & 237.17 & 0.02685 & 0.03486 & 0.03552 \\
\bottomrule
\end{tabular}%
}
\vspace{1.25em}

\textbf{(b) $n=20$ GPUs, compression $0.05$}
\vspace{0.45em}
\resizebox{\linewidth}{!}{%
\begin{tabular}{lrrrrrrrr}
\toprule
Policy & Revenue rate & Completion rate & TTFT mean & TTFT P95 & TTFT P99 & TPOT mean & TPOT P95 & TPOT P99 \\
\midrule
Online gate-and-route (Ours) & \textbf{691.67} & \textbf{0.4369} & \textbf{30.68} & \textbf{87.23} & 122.41 & 0.03308 & 0.03595 & 0.03681 \\
Sarathi-style & 560.78 & 0.3866 & 50.60 & 97.58 & \textbf{102.46} & 0.02781 & 0.03484 & 0.03542 \\
vLLM-style & 440.35 & 0.3065 & 56.03 & 105.79 & 115.41 & 0.03384 & 0.03536 & 0.03592 \\
\midrule
\textit{DistServe best split (prefill/solo)} & 414.18 & 0.2866 & 61.46 & 111.79 & 116.94 & \textbf{0.01953} & \textbf{0.02016} & \textbf{0.02050} \\
\textit{DistServe best split (mix/solo)} & 410.67 & 0.2847 & 54.92 & 112.81 & 117.38 & 0.02705 & 0.03488 & 0.03546 \\
\bottomrule
\end{tabular}%
}
\vspace{1.25em}

\textbf{(c) $n=40$ GPUs, compression $0.025$}
\vspace{0.45em}
\resizebox{\linewidth}{!}{%
\begin{tabular}{lrrrrrrrr}
\toprule
Policy & Revenue rate & Completion rate & TTFT mean & TTFT P95 & TTFT P99 & TPOT mean & TPOT P95 & TPOT P99 \\
\midrule
Online gate-and-route (Ours) & \textbf{679.50} & \textbf{0.4329} & \textbf{14.69} & 58.68 & 63.03 & 0.03291 & 0.03570 & 0.03639 \\
Sarathi-style & 551.95 & 0.3818 & 24.17 & \textbf{47.80} & \textbf{50.58} & 0.02789 & 0.03487 & 0.03544 \\
vLLM-style & 437.85 & 0.3048 & 26.06 & 52.23 & 56.76 & 0.03382 & 0.03534 & 0.03590 \\
\midrule
\textit{DistServe best split (prefill/solo)} & 408.78 & 0.2834 & 29.39 & 54.31 & 57.43 & \textbf{0.01952} & \textbf{0.02013} & \textbf{0.02058} \\
\textit{DistServe best split (mix/solo)} & 401.48 & 0.2776 & 26.66 & 55.24 & 58.56 & 0.02736 & 0.03493 & 0.03547 \\
\bottomrule
\end{tabular}%
}
\end{table*}

Across all three cluster sizes, online gate-and-route attains the highest revenue rate, leading the strongest baseline by $21.0\%$, $23.3\%$, and $23.1\%$ at $n=10$, $20$, and $40$ respectively.
The revenue lead is therefore stable as the system scales up at fixed per-GPU intensity, rather than being specific to the $0.1$ compression.
As in the main text, this revenue lead coexists with the expected tail tradeoffs: the DistServe-style rows retain the lowest TPOT by isolating decode, and the simpler work-conserving rules can win individual TTFT tail cells, while online gate-and-route consistently leads on revenue, completion, and mean TTFT.
}

\subsection{\textcolor{tc-ins}{Effect of Finer Workload Classification}}
\label{sec:ec-kmeans-classification}

{\color{tc-ins}
The native code/conversation labels are useful but imperfect class definitions.
Figure~\ref{fig:azure-code-conv-length-density} plots the empirical prompt and output lengths for the two native trace labels.
The code trace is diffuse, while the conversation trace contains several visually distinct regions in the $(P,D)$ space.
This pattern suggests that a single conversation class mixes requests with materially different prefill and downstream decode requirements.
Since the LP uses class-level means and arrival rates to decide prefill admission and mixed/solo allocation, such coarse labels can blur the true resource tradeoff faced by the controller.

\begin{figure}[htbp]
    \centering
    \includegraphics[width=0.95\linewidth]{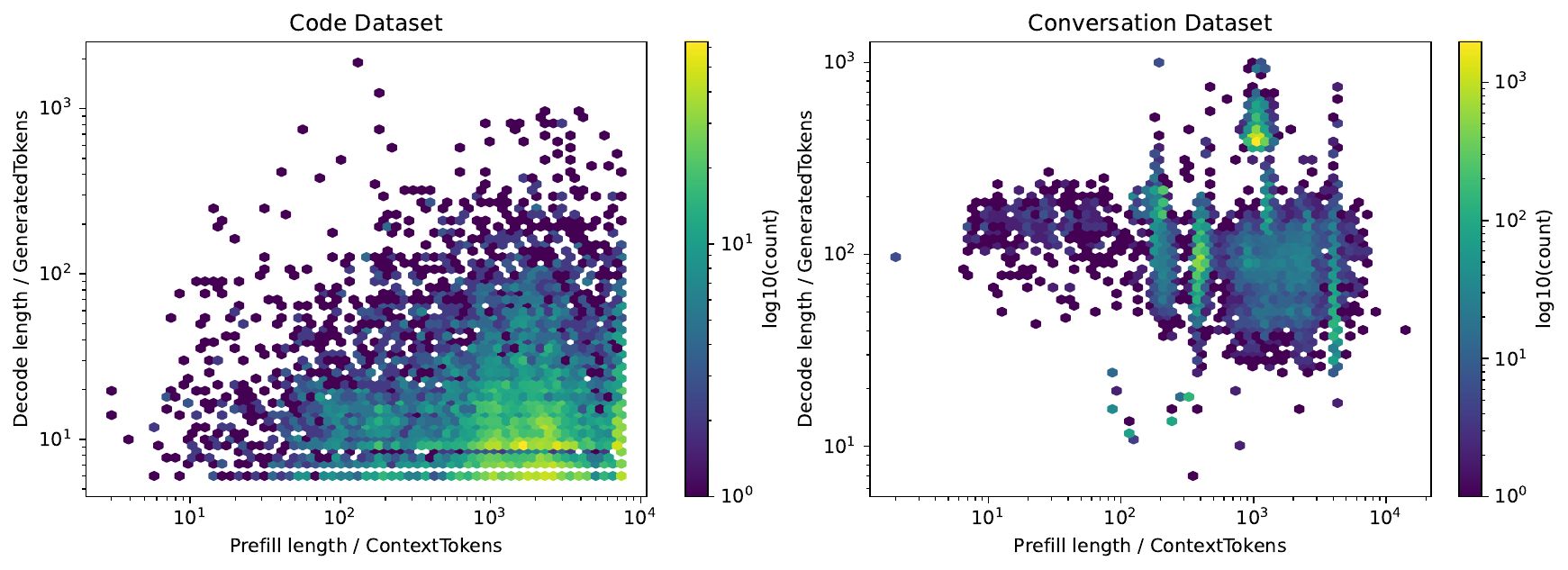}
    \caption{Prompt/output length density for the native code and conversation trace labels. The conversation class contains multiple visually distinct regions, indicating that the native label is an imperfect abstraction of workload heterogeneity.}
    \label{fig:azure-code-conv-length-density}
\end{figure}

As a diagnostic sanity check, we test whether more informative class definitions improve control.
We split the conversation trace using $k$-means clustering on prompt and output lengths and assume the resulting subclass label is available to the scheduler.
This creates $k+1$ classes: one code class and $k$ conversation subclasses.
All other simulator components, including trace replay, latency model, pricing, and online replanning logic, are unchanged.

\begin{table}[htbp]
\centering
\scriptsize
\setlength{\tabcolsep}{3pt}
\caption{Effect of refining the conversation class using $k$-means clustering.}
\label{tab:kmeans-classification}
\resizebox{\linewidth}{!}{%
\begin{tabular}{lrrrrrrrr}
\toprule
Policy & Revenue rate & Completion rate & TTFT mean & TTFT P95 & TTFT P99 & TPOT mean & TPOT P95 & TPOT P99 \\
\midrule
Online gate-and-route (Ours) & 677.96 & 0.4263 & 52.98 & 248.70 & 268.87 & 0.032536 & 0.035145 & 0.035772 \\
Online gate-and-route (Ours, conv $k=2$) & 690.61 & 0.4555 & 79.48 & 202.80 & 245.81 & 0.033340 & 0.035903 & 0.036972 \\
Online gate-and-route (Ours, conv $k=3$) & 711.43 & 0.4680 & 100.29 & 167.29 & 206.04 & 0.033171 & 0.037194 & 0.038177 \\
\bottomrule
\end{tabular}}
\end{table}

Table~\ref{tab:kmeans-classification} shows that, in this diagnostic comparison, refining the conversation class raises the revenue rate from $677.96$ under the native two-class definition to $690.61$ with two conversation subclasses and to $711.43$ with three conversation subclasses.
The gain is not just a mechanical increase in the number of labels; it is consistent with giving the LP more accurate class-level summaries of how admitted prefill work turns into future decode load.
With the refined classes, the online controller can better distinguish requests that are similar in their native label but different in their prefill/decode composition, leading to a better admission and mixed/solo balance.
The improvement in revenue is accompanied by a higher completion rate, while TPOT remains in the same range as the native two-class baseline.
}

\subsection{\textcolor{tc-ins}{Convergence Analysis}}
\label{subsec:convergence}

{\color{tc-ins}
This subsection reports the controlled synthetic experiment used to test convergence and mechanism-level tradeoffs.
The hardware-side parameters are kept the same as in the calibration of Section~\ref{subsec:calibration}.
In contrast, the class definitions and traffic primitives below are synthetic inputs rather than trace-calibrated quantities.

We define two synthetic workload classes:
\begin{itemize}
    \item \textbf{Class 0 (Decode-Heavy):} $P_0=300, D_0=1000$, representing tasks like code generation.
    \item \textbf{Class 1 (Prefill-Heavy):} $P_1=3000, D_1=400$, representing context-heavy analysis or summarization.
\end{itemize}
Arrival rates are symmetric with $\lambda = [0.5,0.5]$, and abandonment rates are $\theta=[0.1,0.1]$.
We use the separate charging scheme with prices $c_p=0.1$ and $c_d=0.2$.
We simulate the system across varying scales \(n\in\{5,20,50,200,500\}\), with five random seeds per configuration.

\paragraph{Revenue and Queue Convergence.}
We first validate the asymptotic optimality of the standard gate-and-route policy.
Figure~\ref{fig:rev_queue_conv} reports per-GPU revenue and queue lengths as the system scale \(n\) increases.
As predicted by the fluid limit, the average per-GPU revenue converges to the optimal value \(R^\star\) from the steady-state LP, while the error bands narrow as \(n\) grows.
The normalized queue lengths for Class 0 and Class 1 also stabilize near their fluid targets.
The prefill gate therefore regulates admission so that stochastic queue lengths track the fluid-optimal backlog needed to maintain utilization, while the decode buffer remains negligible.

\begin{figure}[htbp!]
    \centering
    \captionsetup[subfigure]{justification=centering,font=small}
    \begin{subfigure}[b]{0.44\textwidth}
        \centering
        \includegraphics[width=\linewidth, height=7.3cm]{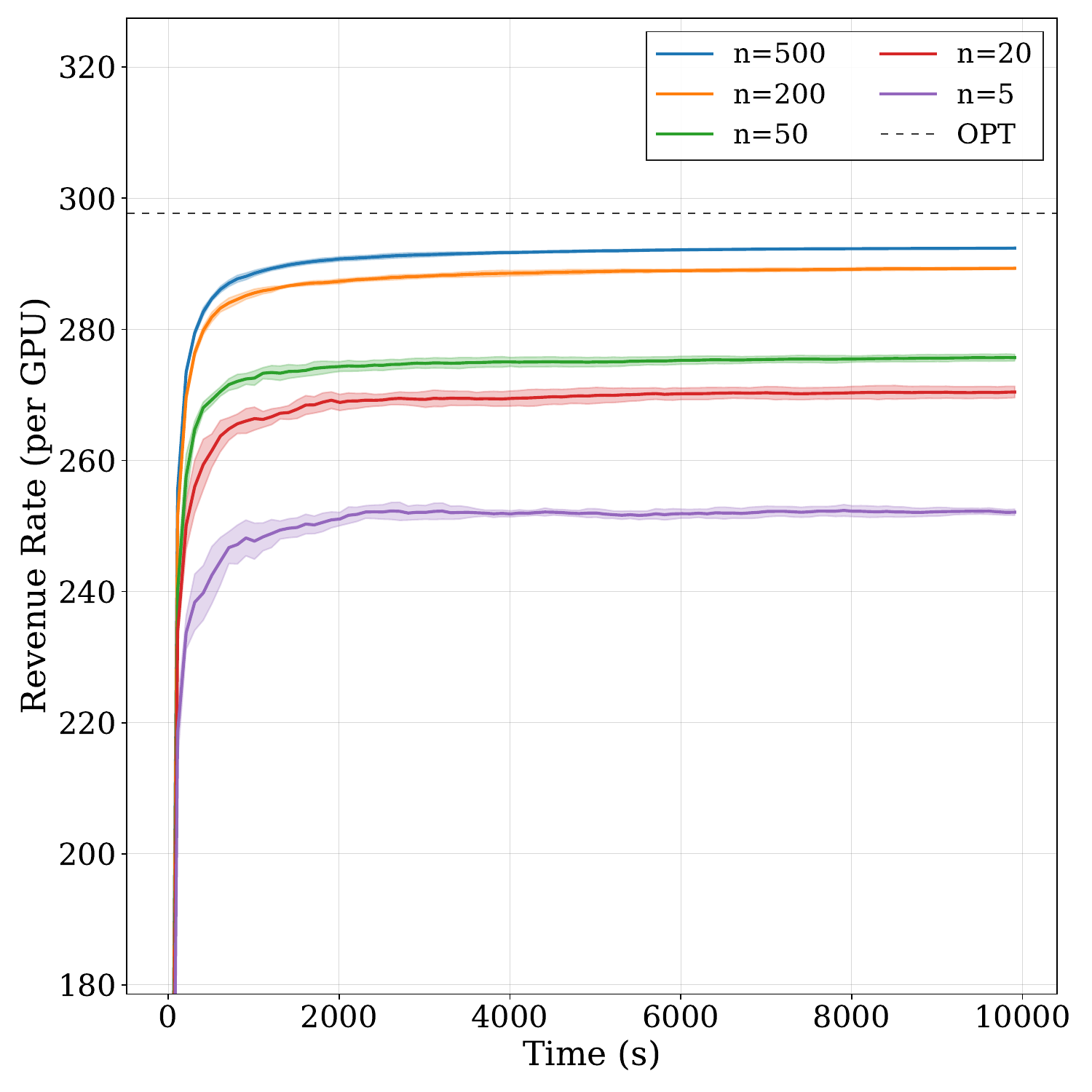}
        \caption{Convergence of per-GPU revenue.}
        \label{fig:rev_conv}
    \end{subfigure}
    \hfill
    \begin{minipage}[b]{0.54\textwidth}
        \begin{subfigure}[b]{\linewidth}
            \centering
            \includegraphics[width=\linewidth, height=3.2cm]{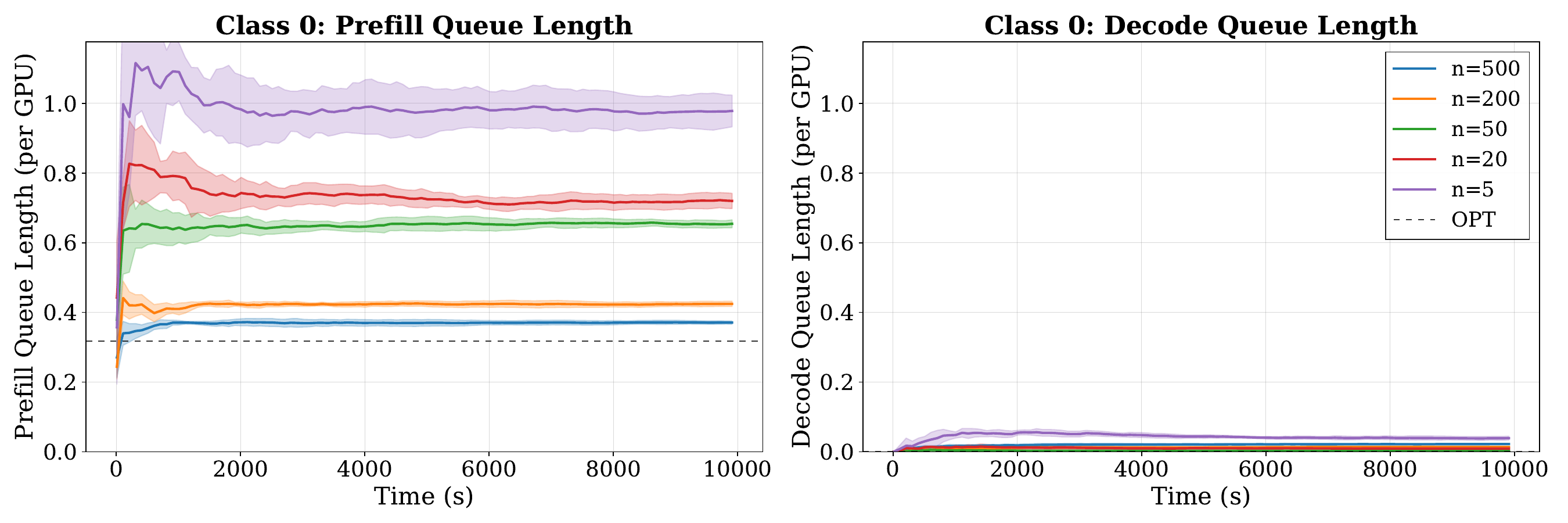}
            \caption{Convergence of average queue lengths for Class 0.}
            \label{fig:queue_class0}
        \end{subfigure}

        \vspace{0.4cm}

        \begin{subfigure}[b]{\linewidth}
            \centering
            \includegraphics[width=\linewidth, height=3.2cm]{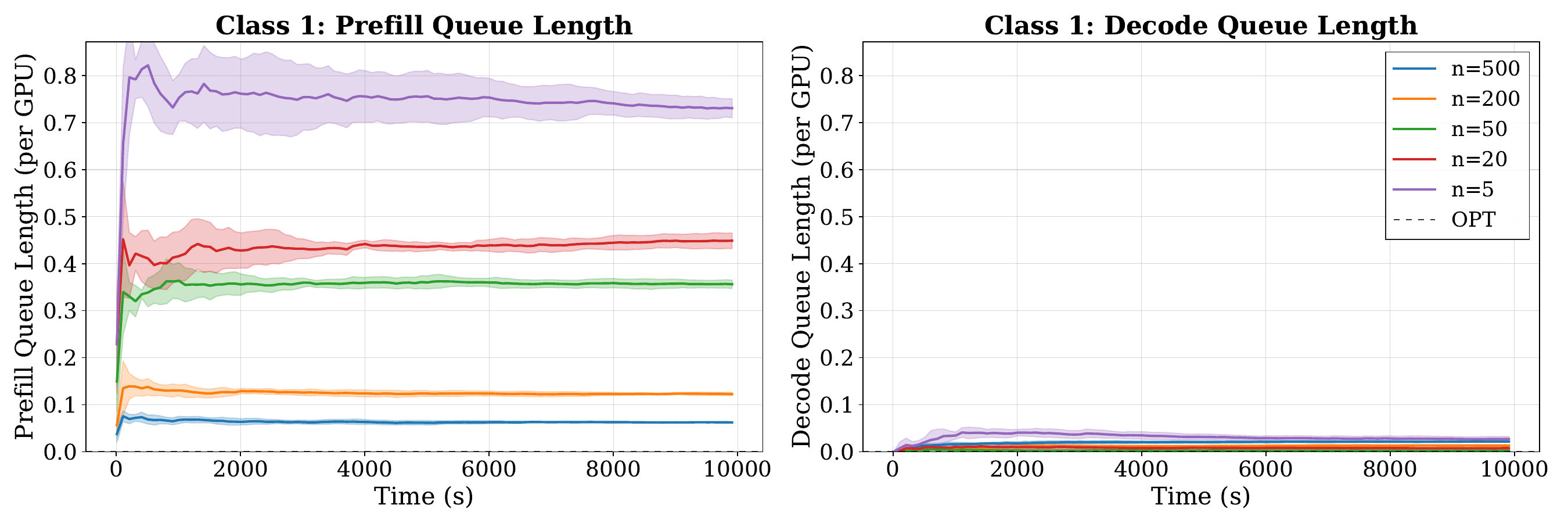}
            \caption{Convergence of average queue lengths for Class 1.}
            \label{fig:queue_class1}
        \end{subfigure}
    \end{minipage}

    \caption{Asymptotic convergence of revenue and queue lengths under the gate-and-route policy.}
    \label{fig:rev_queue_conv}
\end{figure}

\paragraph{Occupancy Convergence Analysis.}
While aggregate revenue converges, a deeper inspection shows a discrepancy in the convergence behavior of specific resource occupancies under the standard policy.
Figure~\ref{fig:occ_standard} plots the fluid-scaled prefill occupancy \(X_i^{(n)}/n\) and decode occupancy \(Y_i^{(n)}/n\).
The prefill occupancy converges tightly to the LP solution \(x_i^\star\), as expected because the prefill gate explicitly targets these values.
The decode occupancy, however, exhibits persistent variance and deviates from the specific class-wise breakdown \(y_{m,i}^\star+y_{s,i}^\star\) predicted by the LP, even though the total reward is optimal.
This occurs because the standard router prioritizes any available solo slot to maximize throughput but is indifferent to which class occupies that slot.
Consequently, the decode class mix fluctuates while still satisfying aggregate capacity constraints.

\begin{figure}[htbp!]
    \centering
    \begin{minipage}[b]{0.98\textwidth}
        \includegraphics[width=\linewidth, height=4cm]{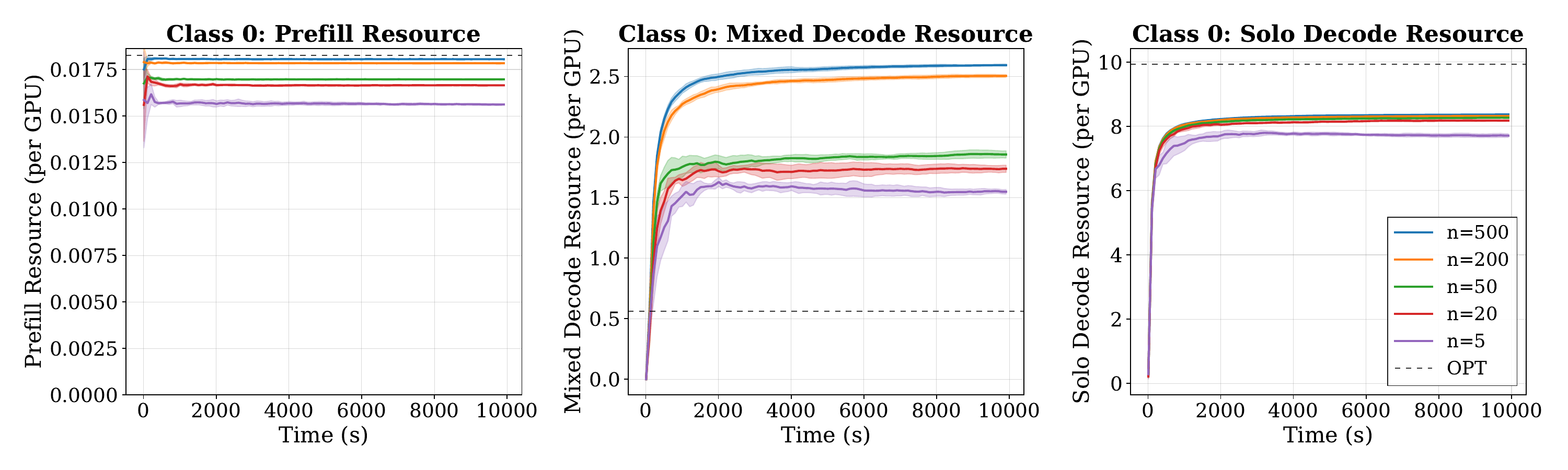}

        \vspace{0.3cm}

        \includegraphics[width=\linewidth, height=4cm]{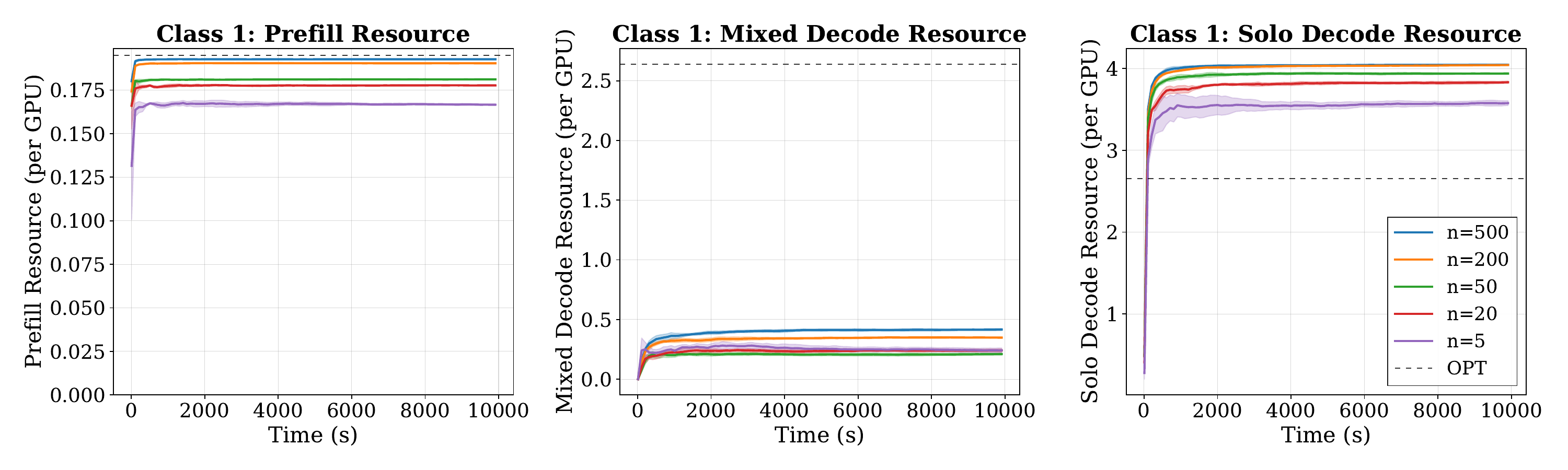}
    \end{minipage}

    \caption{Occupancy convergence under the gate-and-route policy. Top: Class 0 (decode-heavy). Bottom: Class 1 (prefill-heavy). Decode occupancy exhibits looser class-wise convergence.}
    \label{fig:occ_standard}
\end{figure}

\paragraph{SLI-Aware Convergence.}
To address the decode occupancy drift and enforce closer adherence to the fluid plan, we use the SLI-aware gate-and-route policy described in Section~\ref{subsec:sli-aware-policy}.
This policy uses routing probabilities \(p_{s,i}\) derived from the LP solution to probabilistically route jobs to mixed or solo pools.

\begin{figure}[htbp]
    \centering
    \begin{minipage}[b]{0.98\textwidth}
        \includegraphics[width=\linewidth, height=4cm]{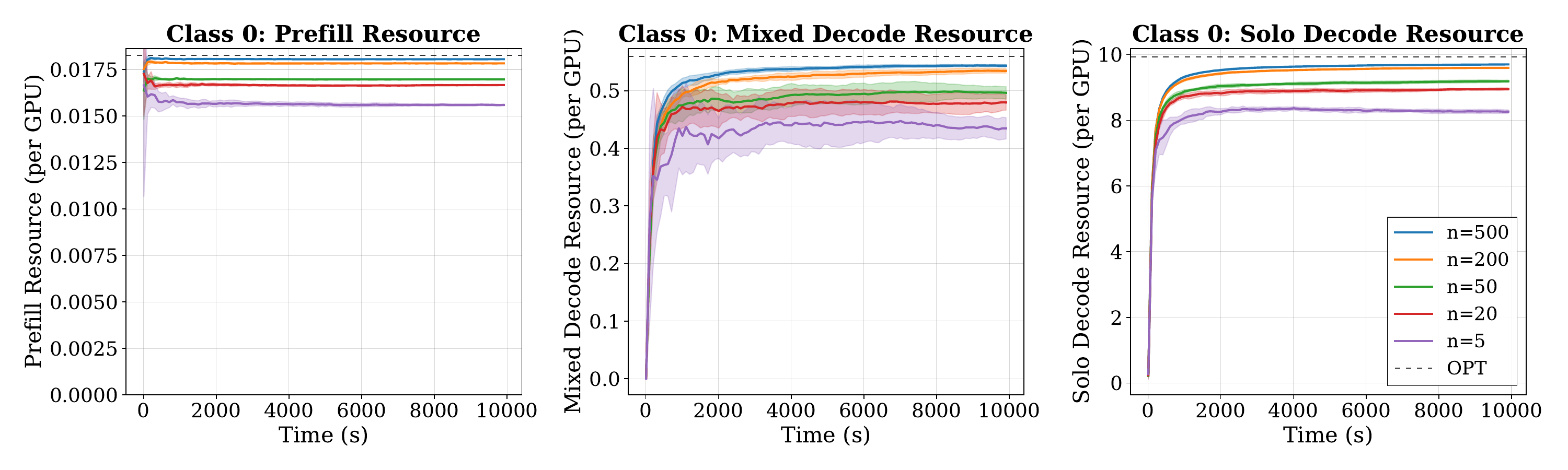}

        \vspace{0.3cm}

        \includegraphics[width=\linewidth, height=4cm]{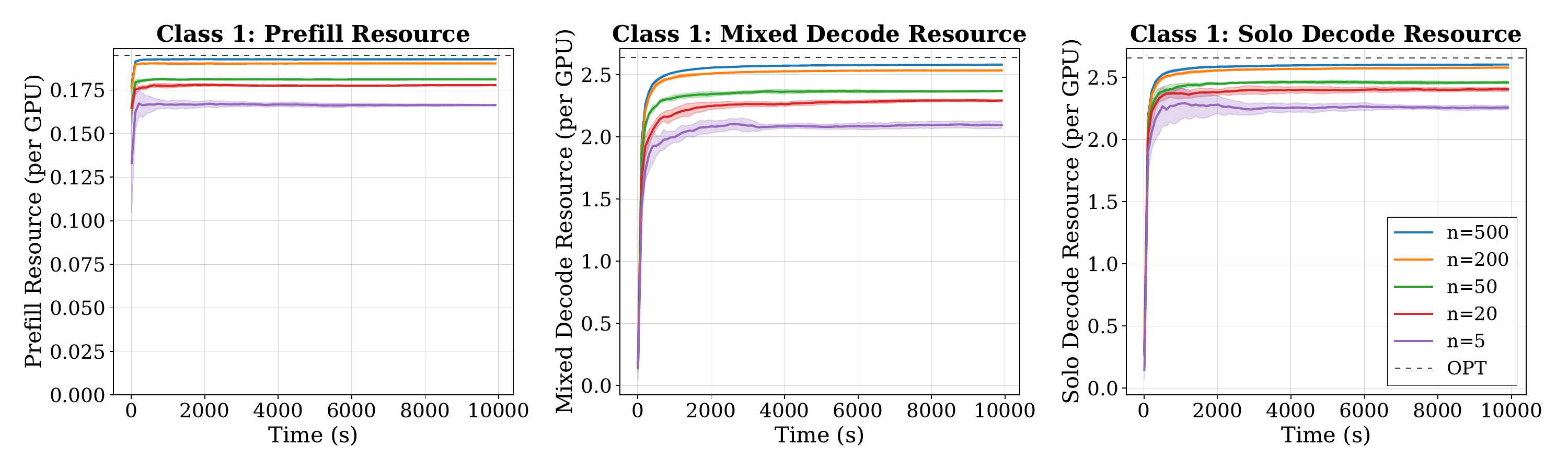}
    \end{minipage}

    \caption{Occupancy convergence under the SLI-aware gate-and-route policy. Top: Class 0 (decode-heavy). Bottom: Class 1 (prefill-heavy). Both prefill and decode occupancies converge tightly to the LP targets.}
    \label{fig:occ_slo}
    \vspace{-5mm}
\end{figure}

Figure~\ref{fig:occ_slo} shows that both prefill and decode occupancies converge to the optimization solution for both classes under the SLI-aware router.
By enforcing the specific routing split \(y_{m,i}^\star\) vs.\ \(y_{s,i}^\star\), the policy makes the stochastic system structurally mimic the fluid limit.
This validates Theorem~\ref{thm:slo_optimality}, confirming that precise control over resource-allocation distributions can be achieved at scale.
}

\subsection{\textcolor{tc-ins}{Additional Ablation Studies on Synthetic Workloads}}
\label{sec:ec-additional-ablation-studies}

To evaluate the contribution of each design component in the proposed framework, we compare the full occupancy-based Gate-and-Route policy with several component ablations across a variety of synthetic instances.
These instances encompass diverse infrastructure hyperparameters ($\alpha$ from $0.02$ to $0.15$, $\beta$ from $10^{-5}$ to $10^{-3}$, and $\gamma$ from $10$ to $50$), distinct target user compositions ($P_i,D_i$ from $200$ to $3000$), and different arrival rates $\lambda$ from $0.25$ to $0.5$, with $n=500$ GPUs throughout.
We report the normalized revenue rate (scaled from 0 to 1) along with standard deviations to summarize variability across the tested settings.

\subsubsection*{Candidate Policies}
We evaluate the following policies to isolate the contributions of occupancy-based admission control (Gate), the solo-first decode router (Route), and static planning.
Each label follows the pattern (prefill rule)(decode rule)-(planning rule), where the prefill rule is either occupancy-based Gate (G) or FCFS (F), the decode rule is either Greedy solo-first routing (G), Immediate same-slot decoding (I), or local FCFS within a GPU (F), and the planning suffix is Static Planning (SP) or Without Static Planning (WSP).
For brevity we write GG-SP for the full policy and use the same convention for the ablations.
In Figure~\ref{fig:synthetic-ablation-comparison}, the FI-WSP variant is annotated as ``(Sarathi)'' to indicate that its FCFS-prefill plus immediate-decode local rule is similar in spirit to a Sarathi-style heuristic; this synthetic-workload ablation is distinct from the trace-driven Sarathi-style baseline evaluated in Section~\ref{subsec:real-trace-online-eval}, which is implemented inside the calibrated finite-system simulator.

\begin{itemize}
    \item \textbf{Gate-and-Route with Static Planning (GG-SP):}
    This is the full occupancy-based Gate-and-Route policy.
    It uses the LP solution to fix the mixed/solo resource partition and to define class-level prefill occupancy targets.
    The prefill gate admits the most under-target class, while the decode router fills solo decode slots first and then mixed decode slots when needed.

    \item \textbf{FCFS-and-Immediate without Static Planning (FI-WSP):}
    This ablation removes all three components: there is no LP-based static partition, no occupancy-based gate, and no decoupled decode router.
    It admits prefills in FCFS order and immediately continues decode on the same GPU slot after prefill completion.

    \item \textbf{Gate-and-Immediate without Static Planning (GI-WSP):}
    This ablation keeps the occupancy-based prefill gate but removes the decoupled decode router.
    As in FI-WSP, decode begins immediately on the same GPU slot after prefill completion.
    Comparing GI-WSP with GG-SP isolates the value of decoupling prefill admission from downstream decode placement.

    \item \textbf{Gate-and-FCFS without Static Planning (GF-WSP):}
    This ablation keeps the occupancy-based prefill gate and the decoupled architecture, but replaces the solo-first decode router with an FCFS-style local rule.
    Whenever a GPU slot becomes available, the system prioritizes the oldest eligible task rather than explicitly preserving solo decode capacity.
    Comparing GF-WSP with GG-SP isolates the contribution of the solo-first routing rule.

    \item \textbf{FCFS-and-Route with Static Planning (FG-SP):}
    This ablation keeps the static mixed/solo partition and the solo-first decode router, but removes the occupancy-based prefill gate and admits prefills in FCFS order.
    Comparing FG-SP with GG-SP isolates the role of the gate in regulating the class mix and downstream decode load.
\end{itemize}

The comparative results, illustrated in Figure~\ref{fig:synthetic-ablation-comparison}, show how the components interact.
The full GG-SP policy performs best among the tested variants because the occupancy gate regulates the class mix entering prefill, the static planning step fixes an appropriate mixed/solo capacity split, and the solo-first router uses the faster decode-only slots effectively.
The ablated policies identify where performance is lost when one of these components is removed or coupled with another stage.
The relatively small standard deviations observed for GG-SP indicate stable performance across the tested synthetic instances.

\begin{figure}[ht]
    \centering
    \includegraphics[width=.8\textwidth]{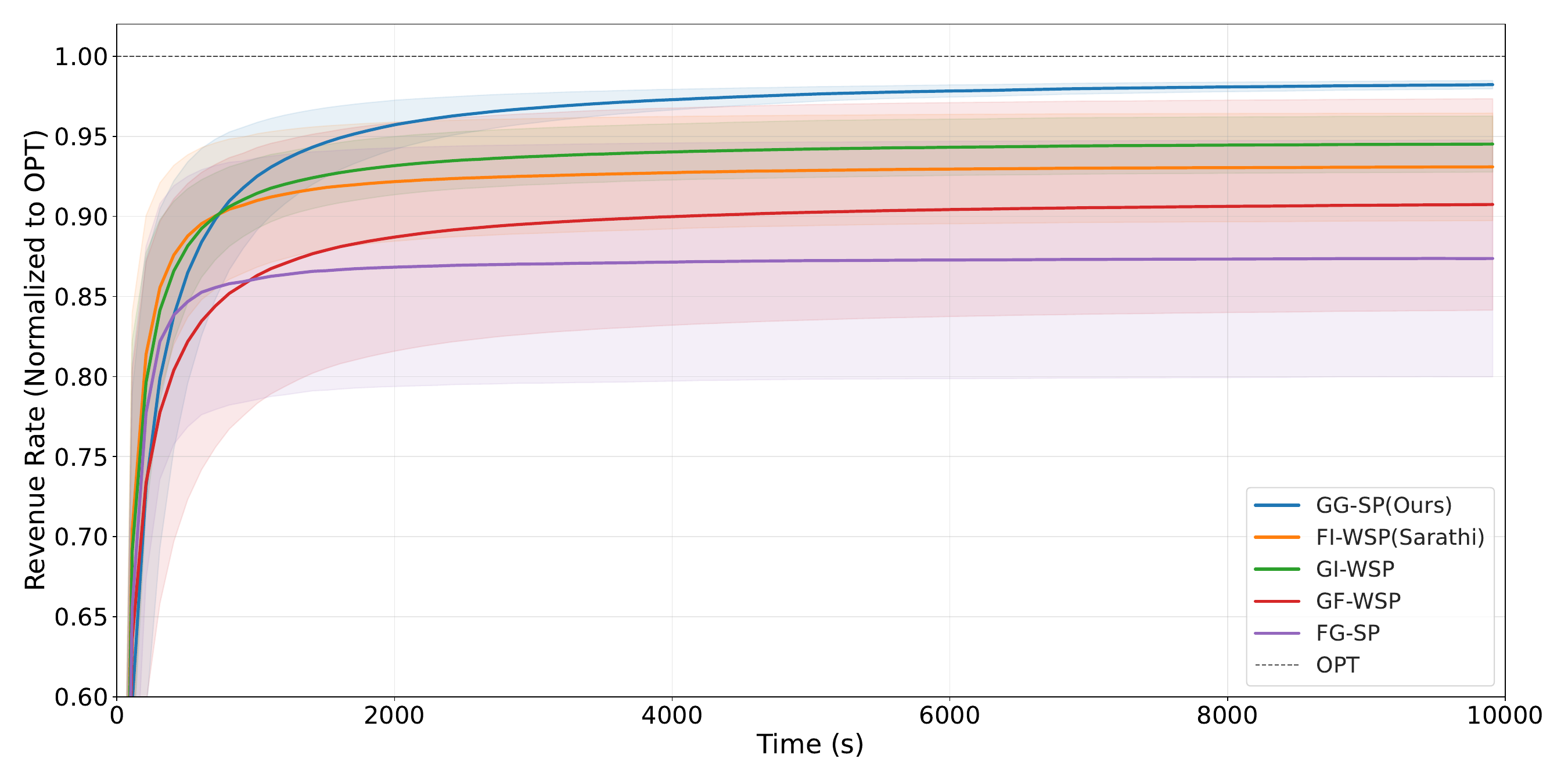}
    \caption{Comparison of normalized average revenue across different policies. The results aggregate performance over various infrastructure settings ($\alpha, \beta, \gamma$) and user class combinations. Error bars indicate the standard deviation.}
    \label{fig:synthetic-ablation-comparison}
    \vspace{-4mm}
\end{figure}